%% file: main.tex
\documentclass[11pt,a4paper]{article}

\usepackage[left=2.5cm, right=2.5cm, top=3cm, bottom=3cm]{geometry}
\usepackage[utf8]{inputenc} 
\usepackage{indentfirst}
\usepackage{amsmath,amssymb,amsthm}
\usepackage{graphicx,epsf}
\usepackage{caption}
\usepackage{subcaption}
\usepackage{tikz}
\usetikzlibrary{shapes,arrows,patterns,calc,shadings}
\usepackage{enumerate}
\usepackage{natbib}
\usepackage{url} 
\usepackage{xcolor}
\usepackage{framed}
\usepackage[colorlinks, citecolor=blue]{hyperref}
\usepackage{tabularx}
\usepackage{booktabs}
\usepackage{comment}
\usepackage{threeparttable}
\usepackage{enumitem}
\usepackage{environ}

\usepackage{lipsum} 
\usepackage[capposition=top]{floatrow}

\makeatletter
\newcommand{\neutralize}[1]{\expandafter\let\csname c@#1\endcsname\count@}
\makeatother

\theoremstyle{plain} 

\newtheorem{theorem}{Theorem}
\newtheorem{theoremalt}{Theorem}[theorem]
\newenvironment{theoremp}[1]{
  \renewcommand\thetheoremalt{#1}
  \theoremalt
}{\endtheoremalt}

\newtheorem{lemma}{Lemma}
\newtheorem{lemmaalt}{Lemma}[lemma]
\newenvironment{lemmap}[1]{
  \renewcommand\thelemmaalt{#1}
  \lemmaalt
}{\endlemmaalt}
\newtheorem{assumption}{Assumption}
\newtheorem{assumptionalt}{Assumption}[assumption]
\newenvironment{assumptionp}[1]{
  \renewcommand\theassumptionalt{#1}
  \assumptionalt
}{\endassumptionalt}
\newtheorem{corollary}{Corollary}

\newtheorem{proposition}{Proposition}
\newtheorem{propositionalt}{Proposition}[proposition]
\newenvironment{propositionp}[1]{
  \renewcommand\thepropositionalt{#1}
  \propositionalt
}{\endpropositionalt}

\theoremstyle{definition}

\newtheorem{remark}{Remark}
\newtheorem{example}{Example}

\DeclareMathOperator*{\argmin}{arg\,min}

\DeclareMathOperator{\diag}{diag}

\DeclareMathOperator{\Cov}{Cov}
\DeclareMathOperator{\Var}{Var}

\renewcommand{\qed}{\hfill{\tiny \ensuremath{\blacksquare} }}

\newcommand{\Ep}{{\mathrm{E}}}

\renewcommand{\Pr}{{\mathrm{P}}}

\newcommand{\bE}{\mathbb{E}}

\newcommand{\bG}{\mathbb{G}}
\newcommand{\sumiN}{\frac{1}{m} \sum_{j = 1}^m}
\newcommand{\sumtT}{\frac{1}{n} \sum_{i = 1}^n}
\newcommand{\sumgG}{\frac{1}{G} \sum_{g = 1}^G}

\NewEnviron{commentP}{} 

\newcommand\Poff{\RenewEnviron{commentP}{}}
\NewEnviron{commentF}{}
\newcommand\Foff{\RenewEnviron{commentF}{}}
\Poff
\Foff

\newcommand{\eps}{\varepsilon}
\newcommand{\indep}{\perp \!\!\! \perp}

\begin{document}

\baselineskip=1.3\normalbaselineskip

\title{Minimum Distance Estimation of Quantile Panel Data Models\footnote{We are grateful to Manuel Arellano, Bo Honoré, Aleksey Tetenov, and seminar participants at Bern University,
Geneva University, UC  Irvine, Fribourg University, 
the 27th International Panel Data Conference, the COMPIE 2022 Conference, the 2nd International Econometrics PhD Conference at Erasmus University Rotterdam, the 2023 Ski and Labor Seminar, the 2023 Young Swiss Economists Meeting, the 2023 IAAE Conference, the 2023 Summer Meeting of the Econometric Society, and the third Causal Inference Optimization-Conscious Econometrics Conference at the University of Chicago
for useful comments.}}
\author{Blaise Melly\thanks{Department of Economics, University of Bern, Schanzeneckstrasse 1, 3001 Bern, Switzerland.  blaise.melly@unibe.ch.}  and Martina Pons\thanks{Department of Economics, University of Bern, Schanzeneckstrasse 1, 3001 Bern, Switzerland. martina.pons@unibe.ch.}}
\date{First version: November 2020\\
This version: \today \\ 
\href{https://martinapons.github.io/files/MD.pdf}{Link to the newest version} }
\maketitle
\begin{abstract}

We propose a minimum distance estimation approach for quantile panel data models where unit effects may be correlated with covariates. This computationally efficient method involves two stages: first, computing quantile regression within each unit, then applying GMM to the first-stage fitted values. Our estimators apply to (i) classical panel data, tracking units over time, and (ii) grouped data, where individual-level data are available, but treatment varies at the group level. Depending on the exogeneity assumptions, this approach provides quantile analogs of classic panel data estimators, including fixed effects, random effects, between, and Hausman-Taylor estimators. In addition, our method offers improved precision for grouped (instrumental) quantile regression compared to existing estimators. We establish asymptotic properties as the number of units and observations per unit jointly diverge to infinity. Additionally, we introduce an inference procedure that automatically adapts to the potentially unknown convergence rate of the estimator. Monte Carlo simulations demonstrate that our estimator and inference procedure perform well in finite samples, even when the number of observations per unit is moderate. In an empirical application, we examine the impact of the food stamp program on birth weights. We find that the program's introduction increased birth weights predominantly at the lower end of the distribution, highlighting the ability of our method to capture heterogeneous effects across the outcome distribution.

\end{abstract}

\section{Introduction}\label{introduction}
Quantile regression, introduced by \cite{Koenker1978}, is a powerful tool for analyzing the effect of policies on the distribution of an outcome variable. The quantile treatment effect function provides more information than the average treatment effect, allowing, for instance, evaluation of the treatment's impact on inequality. When panel data are available, new identification and estimation strategies become feasible. Researchers can alleviate endogeneity concerns, for instance, by allowing for correlated group effects. They can obtain more precise estimates using a random-effects estimator or exploit individual-level variables to identify the impact of group-level variables, e.g., with the \cite{Hausman1981} estimator. 

We define panel data as a dataset structure where observations are organized along two dimensions. While classical panel data typically consist of individuals repeatedly measured over time, our results extend to group data, where individual-level observations are available, and treatments often vary at the group level. For example, \cite{Autor2013a} use commuting zones in a given decade as groups, while \cite{Angrist2004} employ schools. In both cases, treatments vary only between groups, but individual data are essential for estimating the conditional distribution of outcomes within each group. This paper adopts a general notation ($i$ and $j$ subscripts) applicable to both classical panel and group data. We primarily use terminology (individuals and groups) that is more common to group data, as our application falls into this category. 

As a first contribution of the paper, we propose a new class of minimum distance estimators for quantile panel data models. This class of estimators provides quantile analogs of established panel data methods, including fixed effects, random effects, between, and Hausman-Taylor estimators. 
Our estimation approach involves two stages. The first stage consists of group-level quantile regressions using individual-level covariates. In the second stage, we regress the fitted values from the first stage on individual-level and group-level variables. If these variables are potentially endogenous, an instrumental variable regression or, more generally, the generalized method of moments (GMM) estimator can be applied. This approach allows for the straightforward inclusion of external or internal instruments in the second stage. The proposed estimator is simple to implement, flexible, computationally fast, and applicable across various fields.\footnote{We provide general-purpose packages for both R and Stata.} While this two-step procedure may seem unconventional, we demonstrate in Section \ref{subsec:least squares estimators} that it is numerically equivalent to standard estimators when the least squares estimator is used in the first stage along with appropriate instruments.

As a nonlinear estimator, first-stage quantile regression is subject to a finite-sample bias that diminishes as the number of observations per group increases. Thus, our inference procedures are justified within an asymptotic framework where the number of observations per group $n$ and the number of groups $m$ diverge to infinity.\footnote{Large $n$ (often called large $T$) asymptotics have been widely applied in the nonlinear and dynamic panel data literature. For seminal contributions, see \cite{Phillips1999}, \cite{Hahn2002}, and \cite{ Alvarez2003}.} The asymptotic variance of the sample moments has two components: one arising from the first-stage quantile regression and another from the second-stage GMM regression. Since the regressors vary within groups in the first stage, the relevant number of observations is $mn$, making the variance proportional to $1/(mn)$. The second-stage estimation error, on the other hand, stems from the randomness of the group effects, with the corresponding variance proportional to $1/m$. However, this second component is zero if no group effects exist or the instrument exploits only within variation. Thus, the asymptotic distribution of the estimator is dominated by the component with the slower rate of convergence, which is the second stage error unless there is no group heterogeneity or the instruments are uncorrelated with group membership. As a result, the asymptotic distribution of the estimator is non-standard, as the convergence rate of a coefficient depends on the presence of group heterogeneity and the variation used to identify that coefficient.


In Section \ref{subsec:pointwise}, we consider three specific scenarios before suggesting adaptive estimation and inference procedures in Section \ref{subsec:adaptive}. In the first case, we assume the presence of group heterogeneity and demonstrate that only the coefficients identified through within-group variation can be estimated at the $\sqrt{mn}$ rate. For example, this applies to the fixed effects estimator. In contrast, coefficients associated with variables that are constant within groups rely on between-group variation for identification and are only estimable at the slower $\sqrt m$ rate. Only the second-stage error shows up in the first-order asymptotic distribution for these coefficients. Consequently, in this scenario, some coefficients converge faster than others. Case 2 assumes no group heterogeneity, eliminating second-stage variance and yielding a uniform $\sqrt{mn}$ convergence rate for all coefficients. Finally, we consider the intermediate case when group-level heterogeneity is present but vanishes precisely at the correct rate such that both components of the variance matter asymptotically. 

These asymptotic results provide valuable insights into the mechanics of our estimator, but applying them requires knowing whether group-level heterogeneity is present or not, leading to non-adaptive inference. As the second main contribution of the paper, Section \ref{subsec:adaptive} introduces adaptive estimation and inference procedures that address this issue. The proposed methodology is robust to different degrees of heterogeneity, allowing the variance of group effects to be zero, bounded, or diminishing at arbitrary rates. We first show that the leading variance term can be adaptively estimated using a traditional cluster-robust variance estimator. This result offers a practical procedure that is simple to implement and circumvents the need to estimate challenging components like the variance of first-stage coefficients. By inverting this estimated variance, we obtain a GMM estimator that is uniformly efficient in the unknown relative convergence rates of the moments. This result is non-standard because of the potentially different convergence rates of the moments; the efficient weighting matrix may be asymptotically singular. Finally, we suggest an overidentification test, which provides, for instance, the quantile equivalent of the Hausman test for the exogeneity of the between variation.

In the context of group data, the most closely related work is the IV quantile regression estimator proposed by \cite{Chetverikov2016}. They focus on the effect of variables that vary only between groups and implicitly assume the presence of group-level heterogeneity.\footnote{Their asymptotic distribution is degenerate in the absence of group effects, suggesting that the rate of convergence is faster in such cases.} While both their approach and ours share the same first-stage estimation, the second stage differs: we regress the fitted values on all variables, whereas they regress the estimated intercept on the group-level regressors. As a result, their estimator is not invariant to reparametrizations of the individual-level regressors. In Table \ref{tab:tab:grouped_bias}, simulations using the same data-generating process as \cite{Chetverikov2016} show that our minimum distance (MD) estimator exhibits substantially lower variance and mean squared error (MSE) across all sample sizes considered—reducing the MSE by a factor of up to 20. In Section \ref{sec:GIVQR}, we demonstrate and explain why our estimator is more precise than theirs. Furthermore, we contribute to this literature by providing efficient adaptive estimation and inference procedures that remain valid regardless of the degree of group-level heterogeneity, deriving the limiting distribution of the estimator for the coefficients on the individual-level variables, and relaxing the growth condition of $n$ relative to $m$.

Our class of estimators includes the MD estimators of \cite{Chamberlain1994} as a special case. We extend his framework by incorporating individual-level regressors and accommodating endogenous regressors and group effects but requiring the number of groups to approach infinity.\footnote{\cite{Chamberlain1994} uses different terminology because he focuses on cross-sectional regressions. He analyzes a quantile regression model with a finite number of combinations of regressor values, where the number of cells (groups in our terminology) is finite, and the regressors are constant within each cell.} Interestingly, in \cite{Chamberlain1994}, all the variance arises from the first-stage estimation, consistent with classical MD estimation, while in \cite{Chetverikov2016}, the variance originates entirely from the second stage. In our framework, the variance can stem from either the first or second stage, depending on the presence or absence of group effects, with the estimated standard errors capturing the relevant leading component.

Our paper also contributes to the literature on quantile panel data models.\footnote {The main text focuses on the large $n$ (often called large $T$) literature. In short panels, \cite{Chernozhukov2013a} derive bounds for quantile effects. \cite{Arellano2016} introduce a class of correlated random effects quantile regression estimators that are consistent in finite $n$. They apply this approach to study earnings and consumption dynamics in \cite{Arellano2017}.} \cite{Koenker2004} introduced a penalized quantile fixed effects estimator that treats individual heterogeneity as a pure location shift. \cite{Kato2012} extend this approach by allowing group effects to depend on the quantile of interest and by contributing to the asymptotic theory of the estimator. \cite{Galvao2015} propose a two-step minimum distance (MD) estimator as a computationally efficient method for estimating fixed effects quantile models. Our framework nests this estimator. However, their focus is solely on the effects of individual-level covariates without exploiting variation between individuals.\footnote{\cite{Galvao2015} consider a traditional panel data setting; thus, in their terminology, they focus on estimating the effects of time-varying regressors.} In contrast, we aim to estimate the effects of individual- and group-level regressors while allowing for internal and external instruments. \cite{Galvao2019} suggest using the usual pooled quantile regression estimator in the presence of random effects. Our random effects estimator differs in that it targets the conditional quantile function given the group effects (see Remark \ref{remark:conditional} for a discussion on conditional effects). In other words, we estimate a different parameter for which quantile regression is inconsistent, even when the random effects are uncorrelated with the covariates.

\cite{Chernozhukov2013} introduced distribution regression as an alternative to quantile regression for estimating the entire conditional distribution of outcomes given covariates. \cite{Fernandez-Val2022a} extend this approach by developing a dynamic distribution regression panel data model with heterogeneous coefficients across groups. In their framework, the first stage involves regressing the outcome on individual-level covariates using distribution regression, followed by a second stage where the coefficients are projected onto group-level instruments. We build on their proof strategy to demonstrate that our inference procedure remains uniformly valid with respect to the degree of heterogeneity. However, our approach differs in several key aspects: we use quantile regression instead of distribution regression, project the fitted values rather than the coefficients, consider both individual-level and group-level instruments, and optimally combine these instruments using GMM. Although our method is limited to continuous outcomes, we find that quantile regression coefficients are easier and more intuitive to interpret.

As a third main contribution, we show the practical relevance of our approach in an empirical application. 
Specifically, we extend the work of \cite{Almond2011} by estimating the distributional effect of the Food Stamp Program on birth weight. Following the enactment of the Food Stamp Act, the number of counties implementing the program increased substantially in the late 1960s and early 1970s. To apply our minimum distance estimator, we define groups as county-trimester cells. The subscript $j$ indexes a county-trimester cell, while the subscript $i$ defines an individual within this cell. We estimate the model separately for black and white mothers and find that the Food Stamp Program has a positive impact on the lower tail of the birth weight distribution, particularly among black mothers.

The remainder of the paper is structured as follows. Section \ref{sec:model} introduces the model and the estimator and shows that traditional least-squares estimators can be implemented as MD estimators. Section \ref{sec:asym. theory} develops the asymptotic theory. Section \ref{sec:GIVQR} extends the discussion to grouped data and compares our estimator with the grouped IV quantile regression approach of \cite{Chetverikov2016}. Section \ref{sec:FE+BE+RE} applies our framework to traditional panel data, proposing quantile analogs of the within, between, and random effects estimators and the Hausman test. Both Sections \ref{sec:GIVQR} and \ref{sec:FE+BE+RE} include Monte Carlo simulations to evaluate finite sample performance. In Section \ref{sec:food stamp}, we present the empirical application, and Section \ref{sec:conclusion} concludes.

\section{Model and Minimum Distance Estimator}\label{sec:model}

\subsection{Quantile Model}\label{subsec:quantile model}
We want to learn the effects of the individual-level variables $x_{1ij}$ and the group-level variables $x_{2j}$ on the distribution of an outcome $y_{ij}$. We observe these variables for the groups $j=1,\dots,m$ and individuals $i=1,\dots, n$.\footnote{We assume a balanced panel for notational simplicity. However, the results generalize to unbalanced datasets.} 
For some quantile index $0<\tau<1$, we assume that 
\begin{equation}\label{eq:model}
Q(\tau,y_{ij}|x_{1ij},  x_{2j},v_j) =   x'_{1ij} \beta(\tau) +  x'_{2j} \gamma (\tau) + \alpha(\tau,v_j),
\end{equation}
where $Q(\tau, y_{ij}|x_{1ij},  x_{2j},v_j)$ is the $\tau$th conditional quantile function of the response variable $y_{ij}$ for individual $i$ belonging to group $j$ given the $K_1$-vector of individual-level regressors ${x}_{1ij}$, the $K_2$-vector of group-level variables $x_{2j}$, and an unobserved random vector $v_j$ of unrestricted and unknown dimension. In total, there are $K_1 + K_2 = K$ parameters to estimate. The parameters $\beta(\tau)$, $\gamma(\tau)$ and the unobserved group heterogeneity $\alpha(\tau,v_j)$ can depend on the quantile index $\tau$. Depending on the setting, $\beta(\tau)$ or $\gamma(\tau)$ (or both) might be the parameters of interest. We normalize $\bE[\alpha(\tau,v_{j})]=0$, which is not restrictive because $x_{2j}$ includes a constant. 

\begin{remark}[\bfseries Conditional versus unconditional effects]\label{remark:conditional}
In contrast to the average effect, the definition of a quantile treatment effect depends on the conditioning variables. In this paper, we model the distribution of $y_{ij}$ conditionally on the covariates and the group effect $\alpha(\tau,v_j)$. Thus, even if the group effects are independent of the regressors, we identify different parameters than those identified by quantile regression as introduced by \cite{Koenker1978} or by instrumental variable quantile regression as introduced by \cite{Chernozhukov2005}. The following example illustrates the difference between these parameters. Consider an application where each group $j$ corresponds to a region and each unit $i$ to an individual within this region. We do not have any $x_{1ij}$ variable. We are interested in the effect of a binary treatment $x_{2j}$, which has been randomized and is, therefore, independent from $\alpha(\tau,v_{j})$. $\gamma(\tau)$ is the effect of this treatment for individuals that rank at the $\tau$ quantile of $y_{ij}$ in \textit{their} region. On the other hand, the quantile regression of $y_{ij}$ on $x_{2j}$ identifies the effect for individuals that rank at the $\tau$ quantile in the whole country (given the treatment status).
These are different parameters except if $\alpha(\tau,v_j) =0$ for all $j$ or if the treatment effect is homogeneous such that $\gamma(\tau)=\gamma$ for all $\tau$. Whether conditional or unconditional quantile treatment effects are of interest depends on the question. Conditional quantile treatment effects are particularly useful for studying within-group inequalities when groups might be regions or industries. For example, \cite{Autor2016a} and \cite{Engelhardt2021} study the effect of the minimum wage on within-state inequality, while \cite{Autor2021} study the effect of trade shock on wage inequality within local labor markets. 
If the unconditional effect is of interest, one can naturally obtain the unconditional distribution functions by integrating out the group effects (and possibly the other variables) and then inverting the resulting distribution functions to obtain the unconditional quantile functions, see \cite{Chernozhukov2013}.
\end{remark}

When model (\ref{eq:model}) holds, the $\tau$ quantile regression of ${y_{ij}}$ on $x_{1ij}$ and a constant using only observations for group $j$ identifies the slope $\beta(\tau)$ and the intercept $x_{2j}'\gamma(\tau)+\alpha(\tau,v_j)$. We need to consider variation across groups to identify the coefficient on the group-level variables. Note that model (\ref{eq:model}) implies
\begin{equation*}\bE\left[Q(\tau,y_{ij}|x_{1ij},  x_{2j},v_j)|x_{1ij},x_{2j}\right]=x'_{1ij} \beta(\tau) +  x'_{2j} \gamma (\tau) + \bE\left[\alpha(\tau,v_j)|x_{1ij},x_{2j}\right].\end{equation*}
If $\alpha(\tau,v_j)$ is exogenous with respect to $x_{1ij}$ and $x_{2j}$ and the linear model is correctly specified, $\bE[\alpha(\tau,v_j)|x_{1ij},x_{2j}]=0$ and a linear regression identifies the parameters of interest.\footnote{Uncorrelation between $\alpha(\tau,v_j)$ and $x_{1ij}$ and $x_{2j}$ is sufficient to identify the linear projection.} The last representation suggests a two-step estimation strategy: (i) group-level quantile regression of $y_{ij}$ on $x_{1ij}$, (ii) OLS regression of the fitted values from the first stage on $x_{1ij}$ and $x_{2j}$. 

When the group effects $\alpha(\tau,v_{j})$ are endogenous (possibly correlated with  $x_{1ij}$ and $ x_{2j}$), we assume that there is a $L$-dimensional vector ($L \geq K$) of valid instruments $z_{ij}$ satisfying
\begin{equation}\label{eq:uncorrelation}\bE[z_{ij} \alpha(\tau,v_{j})] =\bE\left[z_{ij}\left(Q(\tau,y_{ij}|x_{1ij},x_{2j},v_j)-x_{1ij}'\beta(\tau)-x_{2j}'\gamma(\tau)\right)\right] =0.\end{equation}Note that $\beta(\tau)$ is identified in model (\ref{eq:model}) as long as there is some variation in $x_{1ij}$ within some groups. For instance, we can include the demeaned regressors, $\dot x_{1ij}=x_{1ij}-\bar x_{1j}$ with $\bar x_{1j}=n^{-1}\sum_{i = 1}^nx_{1ij}$, in the vector of instruments $z_{ij}$ because this variable will satisfy condition (\ref{eq:uncorrelation}) under strict exogeneity.\footnote{In the special case of traditional panel data, the demeaned regressors correspond to the within transformation.} On the other hand, we need additional instruments to identify $\gamma(\tau)$. Equation (\ref{eq:uncorrelation}) suggests a similar estimation strategy as in the exogenous case but with the instrumental variable estimator (or more generally the GMM estimator) in the second stage: (i) group-level quantile regression of $y_{ij}$ on $x_{1ij}$, (ii) GMM regression of the fitted values from the first stage on $x_{1ij}$ and $x_{2j}$ using $z_{ij}$ as instrument.

\begin{remark}[\bfseries Skorohod representation]
The following Skorohod representation implies the model defined in equation (\ref{eq:model}):
\begin{align*}y_{ij}&=x_{1ij}\beta(u_{ij})+x_{2j}\gamma(u_{ij})+\alpha(u_{ij},v_{j})\\
&=q(x_{1ij},x_{2j},u_{ij},v_j),\end{align*}
where $q(x_{1ij},x_{2j},u_{ij},v_j)$ is strictly increasing in the third argument (while fixing the other arguments).\footnote{This is the same model as in \cite{Chetverikov2016}, where a similar Skorohod representation is derived in their footnote 7.} We normalize $u_{ij}|x_{1ij},x_{2j},v_j\sim U(0,1)$ such that $q(x_{1ij},x_{2j},\tau,v_{j})$ is the $\tau$ conditional quantile function. $u_{ij}$ ranks the individuals within a group and $v_j$ captures the group heterogeneity. In this model, a sufficient condition for equation (\ref{eq:uncorrelation}) is $(u_{ij},v_j)\indep z_{ij}$. If the instrument does not vary within groups, only $v_j\indep z_j$ is sufficient. \end{remark}

\begin{remark}[\bfseries Heterogeneous coefficients]\label{remark:heterogeneous}
Our model allows only the intercept to differ between groups. Now consider a more general model where we also allow the slopes to differ between groups:
\begin{equation}\label{eq:het_sko}y_{ij}=x_{1ij}'\beta(u_{ij},v_j)+x_{2j}'\gamma(u_{ij},v_j)+\alpha(u_{ij},v_j).\end{equation}
If we maintain the conditional strict monotonicity assumption with respect to $u_{ij}$, this model implies that
\begin{equation}\label{eq:heterogeneous}Q(\tau,y_{ij}|x_{1ij},  x_{2j},v_j) =   x'_{1ij} \beta(\tau,v_j) +  x'_{2j} \gamma (\tau,v_j) + \alpha(\tau,v_j).\end{equation}
In the exogenous case where $(x_{1ij},x_{2j})\indep v_j$, it follows that  \begin{align*}\bE\left[Q\left(\tau,y_{ij}|x_{1ij},  x_{2j},v_{j}\right)|x_{1ij},  x_{2j}\right] &=   x'_{1ij} \int\beta(\tau,v)dF_V(v) +  x'_{2j} \int\gamma (\tau,v)dF_V(v) + \int\alpha(\tau,v)dF_V(v)\\
&=x_{1ij}'\bar\beta(\tau)+x_{2j}'\bar\gamma(\tau)\end{align*}because we have normalized $\bE[\alpha(\tau,v_j)]=0$. This implies that the linear projection of $Q(\tau,y_{ij}|x_{1ij},  x_{2j},v_{j})$ on $x_{1ij}$ and $x_{2j}$ identifies the coefficients $\beta(\tau)$ and $\gamma(\tau)$ when the homogenous model (\ref{eq:model}) holds and the average effect over all groups at the $\tau$ quantile of their conditional distribution when the heterogenous model (\ref{eq:heterogeneous}) holds.\footnote{In the endogenous case, we obtain the instrumental variable projection instead of the standard linear projection. For instance, if $x_{2j}$ is an endogenous binary variable and $z_{ij}$ is a binary instrument, we identify the average treatment effects for the compliers at the $\tau$ quantile of their conditional distribution.} Naturally, it is also possible to model the heterogeneity between groups by estimating more flexible linear projections of $Q(\tau,y_{ij}|x_{1ij},  x_{2j},v_{j})$. For instance, we can interact $x_{1ij}$ 
with observable characteristics $x_{2j}$.\footnote{Starting with model (\ref{eq:het_sko}), one can simultaneously analyze both within-group and inter-group heterogeneity by constructing a quantile function with two quantile indices: one for the heterogeneity across groups and one for the heterogeneity within groups. These heterogeneous coefficients are identified through a two-step quantile regression: (i) a group-by-group quantile regression of \( y_{ij} \) on \( x_{1ij} \), followed by (ii) a quantile regression of the fitted values from the first stage on \( x_{1ij} \) and \( x_{2j} \). \cite{Pons2024} explores a version of this model that defines different parameters and falls outside the scope of this paper.}
\end{remark}

\subsection{Quantile Minimum Distance Estimators}

Motivated by the representation in equation (\ref{eq:uncorrelation}), we suggest the following the two-step procedure. In the first step, for each group $j$ and quantile $\tau$, we regress $y_{ij}$ on individual-level variables $x_{1ij}$ and a constant using quantile regression. The intercept of the first stage regression captures both the group effect $\alpha(\tau,v_j)$ and the term $x_{2j}' \gamma(\tau)$ as these vary only between groups. 
In a second step, we regress the fitted values of the first stage on $x_{1ij}$ and $ x_{2j}$, using GMM with instruments $z_{ij}$.

Formally, the first stage quantile regression solves the following minimization problem for each group and quantile separately: 
\begin{equation}\label{eq:first stage estimator}
\hat\beta_j(\tau)= \left(  \hat \beta_{0,j} , \hat \beta_{1,j}' \right)' = \argmin_{(b_0, b_1) \in \mathbb R^{K_1 + 1}} \frac{1}{n} \sum_{i = 1}^n\rho_\tau (y_{ij} - b_0 - x_{1ij}' b_1 ),
\end{equation}
where $\rho_\tau (x) = (\tau - 1 \{x < 0\})x$ for $x \in \mathbb{R}$ is the check function. The true vector of coefficients for group $j$ is given by $\beta_j(\tau) = (\alpha(\tau,v_j)+x_{2j}'\gamma(\tau), \beta(\tau)')'$. When the model does not contain any $x_{1ij}$ variables, quantile regression computes the sample quantiles in each group. 

\textbf{Notation}. Throughout the paper, we will use the following notation. Let $\tilde x_{ij} = (1, x_{1ij}') '$ and $x_{ij} = (x_{1ij}', x_{2j}')'$. For each group $j$ we define the following matrices. The $n \times  K_1$ matrix of individual-level regressors $ X_{1j} = (   x_{11j} ,  x_{12j}, \dots ,  x_{1nj}  )'$, the $n \times  K$ matrix containing all regressors $X_j=(x_{1j}, x_{2j},\dots, x_{nj})'$ and the $n \times  L$ matrix of instruments $Z_j=(z_{1j},z_{2j},\dots,z_{nj})'$.
Further, we define two matrices for all observations. The $mn\times K$ matrix of regressors for all groups $X =$  $(X_1', \dots , X_m')'$ and the $mn\times L$ matrix of instruments for all groups as $Z = (Z_1', \dots, Z'_m )'$. 
We let $Y$ be the response variable's $mn \times 1$ vector. 
The fitted value for individual $i$ in group $j$  at quantile $\tau$ is $\hat y_{ij}(\tau)=\hat\beta_{0,j}(\tau)+x_{1ij}'\hat\beta_{1,j}(\tau)$. We denote the $n \times  1$ column vector of fitted values for group $j$ by $\hat Y_j(\tau)=(\hat y_{1j}(\tau), \dots, y_{nj}(\tau))'$, and the $mn\times 1$ vector of fitted values by $\hat Y (\tau)= (\hat Y_1'(\tau), \dots, \hat Y_m'(\tau))'$.

\begin{remark}[\bfseries Alternative first-stage estimators]
The quantile regression estimator proposed by \cite{Koenker1978} is not necessarily efficient. \cite{Newey1990a} suggest a semiparametrically efficient weighted estimator of $\beta_j(\tau)$. However, we opt for the unweighted quantile regression estimator due to the challenges associated with estimating the weights and the complicated interpretation of the estimates in cases of misspecification. In our model (\ref{eq:model}), the variation within groups is assumed to be exogenous. If this assumption were violated, one could apply an instrumental variable (IV) quantile regression (see, e.g., \citealp{Chernozhukov2006}) in the first stage, followed by the second-stage GMM regression described below.\footnote{An IV extension of the MD estimator by \cite{Galvao2015} is suggested in \cite{Dai2021}.} We do not pursue this (computationally intensive) extension in this paper.\end{remark}

The second stage consists of the linear GMM regression of $\hat Y(\tau)$ on $X$ using $Z$ as an instrument.
The estimator has the following closed-form expression:
\begin{equation}\label{eq:second stage estimator}
\hat \delta (\hat W, \tau) = \left( X' Z \hat W(\tau) Z'  X \right) ^{-1}  X' Z \hat W(\tau) Z' \hat Y(\tau) ,
\end{equation}
where $\hat W(\tau)$ is a $L \times L$ symmetric weighting matrix.  When $L = K$, the second step estimator in equation (\ref{eq:second stage estimator}) simplifies to the IV estimator using $Z$ as an instrument, and we can drop the dependence on $\hat W(\tau)$. 

Our two-step estimator is extremely simple to implement; it requires only routines performing quantile regression and GMM estimation, which are already available in many software applications. Quantile regression, which is computationally more demanding due to the absence of a closed-form solution, is used only in the first stage, where there are fewer observations and a limited number of parameters to estimate. The first stage is also embarrassingly parallelizable, increasing the computational speed. For this reason, our estimator remains computationally attractive in large datasets with numerous groups. The second stage is a straightforward GMM estimator, which includes OLS and two-stage least squares as special cases. Traditional panel data methods can also be used in the second stage. For instance, in our application, we observe individuals born in a given trimester in a given county. The subscript $j$ defines a county-trimester cell, while the subscript $i$ defines an individual within this cell. In the second stage, we include trimester, county, and state $\times$ year fixed effects to estimate the effect of food stamps on the birth weight distribution.  

\begin{remark}[\bfseries Interpretation as a minimum distance estimator]\label{remark:mdestimator}
Our estimator can be written as an MD estimator, where the second stage imposes restrictions on the first-stage coefficients. For simplicity, in this remark, we consider the case where all the regressors are exogenous and $Z =  X$. Define
\begin{equation}\label{eq:restriction}
    \underset{\scriptscriptstyle (K_1+1) \times K}{R_j} = \begin{pmatrix}  0 & x_{2j}' \\ I_{K_1} & 0
    \end{pmatrix}
\end{equation}
such that $\tilde X_j R_j = X_j$. It follows that our MD estimator minimizes
\begin{align}\label{eq:MD representation}
\hat \delta(\tau) &= \argmin_{\delta} \sum_{j = 1}^m (\tilde X_j\hat\beta_j(\tau)-X_j\delta)'(\tilde X_j\hat\beta_j(\tau)-X_j\delta)\nonumber \\
&= \argmin_{\delta} \sum_{j = 1}^m (\tilde X_j\hat\beta_j(\tau)-\tilde X_j R_j \delta)'(\tilde X_j\hat\beta_j(\tau)-\tilde X_j R_j \delta)\nonumber \\
&= \argmin_{\delta} \sum_{j = 1}^m (\hat\beta_j(\tau)-R_j \delta)'\tilde X_j'\tilde X_j(\hat\beta_j(\tau)-R_j \delta),
\end{align}
which corresponds to the definition of a weighted minimum distance estimator that imposes the linear restrictions $\beta_j(\tau)=R_j \delta(\tau)$ with weights $\tilde X_j'\tilde X_j$.

Thus, our estimator is an MD estimator. However, it does not correspond to the textbook definition of a ``classical minimum distance'' estimator.\footnote{See section 14.6 in \cite{Wooldridge2010}.} In the classical MD setup, all the sampling variance arises in the first stage: if we know the first stage coefficients, we know the final coefficients. It follows that the efficient weighting matrix $\tilde W(\tau)$ is the inverse of the first-stage variance. In our case, the second stage also contributes to the variance due to the presence of the group effects $\alpha(\tau,v_j)$. Even if we know $\beta_j(\tau)$ (for a finite number of groups), we cannot exactly pinpoint $\gamma(\tau)$. The group effects play a role similar to misspecification in the classical MD, but with our estimator, the resulting bias disappears asymptotically as the number of groups increases. This is the second important difference: the dimension of our first stage estimates increases with the sample size while it is fixed for classical MD estimators.

\end{remark} 

\subsection{Least Squares Minimum Distance Estimators}\label{subsec:least squares estimators}

This paper proposes a two-step estimator in which the first stage involves performing quantile regressions within each group. Although this approach might seem unusual and specific to quantile models, we demonstrate in this subsection that this method yields numerically identical results to traditional least squares panel estimators, provided that OLS is applied in the first stage. A more detailed discussion, including formal statements, is provided in Appendix \ref{sec:formal ls}, with proofs available in Appendix \ref{app:proofs linear models}.


Consider first a model with group effects and individual-level regressors
\begin{equation}\label{FELS}
y_{ij}=x_{1ij}'\beta+\alpha_{j}+\varepsilon_{ij}.
\end{equation}
Typically, when estimating this model with fixed effects, most researchers apply the within transformation, which is widely recognized as being equivalent to a dummy variable regression.\footnote{In the context of traditional panel data, we would refer to the $j$ units as ``individuals'' and the $i$ units as ``time periods''. Thus, this transformation corresponds to the time-demeaning applied in the traditional panel data literature, eliminating time-invariant individual effects.}
Yet, a third equivalent method exists for computing the least squares fixed effects estimator, which involves exploiting the exogenous within-group variation using instrumental variables. Setting $\dot x_{1ij}$ as an instrument for $x_{1ij}$ in an instrumental variable regression is numerically identical to the least squares fixed effects.

Corollary \ref{cor:FE} in Appendix \ref{sec:formal ls} presents a fourth way to compute the least squares fixed effects estimator, which aligns with our paper's approach. This minimum distance method involves two steps: first, regressing with OLS the dependent variable on $x_{1ij}$ within each group, then using IV to regress the first-stage fitted values on $x_{1ij}$, with $\dot x_{1ij}$ as the instrument. This approach offers the most computationally efficient alternative for quantile estimation, as it divides the problem into two convex optimization steps rather than a single nonconvex IV quantile regression, and it avoids the challenges of high-dimensional quantile regression.


The two-step procedure is not specific to fixed effects but applies to a wide range of estimators. Proposition \ref{prop:md = gmm} in Appendix \ref{app:linear models} shows that the MD least squares estimator is algebraically identical to the one-step GMM regression of $y_{ij}$ on $x_{ij}$ under the condition that for each group $j$, the matrix of instruments lies in the column space of the matrix of first stage regressors.\footnote{The intuition is as follows. The fitted values of the first-stage least squares regression can be written as $P_{X_j} Y_j$ where $P_{X_j}$ is the first-stage least squares projection matrix of group $j$. If the instrument matrix, $Z_j$, is in the column space of $\tilde X_j$, it follows that $P_{X_j} Z_j = Z_j$. Therefore, $Z' \hat Y = Z'Y$ and the two GMM regressions are numerically identical.} For example, $\dot x_{1ij}$, $\bar x_{1j}$, and $x_{2j}$ satisfy the condition.

We extend the model by incorporating group-level regressors, $x_{2j}$:
\begin{equation}\label{FELS2}
y_{ij}=x_{1ij}'\beta+x_{2j}'\gamma+\alpha_{j}+\varepsilon_{ij}.
\end{equation}
By selecting different instrumental variables for the second-step GMM regression, we can numerically obtain the most common least squares panel data estimators. 
For example, using the group-averaged variables, $\bar{x}_{1j}$ and $x_{2j}$, as instruments yields the between estimator. Instrumental variable approaches are also available for random effects estimation. Although FGLS is the most common estimator for the random effects model, \cite{Im1999} show that the overidentified 3SLS estimator, with instruments $\dot x_{1ij}$, $\bar x_{1j}$, and $x_{2j}$, is numerically identical to the random effects estimator. Since 3SLS is a special case of a GMM estimator, using the first-stage fitted values as dependent variables does not change the estimates. Alternatively, the random effects estimator can be implemented using the theory of optimal instrument with a just identified 2SLS regression (see \citealp{Im1999, Hansen2021}). Additionally, the \cite{Hausman1981} estimator can be implemented by selecting the following instruments: $\dot x_{1ij}$ and the group average of the exogenous regressors. External instruments might also be included. 


\section{Asymptotic Theory}\label{sec:asym. theory}

\subsection{Preliminaries: Assumptions, Consistency, and Sample Moments}

In this section, we state the assumptions and present the asymptotic results. All the proofs are included in Appendix \ref{app: proofs asy. results}. For simplicity of notation, in the following, we write $\alpha_j(\tau)$ instead of $\alpha(\tau, v_j)$. We prove weak uniform consistency and weak convergence of the whole quantile regression process for $\tau\in\mathcal{T}$, where $\mathcal{T}$ is a compact set included in $(0,1)$. The symbol $\ell^\infty ( \mathcal{T} )$ denotes the set of component-wise bounded vector valued function of $\mathcal{T}$, $\rightsquigarrow$ denotes weak convergence, and for a random variable $h_{ij}$, $\bE_{i|j}[h_{ij}]$ indicates the expectation over $i$ in group $j$.

We start by writing the sampling error of $\hat\delta(\hat W, \tau)$ as a sum of a component arising from the first stage estimation error of $\beta_j(\tau)$ and a component arising from the second stage noise $\alpha_j(\tau)$:
\begin{lemma}[\textbf{Sampling error}]\label{l:sampling error}
Assume that the model in equation (\ref{eq:model}) holds, then
$$\hat\delta(\hat W,\tau)-\delta(\tau)=\hat G(\tau) \frac{1}{mn} \sum_{j = 1}^m\sum_{i = 1}^n z_{ij} \left(\tilde x_{ij}'(\hat \beta_j(\tau)-\beta_j(\tau))+\alpha_j(\tau)\right),$$
where $ \hat G(\tau) = \left(S_{ZX}'\hat W(\tau)S_{ZX}\right)^{-1}S_{ZX}'\hat W(\tau)$ and  $S_{ZX}=\frac{1}{mn}\sum_{j=1}^m\sum_{i = 1}^n z_{ij}x_{ij}'$.
\end{lemma}
To keep the notation light, we suppress the dependency of $\hat G(\tau)$ on $\hat W(\tau)$.
We now state assumptions that ensure that both components are well-behaved. For the analysis of the first stage estimator, we rely on results derived in \cite{Galvao2020} and make the assumptions required in their Theorem 2:

\begin{assumption}[\textbf{Sampling}] \label{a:sampling} (i) The processes $\{(y_{ij}, x_{ij}, z_{ij}):i = 1, \dots, n\}$ are independent across $j$. (ii) For each $j$, the observations $(y_{ij}, x_{ij},z_{ij})_{i=1,\dots, n}$ are i.i.d. across $i$.
\end{assumption}

\begin{assumption}[\textbf{Covariates}]  \label{a:covariates} (i) For all $j=1,\dots,m$ and all $i=1,\dots, n$, $\lVert x_{ij}\rVert \leq C$ almost surely. (ii) The eigenvalues of $\bE_{i|j}[ \tilde x_{ij} \tilde x_{ij}']$ are bounded away from zero and infinity uniformly across $j$.  
\end{assumption}

\begin{assumption}[\textbf{Conditional distribution}]  \label{a:conditional distribution}
The conditional distribution $F_{y_{ij} | x_{1ij}, v_j} (y|x,v)$ is twice differentiable w.r.t. y, with the corresponding derivatives $f_{y_{ij} | x_{1ij}, v_j} (y|x,v)$ and $f'_{y_{ij} | x_{1ij}, v_j} (y|x,v)$. Further, assume that
\begin{align*}
f_{max} = \sup_j \sup_{y \in \mathbb{R}, x \in \mathcal{X}} |f_{y_{ij} | x_{1ij}, v_j} (y|x, v)| < \infty
\end{align*}
and
\begin{align*}
\bar f' = \sup_j \sup_{y \in \mathbb{R}, x \in \mathcal{X}} |f'_{y_{ij} | x_{1ij}, v_j} (y|x,v)| < \infty.
\end{align*}
where $\mathcal{X}$ is the support of $x_{1ij}$
\end{assumption}

\begin{assumption}[\textbf{Bounded density}] \label{a:bounded density}
There exists a constant $f_{min} < f_{max}$ such that 
\begin{align*}
0 < f_{min} \leq \inf_j \inf_{\tau \in \mathcal{T}} \inf_{x \in \mathcal{X}} f_{y_{ij} | x_{1ij}, v_j} (Q(\tau,y_{ij}|x, v)|x,v).
\end{align*}
\end{assumption}
These are quite standard assumptions in the quantile regression literature. In Assumption \ref{a:sampling}, we assume that the processes are independent across $j$; this assumption can also be relaxed by allowing for clustering between groups. We also assume that the observations are i.i.d. within groups, but this can be relaxed at the cost of a more complex notation by applying Theorem 4 in \cite{Galvao2020}, which requires only stationarity and $\beta$-mixing. The estimator of the asymptotic variance that we suggest below is consistent in both cases. Assumption \ref{a:covariates} requires that the regressors are bounded and that $\bE_{i|j}[\tilde x_{ij}\tilde x_{ij}']$ is invertible. Assumptions \ref{a:conditional distribution} and \ref{a:bounded density} impose smoothness and boundedness of the conditional distribution, the density, and its derivatives.

For the second stage GMM regression we impose the following assumptions:
\begin{assumption}[\textbf{Instruments}]\label{a:instruments}
(i) For all $j=1,\dots,m$ and all $i=1,\dots, n$, $|| z_{ij} || \leq C$ a.s. (ii) For all $j=1,\dots,m$ and all $i=1,\dots, n$,  $\bE[z_{ij}\alpha_j(\tau)]=0$. (iii) For all $j=1,\dots,m$ and all $i=1,\dots, n$,  $y_{ij}$ is independent of $z_{ij}$ conditional on $(x_{ij},v_j)$. (iv) As $m\rightarrow \infty$, $m^{-1}\sum_{j = 1}^m \bE_{i|j}[z_{ij}x_{ij}']\rightarrow\Sigma_{ZX}$ where the singular values of $\Sigma_{ZX}$ are bounded from below and from above.
\end{assumption}
\begin{assumption}[\textbf{Group effects}]\label{a:group effects}
\mbox{}
\newline
(i) For all $j=1,\dots,m$, $\bE\left[\sup_{\tau \in \mathcal{T}}|\alpha_j(\tau)|^{4+\varepsilon_C}\right]$ $\leq C$ for $\varepsilon_C>0$. 
(ii) For some (matrix-valued) function $\Omega_2 : \mathcal{T} \times \mathcal{T} \rightarrow  \mathbb{R}^{L \times L}$, $m^{-1}\sum_{j = 1}^m \bE_{i|j}[\alpha_j(\tau_1) \alpha_j(\tau_2) z_{ij}z_{ij}']\underset{p}{\rightarrow} \Omega_{2}(\tau_1, \tau_2)$ uniformly over $\tau_1, \tau_2 \in \mathcal{T}$. 
(iii) For all $\tau_1, \tau_2 \in \mathcal{T}$, $|\alpha_j(\tau_2) - \alpha_j(\tau_1)| \leq C | \tau_2 - \tau_1 |$.
\end{assumption}

These assumptions are the same as in \cite{Chetverikov2016}. For the instrumental variables, we assume that (i) they are bounded, (ii) they are not correlated with the group effect (exclusion restriction), (iii) they do not affect the first stage estimation (this is often satisfied by construction, e.g. when the instruments do not vary within individuals or are a linear transformation of the first stage regressors), and (iv) they satisfy the relevance conditions. For the group effects, we assume that they have a finite fourth moment, and the average variance of $z_{ij}\alpha_j(\tau)$ converges to a well-defined matrix. 

Since the unobserved heterogeneity $\alpha_j(\tau)$ is group-specific, we require that the number of groups $m$ diverges to infinity. The first stage quantile regression estimator is a nonlinear estimator that is potentially biased in finite samples. Hence, the number of observations per group, $n$, must also diverge to infinity for consistency. \cite{Galvao2020} show that the bias is approximately of order $1/n$. For unbiased asymptotic normality, we need the bias to shrink faster than the standard deviation of the estimator. We will see that some elements of $\hat \delta(\hat W, \tau)$ converge at the $\sqrt m$ rate such that we need that $n$ goes to infinity more quickly than $\sqrt{m}$. On the other hand, other elements converge at the $\sqrt{mn}$ rate so that $n$ must go to infinity more quickly than $m$. We state these three different relative growth rates in the following assumption:
\begin{assumption}[\textbf{Growth rates}] \label{a:growth condition}
As $m\rightarrow \infty$, we have
\mbox{}
\begin{enumerate}[label=(\alph*)]
    \item $\frac{\log m}{n}\rightarrow 0$,
    \item $ \frac{ \sqrt{m} \log  n }{n}\rightarrow 0$,
    \item $ \frac{ m \left(\log  n\right)^2 }{n}\rightarrow 0$.
\end{enumerate}
\end{assumption}

Finally, we assume that the estimated weighting matrix uniformly converges to a strictly positive definite matrix that is continuous in the quantile index.
\begin{assumption}[\textbf{Full-rank weighting matrix}] \label{a:weighting}
Uniformly in $\tau\in\mathcal T$, $\hat W(\tau)\underset{p}{\rightarrow} W(\tau)$ where $W(\tau)$ is strictly positive definite and, for all $\tau_1, \tau_2 \in \mathcal{T}$, $|| W (\tau_2) - W(\tau_1) || \leq C | \tau_2- \tau_1 |$.
\end{assumption}

Our first result establishes the uniform consistency of our estimator under the weakest growth rate condition:
\begin{theorem}[\textbf{Uniform consistency}]\label{t:consistency}
Let the model in equation (\ref{eq:model}), Assumptions \ref{a:sampling}-\ref{a:group effects}, \ref{a:growth condition}(a), and \ref{a:weighting} hold. Then,
$$\underset{\tau\in\mathcal T}{\sup}\lVert \hat \delta(\tau)-\delta(\tau)\rVert = o_p(1).$$
\end{theorem} 

We now study the asymptotic distribution of our estimator. In Lemma \ref{l:sampling error}, we see that the sample moment condition is the sum of two terms. It is useful to consider them separately:
\begin{align}
\bar g^{(1)}_{mn}(\hat \delta, \tau)&= \frac{1}{mn}\sum_{j = 1}^m\sum_{i = 1}^n z_{ij} \tilde x_{ij}'\left(\hat \beta_j(\tau)-\beta_j(\tau)\right)\\
\bar g^{(2)}_{mn}(\hat \delta, \tau)&= \frac{1}{mn}\sum_{j = 1}^m\sum_{i = 1}^n z_{ij} \alpha_j(\tau)\end{align}
such that total moment condition is the sum of both components: $\bar g_{mn}(\hat\delta,\tau)=\bar g^{(1)}_{mn}(\hat\delta,\tau)+\bar g^{(2)}_{mn}(\hat \delta,\tau)=\frac{1}{mn}\sum_{j = 1}^m\sum_{i = 1}^ng_{ij}(\hat \delta,\tau)$. Lemma \ref{l:asymptotic moments} establishes joint asymptotic normality for the entire moment condition processes.
\begin{lemma}[\textbf{Asymptotic distribution of the sample moments}]\label{l:asymptotic moments}
Let the model in equation (\ref{eq:model}), and Assumptions \ref{a:sampling}-\ref{a:group effects} hold.
\begin{enumerate}[label=(\roman*), wide, labelindent=0pt]
\item  Under Assumption \ref{a:growth condition}(c), as $m\rightarrow \infty$,
\begin{equation}\label{eq:asymptotic moment 1}
    \sqrt{mn}\bar g^{(1)}_{mn}(\hat\delta, \cdot) \rightsquigarrow \mathbb Z_1(\cdot)\text{, in $\ell^\infty(\mathcal T)$,}
\end{equation}
where $\mathbb Z_1(\cdot)$ is a mean-zero Gaussian process with uniformly continuous sample paths and covariance function $\Omega_{1}(\tau,\tau')=\bE\left[\Sigma_{ZXj}V_j(\tau,\tau')\Sigma_{ZXj}'\right]$ with $\Sigma_{ZXj}=\bE_{i|j}[z_{ij} \tilde x_{ij}']$ and $V_j(\tau,\tau')$ is the asymptotic variance-covariance matrix of $\hat\beta_j(\tau)$ and $\hat\beta_j(\tau')$:
\begin{align*}V_j(\tau,\tau') = \bE_{i|j}[f_{y|x}(Q_{y|x, \nu_j}(\tau|\tilde x_{ij} )|\tilde x_{ij})\tilde x_{ij} \tilde x_{ij}']^{-1} (\min(\tau,\tau')-\tau\tau') \bE_{i|j}[\tilde x_{ij}\tilde x_{ij}']\\ \times \bE_{i|j}[f_{y|x}(Q_{y|x, \nu_j}(\tau'|\tilde x_{ij} )| \tilde x_{ij})\tilde x_{ij} \tilde x_{ij}']^{-1}\end{align*}

\item Under Assumption \ref{a:growth condition}(b), As $m\rightarrow \infty$,
\begin{equation}\label{eq:asymptotic moment 1b}
    \sqrt{m}\bar g^{(2)}_{mn}(\hat\delta, \cdot) \rightsquigarrow \mathbb Z_2(\cdot)\text{, in $\ell^\infty(\mathcal T)$,}
\end{equation}
where $\mathbb Z_2(\cdot)$ is a mean-zero Gaussian process with uniformly continuous sample paths and covariance function $\Omega_2(\tau,\tau')$, which is defined in Assumption \ref{a:group effects}(ii). 
\item Under Assumption \ref{a:growth condition}(c), as $m\rightarrow \infty$, $\underset{\tau,\tau'\in\mathcal T}{\sup} \left\lVert \Cov \left (\bar g^{(1)}_{mn}(\hat\delta, \tau), \bar g^{(2)}_{mn}(\hat\delta, \tau')\right )\right\rVert=o_p\left(\frac{1}{\sqrt{mn}}\right)$.
\end{enumerate}
\end{lemma}

$\bar g^{(1)}_{mn}(\hat\delta, \cdot)$ reflects the estimation error that arises in the first-stage quantile regression estimation. Since the first-stage regressors vary within groups, the relevant number of observations is $mn$, and correspondingly, the variance is proportional to $1/(mn)$. On the other hand, since the expected bias of the first-stage quantile regression is of order $1/n$, for asymptotic unbiasedness, we must require that $n$ goes to infinity slightly faster than $m$. In the proof, we build on results derived in \cite{Volgushev2019} and in \cite{Galvao2020}. $\bar g^{(2)}_{mn}(\hat\delta, \cdot)$ reflects the estimation error due to the randomness in $\alpha_j(\tau)$. This moment can also be interpreted as the moment that would be relevant if we knew $\beta_j(\tau)$. Since $\alpha_j(\tau)$ varies only between groups, the relevant number of observations here is $m$ and, accordingly, the variance of this moment converges at the slower rate of $1/m$. For asymptotic unbiasedness, we need only the weaker condition \ref{a:growth condition}(b), which requires that $n$ goes to infinity slightly faster than $\sqrt m$.

\subsection{Asymptotic Distribution when the Degree of Heterogeneity is Known}\label{subsec:pointwise}

The sample moment condition $\bar g_{mn}(\hat\delta,\tau)$ is thus the sum of two components that converge to zero at different rates. The asymptotic distribution of the estimator is dominated by the component with the slower rate of convergence, $\bar g_{mn}^{(2)}(\hat\delta,\tau)$, except its variance is zero, which is the case if $\Var(\alpha_j(\tau))=0$ or if $\bar z_j = \frac{1}{n}\sum_{i = 1}^n z_{ij} = 0$ for $j=1,\dots,n$. Since the degree of group-level heterogeneity affects this variance, it is useful to consider three cases: strong, no, and weak heterogeneity. In this subsection, we derive the asymptotic distribution of our estimator when the degree of heterogeneity is known. We suggest adaptive estimation and inference procedures in the following subsection.

\textbf{Case 1: Strong group-level heterogeneity.} We start with the case of strong heterogeneity that we define to be $\Var(\alpha_j(\tau))>\varepsilon>0$ uniformly in $\tau$.
The variance of an element of the vector $\bar g^{(2)}_{mn}(\hat\delta, \tau)=\frac{1}{mn}\sum_{j=1}^m \bar z_j\alpha_j(\tau)$ equals to zero when the corresponding instrument satisfies $\bar z_{j}=0$ for all $j$. For this reason, we distinguish between two sorts of instruments: $L_{1}$ instruments in $z_{1ij}$ satisfy $\bar z_{1j}=n^{-1}\sum_{i=1}^nz_{1ij}=0$ for all $j$, while $L_{2}$ instruments in $z_{2ij}$ satisfy $\bar z_{2j}\neq 0$ at least for some groups $j$.\footnote{Note that all the instruments that vary only within groups can be normalized to have mean zero. For instance, we can identify the effect of the individual-level variable $x_{1ij}$ by using the instrument $\dot x_{1ij}$, which has a zero mean in all groups.} We order the instruments such that $z_{ij}=(z_{1ij}',z_{2ij}')'$. It follows that
\begin{equation*}
    \bar g_{mn}(\hat\delta,\tau)=\begin{pmatrix}
        \bar g_{mn,1}(\hat\delta,\tau)\\
        \bar g_{mn,2}(\hat\delta,\tau)
    \end{pmatrix}=\begin{pmatrix}
        \bar g_{mn,1}^{(1)}(\hat\delta,\tau)\\
        \bar g_{mn,2}^{(1)}(\hat\delta,\tau)+
        \bar g_{mn,2}^{(2)}(\hat\delta,\tau)
    \end{pmatrix}
\end{equation*}
where $g_{mn,1}(\hat\delta,\tau)$ is a $L_1\times 1$ vector and $g_{mn,2}(\hat\delta,\tau)$ is a $L_2\times 1$ vector. Thus, in the case of strong heterogeneity, some moments converge at the fast rate $\sqrt{mn}$ while others converge at the slow rate $\sqrt{m}$.  


We order the coefficients in $\delta(\tau)$ (and the corresponding regressors in $x_{ij}$) such that the first $M_1$ elements are identified using only the $L_1$ fast moments while the remaining $M_2$ elements require the $L_2$ slow moments for identification. We denote by $\delta_1(\tau)$ the former and by $\delta_2(\tau)$ the latter coefficients. Formally, we partition $\Sigma_{ZX}$ such that it is block lower triangular:
\begin{align}\label{eq:lower diagonal}
\Sigma_{ZX}=\begin{pmatrix}\Sigma_{11} & \Sigma_{12} \\
\Sigma_{21} & \Sigma_{22}
\end{pmatrix}=\begin{pmatrix}\Sigma_{11} & 0 \\
\Sigma_{21} & \Sigma_{22}
\end{pmatrix}\end{align}
where $\Sigma_{11}$ is a full column rank $L_1\times M_1$ matrix, $\Sigma_{12}$ is $L_1\times M_2$, $\Sigma_{21}$ is $L_2\times M_1$ and $\Sigma_{22}$ is $L_2\times M_2$. Assumption \ref{a:instruments}(iv) implies that $\Sigma_{22}$ also has full column rank.
Note that $L_1$ and $M_1$ can be equal to zero such that equation (\ref{eq:lower diagonal}) is a definition and not an assumption. Accordingly, the adaptive estimation and testing procedures suggested in Section \ref{subsec:adaptive} do not require the user to classify the instruments or the regressors.

In the exactly identified case, our estimator simplifies to the instrumental variable estimator such that 
\begin{equation*}
    \hat G(\tau) = S_{ZX}^{-1}\underset{p}{\rightarrow}\Sigma_{ZX}^{-1}=\begin{pmatrix}
                \Sigma_{11}^{-1} & 0\\
        -\Sigma_{22}^{-1}\Sigma_{21}\Sigma_{11}^{-1} & \Sigma_{22}^{-1}
    \end{pmatrix}
\end{equation*}
and the first-order asymptotic distributions can be written as
\begin{align*}
        \sqrt{mn}(\hat\delta_1(\tau)-\delta_1(\tau))\underset{d}{\rightarrow}\Sigma_{11}^{-1}\sqrt{mn}\bar g_{mn,1}^{(1)}(\hat\delta,\tau)\\
        \sqrt{m}(\hat\delta_2(\tau)-\delta_2(\tau))
   \underset{d}{\rightarrow}\Sigma_{22}^{-1}\sqrt{m}\bar g_{mn,2}^{(2)}(\hat\delta,\tau)
\end{align*}
$\hat\delta_1(\tau)$ is only a function of $\bar g_{mn}^{(1)}$ and converges thus at the $\sqrt{mn}$ rate. $\hat\delta_2(\tau)$ depends on both $\bar g_{mn}^{(1)}$ and $\bar g_{mn}^{(2)}$ but its first-order asymptotic distribution is dominated by the slower $\bar g_{mn}^{(2)}$ term.


\begin{example}\label{ex:IV} We can illustrate this notation with a simple example. Consider the case of one individual-level variable $x_{1ij}$, one group-level variable $\tilde x_{2j}$, and a constant. As instrumental variables, we use $(\dot x_{1ij}, \tilde x_{2j},1)$, which corresponds to applying the fixed effects estimator for the coefficient on $x_{1ij}$ and the between estimator for the other two coefficients. In this case, only the first instrument has mean zero in all groups such that $L_1=1$ and $L_2=2$. By construction, $\dot x_{1ij}$ is uncorrelated with $\tilde x_{2j}$ and with the constant such that $\Sigma_{ZX}$ is block lower diagonal as defined in equation (\ref{eq:lower diagonal}) with $M_1=1$ and $M_2=2$. The coefficient on $x_{1ij}$ converges at the fast rate because it is not affected by the group-level effects $\alpha_j(\tau)$. On the other hand, the intercept and the coefficient on the group-level variable $\tilde x_{2j}$ converge only at the slow rate of convergence $\sqrt m$ because they are affected by the group effects. In a many application, we see that $M_1=K_1$, but this is not always true. For instance, in the previous example, all coefficients would converge at the slow rate if we used $(x_{1ij},\tilde x_{2j},1)$ as instrumental variables such that $L_1$ and $M_1$ would be equal to zero.
\end{example}

In an overidentified model, using a full-rank weighting matrix \( W(\tau) \) as in Assumption \ref{a:weighting} can lead to contamination of the entire parameter vector $\hat\delta(\tau)$ by the slower-converging moments. This occurs because $\hat G(\tau)$ will not retain a block-lower-triangular structure, even if $S_{ZX}$ does. In Example \ref{ex:re} below, this issue arises with the 2SLS estimator of $\delta_1(\tau)$, which converges at the slow $\sqrt{m}$ rate. In contrast, both the exactly identified IV estimator and the efficient GMM estimator achieve the faster $\sqrt{mn}$ rate.

To avoid contamination, we must give asymptotically infinitely higher weights to fast moments than slow ones. The following assumption imposes this critical condition on the weighting matrix.\footnote{The intuition behind the sequence $a_n(\tau)$ will be elucidated in the following subsection, where we delve into a comprehensive discussion of our efficient GMM estimator.}
\begin{assumptionp}{\ref*{a:weighting}$'$}[\textbf{Heterogeneous weighting matrix}] \label{a:weighting 2} Uniformly in $\tau\in\mathcal T$,
\begin{equation*}
    \hat W(\tau) = \underbrace{\begin{pmatrix}
        W_{11}(\tau) & a_n(\tau) W_{12}(\tau) \\ a_n(\tau) W_{21}(\tau) & a_n(\tau) W_{22}(\tau)
    \end{pmatrix}}_{W_{mn}(\tau)} + \begin{pmatrix}
        o_p(1) & o_p \left (\sqrt{a_n(\tau)}\right ) \\ o_p \left (\sqrt{a_n(\tau)} \right ) & o_p({a_n}(\tau))
    \end{pmatrix}
\end{equation*}
where $a_n(\tau)=\frac{1}{1+\Var(\alpha_j(\tau))n}$. $W_{11}(\tau)$ and $W_{22}(\tau)$ are respectively a $L_1\times L_1$ and a $L_2\times L_2$ full rank matrix. For all $l_1,l_2\in\{1,2\}$ and $\tau_1, \tau_2 \in \mathcal{T}$, $|| W_{l_1l_2} (\tau_2) - W_{l_1l_2}(\tau_1) || \leq C | \tau_2- \tau_1 |$. 
\end{assumptionp}

\begin{example}\label{ex:re} Consider an extension of Example \ref{ex:IV} where the vector of instrumental variable is now $(\dot x_{1ij}, \bar x_{1j}, x_{2j},1)$. This model is overidentified with $L_1=1$, $L_2=3$, and $K=3$. When we impose Assumption \ref{a:weighting}, $W(\tau)$ is a full rank weighting matrix, and the estimator converges only at the $\sqrt m$ rate. We can obtain a faster-converging estimator by giving asymptotically infinitely more weight to $\dot x_{ij}$ than to the other instruments. When we impose Assumption \ref{a:weighting 2}, the coefficient on $x_{1ij}$ is estimated at the $\sqrt{mn}$ rate. We formalize these results in Theorem \ref{t:first-order} below. 
\end{example}

\textbf{Case 2: no group-level heterogeneity.} We now consider the case where there is no heterogeneity across groups, i.e., when $\alpha_j(\tau)=0$ uniformly in $j$ and $\tau$. When this occurs, $\bar g^{(2)}_{mn}(\hat\delta,\cdot)=0$ such that all the variance arises from the first stage moment, which can be estimated at the fast rate $\sqrt{mn}$. This case corresponds to the textbook definition of a classical minimum distance estimator. All the coefficients converge at the fast $\sqrt{mn}$ rate, and the asymptotic distribution of $\hat\delta(\tau)$ can be derived straightforwardly. By comparing Case 1 and 2, we notice that the limiting distribution and even the rate of convergence of the $M_2$ coefficients $\hat\delta_2(\tau)$ differ. Their asymptotic distribution is discontinuous at $0$ in the variance of $\alpha_j(\tau)$, and in many applications, we do not know whether there is group-level heterogeneity. For this reason, we discuss adaptive estimation and inference in the next subsection.

\begin{remark}[Related literature]
\cite{Chamberlain1994} considers a setting with a finite number of design points (groups in our terminology), exogenous group-level variables, and no individual-level variables. His correctly specified case corresponds to our Case 2 (no heterogeneity). Accordingly, our asymptotic distribution corresponds to his in this special case.\footnote{In the absence of group-level heterogeneity, the asymptotic distribution of $\hat\delta(\tau)$ stated in part (ii) of Theorem \ref{t:first-order} is also valid when the number of groups is fixed.} \cite{Chamberlain1994} also considers a misspecified case. We can interpret his misspecification errors as our group effects. However, in this case, he considers pseudo-true parameters that absorb a non-vanishing bias while we allow the number of groups to go to infinity to avoid bias. For this reason, our results fundamentally differ from his results in Case 1.

\cite{Chetverikov2016} do not consider Case 2 (no heterogeneity) in their theoretical results even if, interestingly, it corresponds to one of their data-generating processes in their simulations. In Case 2, their matrix-valued function $J(u_1,u_2)$ defined in their Assumption 6(ii) is uniformly equal to zero. This implies that their asymptotic covariance function $\mathcal{C}(u_1,u_2)$ defined in their Theorem 1 is also uniformly equal to 0. This degenerate asymptotic distribution indicates that the convergence rate is faster than $\sqrt m$ in this case.
\end{remark}

\textbf{Case 3: weak group-level heterogeneity.} In Case 1, the second-stage variance dominates the asymptotic distribution of $\hat\delta_2(\tau)$, while in Case 2, the first-stage variance dominates. It is interesting to consider the intermediate case when group-level heterogeneity is present but vanishes exactly at the right rate such that both components of the variance matter asymptotically. This should provide a good approximation for the applications where the first-stage and the second-stage variances are similar. Formally, we assume that $\Omega_2(\tau_1,\tau_2)$, the covariance function of $\bar g^{(2)}_{mn}(\hat\delta,\tau)$ defined in Assumption \ref{a:group effects}(ii), converges to zero at the $n$ rate:
$$\Omega_2(\tau_1,\tau_2)=n^{-1}\bar\Omega_2(\tau_1,\tau_2).$$ Under this assumption, $\bar g^{(1)}_{mn}(\hat\delta,\cdot)$ and $\bar g^{(2)}_{mn}(\hat\delta,\cdot)$ converge at the same rate. Lemma \ref{l:asymptotic moments} implies then
\begin{equation}\label{eq:asymptotic moment}
    \sqrt{mn}\bar g_{mn}(\hat\delta, \cdot) \rightsquigarrow \mathbb Z(\cdot)\text{, in $\ell^\infty(\mathcal T)$,}
\end{equation}
where $\mathbb Z(\cdot)$ is a mean-zero Gaussian process with uniformly continuous sample paths and covariance function $\Omega_{1}(\tau,\tau')+\bar\Omega_2(\tau,\tau')$.

Theorem \ref{t:first-order} formally states the results for these three cases. Parts (i) provides the asymptotic distribution for the fast and slow coefficients when there is strong heterogeneity, part (ii) in the absence of heterogeneity, and part (iii) when there is weak heterogeneity. 

\begin{theorem}[\textbf{Asymptotic distribution when the degree of heterogeneity is known}]\label{t:first-order}
Let Assumptions \ref{a:sampling}-\ref{a:group effects} hold.

\begin{enumerate}[label=(\roman*), wide, labelindent=0pt]
    \item Case 1 (strong heterogeneity): $\Var(\alpha_j(\tau))>\varepsilon>0$ uniformly in $\tau$ and Assumption \ref{a:weighting 2} holds.
    \begin{enumerate}
        \item In addition, let Assumption \ref{a:growth condition}(c) hold. Then,
        \begin{equation}
        \sqrt{mn}(\hat\delta_1(\hat W(\cdot), \cdot)-\delta_1(\cdot))\rightsquigarrow G_{11}(\cdot) \mathbb Z_{11}(\cdot)\text{, in $\ell^\infty(\mathcal T)$,}\end{equation}
        where $G_{11}(\tau) = \left( \Sigma_{11}' W_{11}(\tau) \Sigma_{11}\right)^{-1} \Sigma_{11}' W_{11}(\tau)$ and $\mathbb Z_{11}(\cdot)$ is the Gaussian process consisting of the first $L_1$ elements of $\mathbb Z_{1}(\cdot)$ defined in Lemma \ref{l:asymptotic moments}(i).
    
        \item In addition, let Assumption \ref{a:growth condition}(b) hold. Then, 
        \begin{equation}
        \sqrt{m}(\hat\delta_2(\hat W(\cdot), \cdot)-\delta_2(\cdot))\rightsquigarrow G_{22}(\cdot) \mathbb Z_{22}(\cdot)\text{, in $\ell^\infty(\mathcal T)$,}\end{equation}
        where $G_{22}(\tau) = \left( \Sigma_{22}' W_{22}(\tau) \Sigma_{22}\right)^{-1} \Sigma_{22}' W_{22}(\tau)$ and $\mathbb Z_{22}(\cdot)$ is the Gaussian process consisting of the last $L_2$ elements of $\mathbb Z_{2}(\cdot)$ defined in Lemma \ref{l:asymptotic moments}(ii).
    \end{enumerate}
    
    \item Case 2 (no heterogeneity): $\alpha_j(\tau)=0$ uniformly in $j$ and $\tau$. In addition, Assumption \ref{a:growth condition}(c) and \ref{a:weighting} hold. Then,
    \begin{equation}
    \sqrt{mn}(\hat\delta(\hat W(\cdot), \cdot)-\delta(\cdot))\rightsquigarrow G(\cdot) \mathbb Z_1(\cdot)\text{, in $\ell^\infty(\mathcal T)$,}\end{equation}
    where $G(\tau) = \left( \Sigma_{ZX}' W(\tau) \Sigma_{ZX}\right)^{-1} \Sigma_{ZX}' W(\tau)$ and $\mathbb Z_1(\cdot)$ is defined in Lemma \ref{l:asymptotic moments}(i).
    
    \item Case 3 (weak heterogeneity): $\Omega_2(\tau,\tau') = n^{-1} \bar\Omega_2(\tau,\tau')$. Assumption \ref{a:weighting} and \ref{a:growth condition}(c) hold. Then, 
    \begin{equation}
    \sqrt{mn}(\hat\delta(\hat W(\cdot), \cdot)-\delta(\cdot))\rightsquigarrow G(\cdot) \mathbb Z(\cdot) \text{, in $\ell^\infty(\mathcal T)$,}\end{equation}
    where $G(\tau) = \left( \Sigma_{ZX}' W(\tau) \Sigma_{ZX}\right)^{-1} \Sigma_{ZX}' W(\tau)$ and $\mathbb Z(\cdot)$ is a mean-zero Gaussian process with uniformly continuous sample paths and covariance function $\Omega_{1}(\tau,\tau')+\bar\Omega_2(\tau,\tau')$.
    \end{enumerate}
\end{theorem}

These asymptotic results are useful for understanding the mechanics behind our estimator, but they have several weaknesses. First, the asymptotic distribution of $\hat\delta_1(W, \tau)$ in part (i)-(a) is only a function of the fast instruments. The same asymptotic distribution can be obtained by ignoring the slow instruments. Consider Example \ref{ex:re}, including $\bar x_{1j}$ as an instrument does not reduce the asymptotic variance of $\hat\delta_1(W, \tau)$ even if this instrument is valid and the between-group variation in $x_{1ij}$ is non-negligible. In other words, the random effects estimator is asymptotically equivalent to the fixed effects estimator.\footnote{This is not specific to quantile models and also affects least squares models with large $n$ (see \citealp{Ahn2014}).} In some applications, the random effects estimator is appreciably more precise such that we would like to exploit the between-group variation efficiently. 
Second, the asymptotic distribution and the convergence rate of the slow coefficients $\hat \delta_2(W, \tau)$ are different, depending on whether there is group-level heterogeneity or not. Thus, performing inference based on Theorem \ref{t:first-order} requires knowing which case is relevant for the specific application. In other words, inference based directly on this result is not adaptive.  Third, in Case 1, the variance coming from the first stage estimation does not appear in the asymptotic distribution of $\hat\delta_2(W, \tau)$ because it converges to zero at a quicker rate. Consequently, inference may have poor properties. We solve these issues in the following subsection by suggesting an efficient estimator and adaptive inference that are both valid in all three cases and more generally uniformly valid in the variance of $\alpha_j(\tau)$.





\subsection{Adaptive Estimation and Inference}
\label{subsec:adaptive}

To establish asymptotic results that are uniformly valid in the variance of \( \alpha_j(\tau) \), we allow \( \Var(\alpha_j(\tau)) \) (and consequently \( \Omega_2(\tau,\tau) \)) to be exactly zero, bounded away from zero, or a sequence converging to zero at an arbitrary rate. This approach nests all three pointwise cases described in the previous section. The sequence $a_n(\tau)=\frac{1}{1 + \Var(\alpha_j(\tau)) n}$ defined in Assumption \ref{a:weighting 2} plays a central role because it is proportional to the rate of convergence of the slow moments relative to the fast moments $\parallel \Var(\bar g_{mn,2}(\delta,\tau)) \parallel^{-1} \parallel  \Var(\bar g_{mn,1}(\delta,\tau) ) \parallel  = O_p(a_n(\tau))$. It is equal to $1$ without heterogeneity, converges to $0$ at the $n^{-1}$ rate with strong heterogeneity, and is always bounded between $0$ and $1$. 
All matrices that are functions of $a_n(\tau)$ may also depend on the sample size. When necessary, we make the notation explicit with the subscript $mn$. For example, let
\begin{equation}\label{eq:Gmn matrix}
    G_{mn}(\tau) = \left ( \Sigma_{ZX}' W_{mn} (\tau) \Sigma_{ZX}  \right) ^{-1} \Sigma_{ZX}' W_{mn}(\tau) ,
\end{equation}
where $W_{mn}(\tau)$ is defined in Assumption \ref{a:weighting 2}. 

Similarly to \cite{Fernandez-Val2022a}, we define the covariance kernel of the limiting process of $\hat\delta(\cdot)-\delta(\cdot)$. For a given integer $T>0$, let $\mathcal{T}_T=(\tau_1,\dots,\tau_T)$ be an arbitrary $T$-dimensional vector on $\otimes_{t=1}^T\mathcal{T}$. Let
\begin{equation*}
    \Sigma_{mn}(\tau,\tau')=G_{mn}(\tau)\left(\frac{\Omega_1(\tau,\tau')}{mn}+\frac{\Omega_2(\tau,\tau')}{m}\right)G_{mn}(\tau')'
\end{equation*}
and $\Sigma_{mn}(\tau)=\Sigma_{mn}(\tau,\tau)$.
The covariance kernel is now given by the limit of the elements of the following $(KT)\times (KT)$ matrix
\begin{equation*}
    H_{mn}=\left(H_{mn}(\tau,\tau')\right)_{T\times T},
\end{equation*}
where
\begin{equation*}
    H_{mn}(\tau,\tau')=\diag(\Sigma_{mn}(\tau))^{-1/2}\Sigma_{mn}(\tau,\tau')\diag(\Sigma_{mn}(\tau'))^{-1/2}.
\end{equation*}

\begin{assumption}[\textbf{Covariance kernel}]\label{a:kernel}
For any integer $T>0$ and any $T$-dimensional vector $\mathcal{T}_T=(\tau_1,\dots,\tau_T)$ on $\otimes_{t=1}^T\mathcal{T}$, there is a $(KT)\times(KT)$ matrix $H$, such that, almost surely,
\begin{equation*}
\underset{m,n\rightarrow\infty}{\lim}H_{mn}=H.
\end{equation*}
In addition, there is $c_{\mathcal{T}_T}>0$ such that for the smallest eigenvalue we have
\begin{equation*}
    \lambda_{\min}(H)>c_{\mathcal{T}_T}.
\end{equation*}
\end{assumption}

We can now state the asymptotic distribution of our estimator that is uniformly valid in $\Var(\alpha_j(\tau))$.

\begin{theorem}[\textbf{Adaptive asymptotic distribution}]\label{t:adaptive}
Assumptions \ref{a:sampling}-\ref{a:growth condition}(c) and \ref{a:kernel} hold. In addition, either Assumption \ref{a:weighting} or Assumption \ref{a:weighting 2} holds. Then, 
    \begin{equation*}
        \diag(\Sigma_{mn}(\cdot))^{-1/2}(\hat\delta(\cdot)-\delta(\cdot))\rightsquigarrow \mathbb G(\cdot)
    \end{equation*}
  where $\mathbb G$ is a centered Gaussian process with a covariance function $H(\tau_k,\tau_l)$.
\end{theorem}

Theorem \ref{t:adaptive} allows for both types of weighting matrices: the asymptotically full-ranked weighting matrix of Assumption \ref{a:weighting} or the heterogeneous weighting matrix of Assumption \ref{a:weighting 2}. The first case covers the exactly identified case and 2SLS. As we show in Proposition \ref{p:weight matrix adaptive} below, our estimated efficient weighting matrix satisfies the second condition.

In order to use these results for inference, we must provide a consistent estimator of the asymptotic variance $\Sigma_{mn}(\tau,\tau')$. The crucial ingredient is the asymptotic variance of the sample moments that we denote by
\begin{equation*}
\Omega_{mn}(\tau,\tau')=\Omega_1(\tau,\tau')/n+\Omega_2(\tau,\tau')
\end{equation*}
Remember that $\Omega_1(\tau,\tau')$ arises from the first stage estimation and $\Omega_2(\tau,\tau')$ from the presence of the group-level heterogeneity $\alpha_j(\cdot)$. The difficulty resides in that the leading term may arise from the first-stage or second-stage estimation depending on the variance of $\alpha_j(\cdot)$ and the type of instrument ($L_1$ or $L_2$). This expression may suggest that we must estimate these two components separately. However, and perhaps surprisingly, we find that we can estimate the leading term of $\Omega_{mn}(\tau,\tau')$ adaptively with a traditional cluster-robust estimator of the variance:
\begin{equation}\label{eq:estimator omega}
  \hat \Omega(\tau,\tau')=\frac{1}{m} \sum_{j = 1}^m \left\{\left(\frac{1}{n}\sum_{i=1}^n z_{ij} \hat u_{ij}(\tau)\right)  \left(\frac{1}{n}\sum_{i=1}^n z_{ij}\hat u_{ij}(\tau')\right)'\right\}
\end{equation}
where $\hat u_{ij}(\tau)$ are the second-stage residuals.

\begin{proposition}[\textbf{Properties of $\hat \Omega(\tau, \tau')$}]\label{p:omegamatrix}
Let assumptions \ref{a:sampling}-\ref{a:group effects} and \ref{a:growth condition}(c) hold. Further, assume that as $m\rightarrow\infty$, for each $l,l'\in \{1,\dots,L\}$ and uniformly in $\tau,\tau'\in\mathcal{T}^2$,
\begin{equation}\label{a:variance}
    m^{-1}\sum_{j=1}^m \bE \left[\left(\bar z_{jl} \bar z_{jl'}\alpha_j(\tau)\alpha_j(\tau')-\Omega_{mn,2ll'}(\tau,\tau')\right)^2\right]\rightarrow C_{l,l'}(\tau,\tau')<\infty
\end{equation}
where $C_{l,l'}(\tau,\tau')$ is continuous in $\tau$ and $\tau'$.
The estimator used to compute $\hat u_{ij}(\tau)$ satisfies (i) $\hat\beta(\tau)-\beta(\tau)=O_p\left (1/\sqrt{mn}\right)$, (ii) $\hat\gamma(\tau)-\gamma(\tau) = O_p(1/\sqrt{mn}) 
+ O_p(\sqrt{\Var(\alpha_j(\tau))}/\sqrt m)$ uniformly in $\tau$. Then, for any $ll'$ entry of the $\hat \Omega(\tau, \tau')$ matrix with $l,l' \in \{1, \dots , L \}$ we have uniformly in $\tau, \tau' \in \mathcal{T}^2$,
\begin{align*}
        \hat \Omega_{ll'}(\tau,\tau')= \Omega_{mn, ll'}(\tau,\tau') +  o_p \left ( 
 \sqrt{ \Omega_{mn,ll} (\tau) \Omega_{mn,l'l'}(\tau') }   \right ).
\end{align*}
\end{proposition}

Beyond the numbered assumptions, two additional conditions are necessary. First, equation (\ref{a:variance}) requires the average variance of $\bar z_j\alpha_j(\tau)$ to converge to a constant. This holds automatically when we sample identically distributed groups but is required because we allow for heterogeneity across groups. Second, we require that the coefficients on the individual-level variables are estimated at the $1/\sqrt{mn}$ rate and the coefficients on the group-level variable at the $1/\sqrt{mn} + \sqrt{\Var(\alpha_j(\tau))}/\sqrt{m}$ rate. This corresponds to assuming that $M_1 = K_1$ so that $\beta = \delta_1$. This assumption ensures that estimation errors in slow coefficients do not contaminate the variance estimate of fast ones. This condition can always be achieved, for instance, by using $\dot x_{1ij}$ as an instrument for $x_{1ij}$ and a weighting matrix that satisfies Assumption \ref{a:weighting 2}. We show below that our efficient weighting matrix satisfies this restriction.\footnote{The only exception in this paper is the between estimator from Section \ref{sec:FE+BE+RE}. In that case, all the coefficients are estimated at the slow rate $1/\sqrt{mn} + \sqrt{\Var(\alpha_j(\tau))}/\sqrt{m}$. Proposition \ref{p:omegamatrix2} in the Appendix demonstrates the consistency of $\hat\Omega(\tau,\tau')$ in this case as well.}

The suggested estimator of the variance is straightforward to implement because it does not require estimating directly the variance of the first-stage quantile regression coefficients. This is a noteworthy advantage because the variance of these coefficients depends on the conditional density of the outcome given the covariates, which can be difficult to estimate. To understand this surprising result, consider the case where $\alpha_j(\tau)=0$ for all groups. This implies that the entire variance comes from the first-stage estimation. The clustered estimator of the variance can be interpreted as a subsampling estimator where the groups represent the subsamples. The variance of the coefficients across groups provides information about the precision of the first-stage estimates. Similar results for the cross-sectional bootstrap are provided in \cite{Liao2018}, \cite{Lu2022}, and \cite{Fernandez-Val2022a}.

Proposition \ref{p:omegamatrix} shows we can consistently estimate the diagonal elements of $\Omega_{mn}(\tau)$ in the sense that the error that we make vanishes more quickly than the true variance. On the other hand, if two sample moments vanish at different rates, we might not be able to estimate their covariance consistently. However, this does not preclude testing hypotheses about combinations of coefficients because the error that we make when we estimate the covariance vanishes faster than the variance of the slowest coefficient. We formalize this result in Proposition \ref{p:cov matrix} below for the case of a single linear null hypothesis.\footnote{We can similarly extend this proposition to multiple non-linear hypotheses.}

Define the $TK\times 1$ vector $\delta(\tau_1,\dots, \tau_T) = \left( \delta(\tau_1)', \dots, \delta(\tau_T)' \right) '$ whose covariance matrix is given by the elements of the following $(KT)\times (KT)$ matrix:
\begin{equation*}
    \Sigma_{mn}=\left(\Sigma_{mn}(\tau,\tau')\right)_{T\times T}.
\end{equation*}
For each $\tau, \tau' \in \mathcal T \times \mathcal T$, $\Sigma_{mn}(\tau, \tau')$ is estimated by
\begin{align*}
    \hat \Sigma(\tau, \tau') = \hat G(\tau) \frac{\hat \Omega(\tau, \tau')}{m} \hat G(\tau') 
\end{align*}
where $\hat \Omega(\tau, \tau')$ is defined in equation (\ref{eq:estimator omega}). The following proposition shows that a straightforward z-test provides valid inference that is adaptive in the degree of heterogeneity.

\begin{proposition}[\textbf{Asymptotic normality of the test statistic}]\label{p:cov matrix}
Let the conditions for Theorem \ref{t:adaptive} and Proposition \ref{p:omegamatrix} hold. In addition, let $\eta \in \mathbb{R}^{T\times K}$ with $||\eta|| > \eps > 0$. 
Then,
\begin{equation*}
    \frac{\eta' \left (\hat \delta(\tau_1,\dots, \tau_T) -\delta(\tau_1,\dots, \tau_T) \right)  }{ \sqrt{\eta' \hat \Sigma \eta }} \xrightarrow{d} N(0,1).
\end{equation*}
\end{proposition}

The asymptotic variance derived in Theorem \ref{t:adaptive} depends on the weighting matrix. Following standard GMM arguments, the efficient weighting matrix is given by 
\begin{equation} W(\tau)^*=\Omega_{mn}(\tau)^{-1}=\left(\Omega_1(\tau) / n + \Omega_2(\tau)\right)^{-1}.
\end{equation}
This weighting matrix automatically takes into account the different rates of convergence of the different moments. If some moments converge faster than others, then this matrix, asymptotically, gives infinitely more weight to the fast moments than the slow moments, so the parameters identified by the fast moments will converge at the faster rate.

Usually, we would simply plug in a consistent estimator of $W^*(\tau)$ and obtain a feasible efficient GMM estimator, but Proposition \ref{p:omegamatrix} shows that the off-diagonal elements of $\Omega_{mn}(\tau)$ might not be consistently estimated when the rates of convergence of the corresponding coefficients are different. Proposition \ref{p:weight matrix adaptive} shows that the estimated efficient weighting matrix nevertheless satisfies Assumption \ref{a:weighting 2} such that Theorem \ref{t:adaptive} and Proposition \ref{p:cov matrix} apply to the feasible efficient GM estimator. In addition, Theorem \ref{t:adaptive} reveals that the (first-order) asymptotic distribution of the estimator that uses $W^*(\tau)$ as the weighting matrix is the same as that of the estimator that uses $\hat W^*(\tau)$ as the weighting matrix.
\begin{proposition}[\textbf{Adaptive Efficiency of the GMM Estimator}]    
\label{p:weight matrix adaptive}
Assume that the conditions for Proposition \ref{p:omegamatrix} hold. Then, the estimated efficient weighting matrix $\hat W^* (\tau) =\hat\Omega(\tau)^{-1}$ satisfies Assumption \ref{a:weighting 2}.
\end{proposition}

If there are more moment conditions than parameters to estimate ($L > K$), it is possible to implement overidentification tests in the spirit of \cite{Sargan1958} and \cite{Hansen1982}. More precisely, we can test the validity of the instrumental variables while maintaining the other assumptions.
The test statistic is the GMM objective function evaluated at the efficient GMM estimator:
\begin{equation}
J\left(\hat\delta\left(\hat W^*(\tau),\tau\right),\tau\right)=m \bar g_{mn}\left(\hat\delta\left(\hat W^*(\tau),\tau\right),\tau\right)'\hat\Omega(\tau)^{-1}\bar g_{mn}\left(\hat\delta\left(\hat W^*(\tau),\tau\right),\tau\right).
\end{equation}
Proposition \ref{prop:hausmant test} demonstrates the adaptive validity of the $J$-test implemented with the clustered covariance matrix.
\begin{proposition}[\textbf{Overidentification Test}]\label{prop:hausmant test}
Under the $\mathbb{H}_0: \ \bE[z_{ij}\alpha_i(\tau)] = 0$ for all $\tau\in \mathcal T$ and the assumptions required for Proposition \ref{p:omegamatrix}, 
as $m \rightarrow \infty$, $J\left(\hat\delta\left(\hat W^*(\tau),\tau\right),\tau\right)  \xrightarrow{d} \chi_{L-K}^2$.
\end{proposition}
We prove the result for a single quantile, but the proof easily extends to multiple quantiles. In Section \ref{s:hausman}, we use this overidentification test to suggest a quantile analog of the Hausman test for the exogeneity of
the between variation. 

\section{Grouped (IV) Quantile Regression Model}\label{sec:GIVQR}

\subsection{Chetverikov et al. (2016)}

\citet[][hereafter CLP]{Chetverikov2016} introduce an IV quantile regression estimator for group-level treatments and propose an alternative two-step estimation procedure for model (\ref{eq:model}). In the first stage, they perform quantile regression separately for each group $j$ and quantile $\tau$, regressing $y_{ij}$ on $x_{1ij}$ and a constant. In the second stage, they regress the \textit{intercept} from the first stage on $x_{2j}$. The key distinction between their approach and ours is that CLP use the estimated intercept from the first stage as the dependent variable in the second stage, whereas we use the fitted values.

The intercept represents the fitted value for an observation with $x_{1ij} = 0$. Consequently, their estimator is not invariant to linear reparametrizations of the individual-level regressors. In finite samples, results can differ depending on how we code the variables. For instance, the estimated coefficients may change if age is recorded as years since birth versus years since 16 or if we measure temperature in Celsius versus Fahrenheit. Our estimator, by contrast, does not suffer from this issue. Beyond the undesirable dependence of the results on an arbitrary linear transformation of the individual-level variables, this property also affects the precision of the estimates and introduces potential misspecification bias.  

Figure \ref{fig:firststages} illustrates these concerns using artificial data. Panel (a) shows how the variance of the fitted values increases as we move away from the mean of the regressor. The solid line represents the median regression estimate, while the shaded area depicts the 95\% confidence interval for the median fitted values. The intercept corresponds to the fitted value for an $x$ value that lies far outside the observed support of the variable. As a result, its variance is relatively large and would increase even more if we shifted the location of the individual-level variable.

Panel (b) highlights the misspecification bias using a similar artificial dataset. The dashed orange line represents the true median regression function, while the solid line shows the estimated linear model, which is slightly misspecified. Over the support of the covariates, the misspecification bias remains small because quantile regression minimizes the weighted mean squared error (see \citealp{Angrist2006a}). However, the intercept is more strongly biased, as it lies outside the covariate’s support, and this bias monotonically worsens as we increase the mean of the individual-level covariate.

\begin{figure}
\centering
\begin{subfigure}[b]{0.49\textwidth}
   \centering
    \resizebox{1\linewidth}{!}{
   \includegraphics{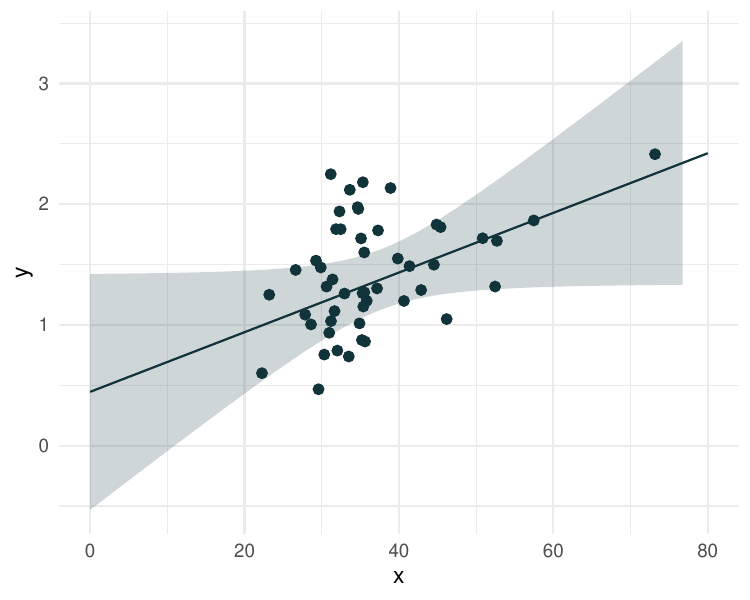}}
      \caption{ Extrapolation }
\label{fig:extrapolation}
\end{subfigure}
\hfill
\begin{subfigure}[b]{0.49\textwidth}
\centering 
   \resizebox{1\linewidth}{!}{
  \includegraphics{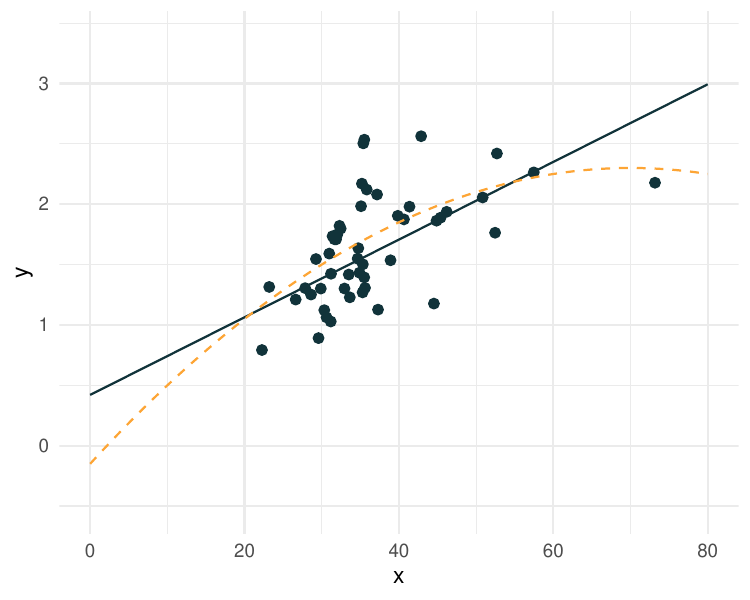}}
               \caption{ Misspecification }
                 \label{fig:misspecification}  
\end{subfigure}
 \caption{First Stage Regressions}
  \label{fig:firststages}  
  \floatfoot{Both panels show generated data for one group and a first-stage fit. Panel (a) uses a correctly specified model. The solid line shows the regression line estimated by median regression, and the shaded area shows the 95\% confidence interval. Panel (b) uses a true quadratic model. The solid line is estimated by linear median regression. The dashed orange line is the true regression line.}
\end{figure}

\subsection{Simulations}\label{sec:simulationsGIVQR}

In this subsection, we compare the MD and CLP estimators of $\gamma(\tau)$ using Monte Carlo simulations. In the following subsection, we formally compare their asymptotic distributions to explain the simulation results. We use the same three data-generating processes (DGP) as in \cite{Chetverikov2016} and generate the outcome as follows: 
\begin{equation} \label{DGP1}
    y_{ij} = \beta_{0}(u_{ij}) + x_{1ij} \beta(u_{ij}) + x_{2j} \gamma(u_{ij}) + \alpha_j(u_{ij}) ,
\end{equation}
where $x_{1ij}$ and $x_{2j}$ are $\textnormal{exp}(0.25 \cdot N[0,1])$ and individual heterogeneity is introduced via the rank variable $u_{ij}\sim U[0,1]$. The quantile coefficient functions are $\gamma(\tau)=\beta (\tau)=\sqrt{\tau}$ and $\beta_1(\tau) = \frac{\tau}{2}$ for $\tau\in [0,1]$. In the first DGP, we set $\alpha_j(u_{ij})=0$, so there is neither group heterogeneity nor endogeneity. In the second DGP, we introduce group-level heterogeneity as
\begin{equation}\label{DGP3}
    \alpha_j(u_{ij}) = u_{ij} \eta_j - \frac{u_{ij}}{2},
\end{equation}
where $\eta_{j}\sim U(0,1)$. There is no endogeneity because the group effects are uncorrelated with the regressors. As group-level heterogeneity is multiplied with the rank variable $u_{ij}$, there is only weak group heterogeneity in the lower tail of the distribution and strong heterogeneity in the upper tail. Finally, in the third DGP, we introduce endogeneity as
\begin{equation} \label{DGP2}
    x_{2j} =  z_j + \eta_j + \nu_j ,
\end{equation}
where $z_j$ and $\nu_j$ are each distributed $\textnormal{exp}(0.25 \cdot N[0,1])$. $z_j$ is a valid instrumental variable for the endogenous $x_{2j}$. We use the same sample sizes as CLP, that is, $(m,n) = \{(25,25),\allowbreak (200,25),\allowbreak (25,200),\allowbreak (200,200) \}$. We perform 10,000 Monte Carlo replications for the set of quantiles $ \tau \in \{0.1, 0.5, 0.9\}$. Since the CLP estimator does not directly estimate $\beta(\tau)$, we present only results for $\gamma(\tau)$.

\input{Tables/clp_simul}
\input{Tables/cilength_simul}

The first stage, group-by-group quantile regressions, are the same for both estimators. In the second stage, CLP regress the estimated intercepts on $x_{2j}$ with OLS (DGP 1 and 2) or using $z_j$ as an instrument (DGP 3). We implement the MD estimator with the IV estimator instrumenting $(x_{1ij},x_{2j})$ with $(\dot x_{1ij},x_{2j})$ (DGP 1 and 2) or $(\dot x_{1ij},z_j)$ (DGP 3). With this choice of instruments, we do not exploit the between variation of $x_{1ij}$.\footnote{Using $x_{1ij}$ instead of $\dot x_{1ij}$ as an instrument has virtually no effect on the results since there is no variation in $x_{1ij}$ across groups.} Since the data generating process of \cite{Chetverikov2016} has a weak instrument when $m$ is small, one should pay attention when looking at the simulation results for the endogenous case.\footnote{With $m = 25$ in over 40\% of the draws, the F-statistics of the first stage of the IV estimations is below 10. The issue disappears when $m = 200.$} It would be straightforward to compute weak instrument robust inference, for instance, using Anderson/Rubin confidence intervals in the second stage. However, we consider this to be outside the scope of this paper.


Table \ref{tab:tab:grouped_bias} presents the bias, standard deviation, and the relative MSE of the MD estimator, defined as the MSE of the MD estimator divided by that of the CLP estimator. No clear pattern emerges regarding the bias of the estimators. Consistent with asymptotic theory, the bias diminishes as the number of observations increases. We observe more pronounced differences in the variance of the estimators. In the homogeneous case, the standard deviation of the MD estimator is four times smaller than that of the CLP estimator. This difference remains stable as the number of groups or individuals increases. The advantage of the MD estimator remains similar at the bottom of the distribution in the exogenous case, but the difference in variance narrows at the upper end. We explain this with the presence of weak heterogeneity at the bottom of the distribution and strong heterogeneity at the top. At the upper end of the distribution, the variance of the estimators converges as $n$ becomes large. Finally, in the endogenous case, the results resemble those of the exogenous case once we exclude findings affected by the weak instrument issue, precisely when $m=25$. The differences in precision explain the large discrepancies in MSE. The MSE of the CLP estimator is up to twenty times larger than that of the MD estimator when $\alpha_j(\tau) = 0$ and remains substantially larger in all scenarios considered except one with a weak instrument.\footnote{Although not shown here, we also computed the traditional quantile regression estimator. When there is heterogeneity, quantile regression is inconsistent for the parameter of interest; however, when $\alpha_j(\tau) = 0$, our estimator and quantile regression are practically indistinguishable in terms of bias and variance while the CLP estimator has a much larger variance.}

Table \ref{tab:tab:cilength} shows the performance of the $95\%$ confidence intervals suggested with our inference procedure. The table reports the coverage rate and the median length of the intervals of our estimator relative to that of the CLP estimator.\footnote{For the CLP estimator, we use heteroskedasticity robust standard errors as suggested in \cite{Chetverikov2016}. Recall that the CLP estimator uses only one observation per group in the second stage. Hence, we would attain the same standard errors if we kept all observations and clustered the standard errors at the group level.} Our suggested inference procedure has coverage close to $95\%$ in all cases. Compared to the CLP estimator, our confidence bands are substantially shorter. In most cases, our estimator yields confidence bands that are less than half the length of those for the CLP estimator. This difference becomes even more pronounced in our empirical application in Section \ref{sec:food stamp}, where the CLP estimator produces confidence bands that are, on average, 14 times wider than those of the MD estimator.

\subsection{Comparison of the asymptotic distributions of CLP and MD estimators}\label{sec:comparison CLP}

CLP focus exclusively on $\gamma(\tau)$, the coefficients on the group-level variables. They assume strong group-level heterogeneity, imposing that $\Var(\alpha_j(\tau)) > 0$ uniformly in $\tau$.\footnote{Although they do not state this assumption explicitly, their asymptotic distribution would become degenerate without it, as they discuss in footnote 9.} Thus, their asymptotic distribution corresponds to case (i)-b of our Theorem \ref{t:first-order}, where the variance from the first stage diminishes more rapidly than that from the second stage. In this subsection, we compare the variance of the CLP and MD estimators, allowing for weak or no heterogeneity.


To simplify notation, we consider the exogenous case where $x_{2j}$ can serve as its own instrumental variable, though the results also extend to the endogenous case. As discussed in Section \ref{sec:asym. theory}, the asymptotic variance consists of two components: one accounting for first-stage error and the other for second-stage noise. We will examine each component separately. First, we consider the scenario where the true first-stage coefficients are known to isolate the variance arising in the second stage. In this scenario, the CLP point estimates can be obtained numerically within our MD framework by regressing the true first-stage fitted values on $x_{1ij}$ and $x_{2j}$, using $\dot x_{1ij}$ and $x_{2j}$ as instruments. Thus, the second-stage variance resulting from the randomness of $\alpha_j(\tau)$ is identical for both estimators. 

We can isolate the first-stage error by setting $\Var(\alpha_j(\tau))=0$. Under this condition, both estimators are classical MD estimators. We can express the CLP estimator as
\begin{equation}
    \hat\delta_{CLP}(\tau)=\underset{\delta\in \mathbb R^{K_1\cdot m +K_2}}{\arg\min} \sumiN \left (\hat \beta_j (\tau) - \tilde R_j \delta \right )' \left (\hat \beta_j (\tau) - \tilde R_j \delta \right ),
\end{equation}
where 
\begin{equation*}
    \underset{\scriptscriptstyle (K_1+1) \times ( K_1 \cdot m + K_2)}{ \tilde R_j} = \begin{pmatrix} 0 & x_{2j}'  \\ l_j'  \otimes I_{K_1} & 0
    \end{pmatrix},
\end{equation*}
and $l_j$ is a $m$-dimensional vector of zeros with a $1$ in the $j$ position. The restriction matrix $\tilde R_j$ differs from the restriction matrix of our estimator defined in equation (\ref{eq:restriction}), as it does not impose equality of the first stage coefficients implied by the model. Thus, our estimator imposes $K_1\cdot (m-1)$ additional correct restrictions. A second difference is that CLP use an identity weighting matrix while we use $\tilde X_j'\tilde X_j$. When there is no group heterogeneity, we are in the classical MD framework, where the efficient weighting matrix is the inverse of the first-stage variance. It follows that weighting by $\tilde X_j'\tilde X_j$ is efficient when the first-stage error is separable and the density of $y$ given $x$ at the $\tau$ quantile is the same across groups. When, in addition, $\tilde X_j'\tilde X_j$ is constant across groups (balanced panel, identical distribution of $x_{1ij}$), then equal weighting is efficient.

Adding valid constraints within an efficient MD framework reduces variance. Therefore, if $\tilde{X}_j'\tilde{X}_j$ is the efficient weighting matrix, our estimator will necessarily have lower first-stage variance than the CLP estimator. However, with an inefficient weighting matrix, adding valid constraints might increase the variance in some relatively pathological cases.\footnote{See the discussion in Section 8 of \cite{Hansen2021}.} While it would be possible to estimate the efficient weighing matrix, we prefer to avoid estimating the first-stage variance and instead opt for a more interpretable estimator in cases of misspecification.

To summarize the comparison, the MD estimator using $\dot{x}_{1ij}$ and $x_{2j}$ as instrumental variables exhibits the same variance due to the randomness of $\alpha_j(\tau)$ but a lower variance due to estimation of $\hat\beta_j(\tau)$ compared to the CLP estimator—this holds formally when the efficient weighting matrix is used. In light of these results, we can understand the simulation results. In the first DGP, all the variance arises from the first stage such that our estimator is more precise, even asymptotically. The relative MSE does not change as $n$ increases. In the second and third DGP, the variances of both estimators will converge as $n\rightarrow \infty$ because the second-stage variance will asymptotically dominate them. This convergence appears, however, to be relatively slow, especially at the bottom of the distribution, where heterogeneity is weak.


Note that CLP also consider a generalization of model (\ref{eq:model}) in which they assume
\begin{align}
    Q(\tau, y_{ij}| x_{1ij} , x_{2j}, v_j) =& \tilde {x}_{1ij}' \beta_j (\tau) \label{eq:clpmodel1} \\ 
    \beta_{j,1}(\tau) =& x_{2j}' \gamma(\tau) + \alpha(\tau, v_j).  \label{eq:clpmodel2}
\end{align}
By default, $\beta_{j,1}(\tau)$ is the first element of the vector $\beta_{j}(\tau)$, but it could represent any element of this vector. In contrast to model (\ref{eq:model}), this approach allows the coefficient on $x_{1ij}$ to vary across groups. It also enables researchers to estimate interaction effects between group-level treatments and individual-level covariates. However, if the effect of a group-level variable varies with individual-level variables, the effect of $x_{2j}$ on the intercept no longer represents an average effect. Instead, it reflects the effect evaluated at $x_{1ij}=0$, which may not be meaningful or of interest. 
To address this issue, researchers would need to estimate all relevant interaction effects and combine them appropriately. This approach has not been discussed in CLP, nor has it been implemented in their simulations or in any applications of their estimator.

 \section{Traditional Quantile Panel Data Estimators}\label{sec:FE+BE+RE}
 \subsection{Fixed Effects, Random Effects and Between Estimators}\label{subsec:FE+BE+RE}

In this subsection, we apply our results to derive quantile analogs of the fixed effects, between, random effects, and Hausman-Taylor estimators. As discussed in Section \ref{subsec:least squares estimators}, we can obtain these estimators by selecting appropriate instrumental variables in the second stage. For fixed effects estimation, model (\ref{eq:model}) implies that $\dot{x}_{1ij}$ is a valid instrument since it varies only within groups and is uncorrelated with the group effects. This instrument automatically satisfies Assumption \ref{a:instruments}.
In this special case, the approach corresponds to the traditional MD estimator, where all variance originates in the first stage. 

In the first stage, $\beta(\tau)$ is estimated separately for each group. The second stage then averages these group-level coefficients using weights proportional to $\tilde X_j'\tilde X_j$ (see equation \ref{eq:MD representation}). \cite{Galvao2015} propose an alternative approach, suggesting efficient weights proportional to the inverse of the variance of the first-stage estimators. Their weights are equivalent to ours when \begin{equation}\label{eq:efficient fe}
f_{y_{ij}|x_{1ij},v_j}(Q(\tau,y_{ij}|x_{1ij},v)|x,v) = f_{y_{ij}|x_{1ij},v_j}(Q(\tau,y_{ij}|x_{1ij},v')|x,v')
\end{equation}
for any $v$ and $v'$, i.e., when the conditional distribution of the group effects is the same across groups. Outside this specific case, our estimator may be less efficient but avoids the need to estimate the first-stage variance, which depends on the conditional densities in equation (\ref{eq:efficient fe}).  A third approach is to take the unweighted average of $\hat\beta_j(\tau)$, which, while not efficient when $\beta_j(\tau)$ are homogeneous, remains straightforward to interpret even if the model is misspecified.

We can implement a quantile between estimator using $\bar{x}_{1j}$ as an instrument to exploit only the variation across groups. On the other hand, combining within and between variations is more complex for quantile models than for least squares models. Applying quantile regression to the whole population without controlling for groups identifies parameters that differ from our intended parameters (see Remark \ref{remark:conditional}). Using our MD estimator with $x_{1ij}$ as an instrument, which corresponds to using OLS in the second stage, consistently estimates $\beta(\tau)$ but only at the slow $\sqrt{m}$ rate because $x_{1ij}$ also varies between groups. The same occurs if both $\dot x_{1ij}$ and $\bar{x}_{1j}$ are combined with 2SLS, as the weights attributed to $\bar{x}_{1j}$ do not vanish asymptotically. Instead, we propose two efficient random effects estimators: an efficient GMM estimator and one using optimal instruments. It is worth highlighting that these estimators not only reduce the asymptotic variance of the estimator but also increase the rate of convergence from $\sqrt m$ to $\sqrt{mn}$ compared to 2SLS.
 

Given the first-stage estimation, we have the following moment condition:
\begin{equation}\label{eq:moment eq} 
\bE[g_{j}(\delta(\tau), \tau)] = \bE[ Z_{j}' ( \tilde X_{j} \hat \beta_j(\tau) - X_{j} \delta(\tau) ] = 0. 
\end{equation}
When the instrument includes both the within-group variation $\dot x_{1ij}$ and the between-group average $\bar x_{1j}$, the efficient GMM estimator will optimally combine these two sources of variation. The weighting matrix is computed as shown in equation (\ref{eq:estimator omega}). According to Propositions \ref{p:weight matrix adaptive} and Theorem \ref{t:first-order}, this estimator is guaranteed to be both $\sqrt{mn}$ consistent and efficient. This random effects estimator has the same first-order asymptotic distribution as the fixed effects estimator that uses only the within-group variation $\dot x_{1ij}$. However, the random effects estimator is expected to have a lower variance in finite samples because it also incorporates the between-group variation. As the number of observations $n$ increases, the influence of the between-group variation diminishes, causing the random effects estimator to converge to the fixed effects estimator. This behavior is similar to what is observed in least squares models, as discussed by \cite{Baltagi2021} and \cite{Ahn2014}.
 
If we impose the stronger assumption that the moment restriction in equation (\ref{eq:moment eq}) holds conditionally on $Z_j$, we can use the theory of optimal instruments to derive a more efficient random effects estimator. Optimal instruments are relevant when a researcher has a conditional moment restriction of the form $\bE[g_{j}(\delta, \tau) | Z_j] = 0$. When a moment condition holds conditional on $Z_j$, an infinite set of valid moments exist, and one could use additional moments to increase efficiency.
The goal is to select the instrument that minimizes the asymptotic variance, which takes the form $Z^*_j = \bE [g_j(\delta, \tau)g_j(\delta, \tau)' | Z_j] ^{-1} R_j(\delta, \tau)$, with $R_j(\delta, \tau) = \bE [ \frac{\partial}{\partial \delta} g_j(\delta, \tau) | Z_j] $ (see, e.g., \citealp{Chamberlain1987} and \citealp{Newey1993}).
To implement the random effect estimator with optimal instruments, we set $Z_j = X_j$. Under the additional assumption that $\bE[\alpha_j^2(\tau) | X_j] = \sigma_\alpha^2(\tau)$,\footnote{We assume homoskedasticity of $\alpha_j(\tau)$ to obtain a simple estimator, similar to the classical least squares random effects estimator. If we were to drop this assumption, we would need to estimate $\bE[\alpha_j(\tau)|X_j]$. Note that we do not assume homoscedasticity in the group-level model such that this assumption does not constrain the heterogeneity of $\beta(\tau)$ across different values of $\tau$.} the optimal instrument simplifies to 
\begin{equation}\label{eq:optimalinstrument_re}
    Z_j^*(\tau) =  \left ( \tilde X_j  \frac{V_j(\tau)}{n} \tilde X_j' + \mathbf l_n'\mathbf l_n  \sigma_\alpha^2(\tau) \right )^+  X_j,
\end{equation}
where $ V_j(\tau) $ is the asymptotic variance from the first stage for a group $j$, $\mathbf l_n$ is a $n$-dimenstional vector of ones, and $^+$ denotes the Moore-Penrose inverse.\footnote{Since the matrix $( \tilde X_j  \frac{V_j(\tau)}{n}\tilde X_j' + \mathbf l_n'\mathbf l_n \sigma_\alpha^2(\tau) )$ is singular, we use the Moore-Penrose inverse.} 

A few remarks about the optimal instruments follow. First, under standard random effects assumptions, the optimal instrument applied to least squares models is numerically identical to the FGLS estimator. 
Second, the optimal instrument depends on $n$ analogously to the efficient weighting matrix of the GMM estimator. As $n$ increases, the first stage variance converges to zero, and the generalized inverse will give infinitely more weights to the within variation and asymptotically converge to the fixed effects estimator.
Third, if $\sigma_{\alpha}(\tau) = 0$, then all the variance arises in the first stage, and this estimator is identical to the efficient MD estimator (see Proposition \ref{prop:emd = giv} in Appendix \ref{app:opt inst MD}).

\subsection{Hausman and Taylor Model}
The Hausman-Taylor model provides a method for finding instrumental variables within the model itself. It represents a middle ground between the random effects model, which assumes orthogonality between $\alpha_j(\tau)$ and $x_{ij}$, and the fixed effects model, which only identifies the effect of individual-level variables. To estimate the effect of group-level variables, \cite{Hausman1981} assume that some elements of $x_{1ij}$ are uncorrelated with $\alpha_j(\tau)$. We consider model (\ref{eq:model}) but partition the regressors into four types of variables, $x_{ij} = [x_{1ij}^{ex} \ x_{1ij}^{en} \ x_{2j}^{ex} \ x_{2j}^{en} ]$, where the superscript $ex$ indicates exogenous variables and the superscript $en$ indicates potentially endogenous variables. Thus,
\begin{align*}
   & \bE[x_{1ij}^{ex} \alpha_j(\tau) ] = 0,\\
   & \bE[x_{2j}^{ex} \alpha_j(\tau) ] = 0.
\end{align*}

The assumptions imply that we can estimate $\delta(\tau)$ using the instrument $z_{ij} = (\dot x_{1ij}^{ex}, \dot x_{1ij}^{en},$ $ \bar x_{1ij}^{ex}, x_{2j}^{ex})$. While $x_{2j}^{en}$ is potentially endogenous, the within variation is uncorrelated with $\alpha_j(\tau)$ as it varies only within $j$. Identification requires at least as many instruments as parameters to estimate, hence $dim(x_{1ij}^{ex}) \geq dim(x_{2j}^{en})$.  In overidentified models, efficient GMM can be implemented, and if conditional moment restrictions are available, optimal instruments can be used. However, implementing optimal instruments is not straightforward, as it requires estimating  
$\bE [x_{ij} |z_{ij}]$, typically done nonparametrically (see \citealp{Newey1993}). We do not contribute to this aspect in this paper. 

\subsection{Hausman Test}\label{s:hausman}

The random effects estimator's consistency relies on stronger orthogonality conditions compared to the fixed effects estimator. Under these stronger assumptions, both estimators are consistent, but the fixed effects estimator is inefficient. \citet{Hausman1978} proposed a test for the null hypothesis of random effects against the alternative of fixed effects. Various generalizations of the Hausman test have been suggested in the literature (e.g., \citet{Chamberlain1982, Mundlak1978, Wooldridge2019}). \citet{Arellano1993} considers a heteroskedasticity and autocorrelation robust generalization based on a Wald test. \citet{Ahn1996} propose a GMM test based on a 3SLS regression as an equivalent method for the Hausman test.

This subsection explains how we can use the overidentification test presented in Section \ref{sec:asym. theory} as a quantile version of the Hausman test for our two-step estimator. The assumption of correct specification of the first stage is maintained under both the null and alternative hypotheses. Compared to the fixed effects estimator, consistency of the random effects estimator additionally requires that $x_{1ij}$ is uncorrelated with $\alpha_j(\tau)$, so that $\mathbb{E}[\dot{x}_{1ij}' \alpha_j(\tau)] = 0$ and $\mathbb{E}[\bar{x}_{1j}' \alpha_j(\tau)] = 0$ are valid moment conditions. By contrast, the fixed effects estimator relies only on the moment condition $\mathbb{E}[\dot{x}_{1ij}' \alpha_j(\tau)] = 0$. Consequently, the overidentification test suggested in Proposition \ref{prop:hausmant test} can be used as a test of the random effects orthogonality conditions. Compared to the traditional Hausman test, our test does not rely on the assumption of conditional homoskedasticity of the errors and is robust to clustering.

\citet{Galvao2019} also note that a quantile regression with all observations and fixed effects quantile regression identify different parameters. They propose a Hausman test based on an auxiliary quantile regression model that incorporates both $x_{1ij}$ and $\bar{x}_{1j}$ as regressors, and tests for the significance of $\bar{x}_{1j}$. However, this test differs from our approach in several ways: it starts from a different model (not conditional on the groups), relies on the correct specification of the auxiliary regression, and does not directly compare random effects and fixed effects estimators.


\subsection{Simulations}\label{sec:simulationsFE+BE+RE}
\input{Tables/panel_simul}
\input{Tables/panel_coverage}

\input{Tables/hausman_test}
This section presents simulation results for the different panel data estimators and the Hausman-type test presented in the previous subsections. These simulations focus on the estimation of $\beta(\tau)$. We consider the following data-generating process  
\begin{align}
    y_{ij} = x_{1ij} + \alpha_j + (1 + 0.1 x_{1ij})\nu_{ij}.
\end{align}
where all variables are scalars, $\nu_{ij} \sim \mathcal{N}(0,1)$, and $x_{1ij} = h_j + 0.5u_{ij}$, with $u_{ij} \sim \mathcal{N}(0,1)$ and 
\begin{equation*}
    \begin{pmatrix}
    h_j \\ \alpha_j 
    \end{pmatrix} \sim \mathcal{N} \begin{pmatrix}
    
    \begin{pmatrix}
    0 \\ 0
    \end{pmatrix}, \begin{pmatrix}
    1 & \lambda \\ \lambda & 1
    \end{pmatrix} \end{pmatrix}.
\end{equation*}
If $\lambda \neq 0$, $x_{1ij}$ is correlated with $\alpha_j$. For the simulation of the panel data estimators, we let $\lambda = 0$ so that all estimators are consistent. In contrast, in the Monte Carlo study of the Hausman test, we set $\lambda = \{0,0.1, 0.2,0.3,0.4 \}$.  The true coefficient takes the values $\beta(\tau) = 1 + 0.1 F^{-1}(\tau)$ where $F$ is the standard normal CDF. 
We consider the samples with $n = \{10,25, 200 \}$ and $m = \{ 25, 200\}$ and focus on the set of quantiles $\mathcal{T} = \{0.1, 0.5, 0.9\}$. All simulation results are based on 10,000 replications.

We compare the performance of four MD estimators, all based on the same first-stage regression but differing in their choice of instruments in the second stage. The `Pooled' estimator uses $x_{1ij}$ as the instrument,\footnote{We refer to this estimator as `Pooled' because the second stage involves a pooled OLS regression. It should not be confused with the one-step pooled quantile regression, which does not account for group structure.} the BE estimator uses $\bar x_{1j}$ as the instrument, the RE-GMM estimator combines $\dot x_{1ij}$ and $\bar x_{1j}$ using efficient GMM, and the RE-OI estimator combines the same exogenous variation using the single optimal instrument defined in equation (\ref{eq:optimalinstrument_re}). We use the estimator of  \cite{Powell1991} for $V_j(\tau)$ and the estimator of \cite{Nerlove1971} for $\sigma^2_\alpha(\tau)$. 


Table \ref{tab:panel_bias} shows that the estimators perform well also when both $m$ and $n$ are small. The RE-GMM estimator performs similarly to the RE-OI estimator except with small $n$ when the RE-GMM estimator outperforms the RE-OI both in terms of bias and variance. As expected, the RE-GMM, the RE-OI, and the fixed effects (FE) estimators become indistinguishable as $n$ increases. Whereas with small $n$, there is an apparent gain in using a random effects estimator. From the standard deviations, it is possible to see the different rates of convergence of the estimators. The precision of the fixed effects and random effects estimators increases in similar magnitude when $m$ or $n$ increases. In contrast, the standard deviation of the pooled and between estimators decreases only when $m$ increases. The pooled and the between estimators have the smallest bias but, in most cases, also the largest variance.

The coverage probabilities of the 95\% confidence intervals are provided in Table \ref{tab:panel_coverage}. The confidence intervals of the pooled and the fixed effects estimator perform well in all sample sizes considered. On the other hand, the confidence bands of the random effects estimators slightly undercover the true parameter mostly when $n$ is small. In larger samples, all the coverage probabilities are close to the theoretical level.

Table \ref{tab:hausmantest} shows the rejection probabilities of the overidentification test for different values of $\lambda$. When $\lambda = 0$, the $\mathbb{H}_0$ is satisfied, so we should reject the null at a rate close to $5\%$. If $\lambda  \neq 0$, $ x_{ij}$ is correlated with $\alpha_j$, and some moment conditions used by the RE-GMM estimator are not valid. In this case, higher rejection probabilities suggest a more powerful test. 
The first column shows that the empirical sizes of the test are close to the theoretical levels in most sample sizes. Columns 2-5 show that, as expected, the power of the test is higher in large samples and increases with the correlation between $\bar x_{1j}$ and the unobserved heterogeneity $\alpha_j$. An increase in $m$ substantially improves the power of the test, while a larger number of time periods $n$ improves the results to a lesser extent. In general, the test performs better both in terms of size and power when $m$ is large, which is most often the case in empirical applications. Even if the random effects estimator converges to the fixed effects estimator as $n$ increases and the random effects estimator of $\beta$ will be consistent even if $\lambda \neq 0$, the size and power of the test do not deteriorate. This result is consistent with the findings in \cite{Ahn2014}.


\section{Empirical Application: The Effect of the Food Stamps Program on Birth Weight}\label{sec:food stamp}

In this section, we apply our minimum distance approach to estimate the impact of the food stamp program on the birth weight distribution using grouped data. We complement the analysis of \cite{Almond2011} by providing distributional effects. 
Food stamps constitute an important means-tested program that gives entitled households coupons they can redeem at approved retail food stores. The Food Stamp Act (FSA) was introduced in 1964 and enabled counties to start their own federally funded food stamp program (FSP). In the subsequent years, counties increasingly adopted such programs, and in 1973, an amendment to the FSA required all counties to establish a FSP by 1975. Thus, the share of counties with an FSP increased steadily from 1964 to 1974, and identification exploits the variation in the timing of the adoption across counties. \cite{Almond2011} use data from 1968 (when about 40\% of the counties had introduced the program) to 1977  (two years after the FSP was implemented everywhere) to analyze the effect of the program. 

Given the negative consequences of low birth weight, besides estimating the effect of the policy on average weight, \cite{Almond2011} estimate the effect on the probability that birth weight falls below a certain threshold. As discussed in \cite{Melly2015a}, this procedure leads to biased results unless there is no time effect or group effect or the outcome is uniformly distributed. 

 In this section, we use the subscripts $i$, $c$, and $t$ to denote the birth, the county, and the trimester of birth, respectively.\footnote{Using the same notation as in the paper, the $j$ units are county-trimester combinations, and the $i$ units index individual births within a county in a given trimester. In this section, we use three subscripts for clarity.} The variable of interest is a binary variable that is coded 1 if there was a food stamp program in place three months before birth. Hence, the treatment is assigned to county-month cells, and in around 1\% of cases, it also varies within groups. 

 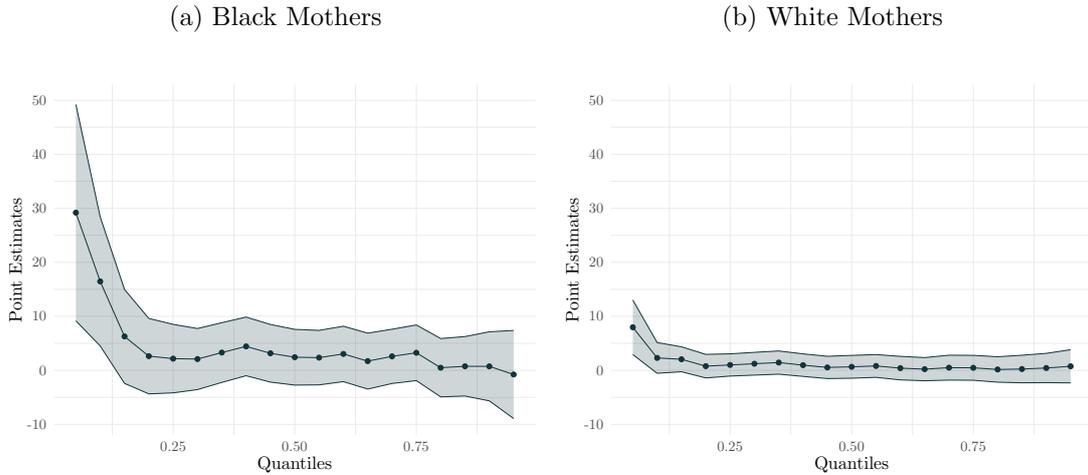
\begin{figure}
\centering
  \caption{Impact of Food Stamp Introduction on the Distribution of Birth Weight}  \label{fig:fsp results}  
\begin{subfigure}{0.45\textwidth}
 \caption{Black Mothers}
\resizebox{1.0\linewidth}{!}{%
  \input{Figures/black}}
\end{subfigure}
\begin{subfigure}{0.45\textwidth}
    \caption{White Mothers}
\resizebox{1.0\linewidth}{!}{%
  \input{Figures/white}}
\end{subfigure}
  \floatfoot{The figure shows the impact of the food stamp introduction on the conditional distribution of birth weight. The panels show point estimates and 95\% confidence bands (shaded area) computed using standard errors clustered at the county level. The panel on the left (right) shows the effects for blacks (whites). The regressions include county, time, and state-year fixed effects.}
\end{figure}

We consider the following model separately for black and white mothers:
\begin{equation}
Q(\tau , bw_{ict} | fsp_{ct}, x_{1ict}, x_{2ct} , v_{ct})  = fsp_{ct} \gamma_1(\tau) + x_{1ict} \beta(\tau) + x_{2ct} \gamma_2(\tau) + \alpha(\tau, v_{ct}),
\end{equation}
where $Q(\tau, bw_{ict} | fsp_{ct}, x_{1ict}, x_{2ct}, v_{ct}) $ is the $\tau$th conditional quantile function of the outcome given all the variables. The variable $fsp_{ct}$ indicates whether there is a food stamp program in place, $x_{1ict}$ are variables related to the individual births, such as gender, mother age, and its square as well as the legitimacy status of the birth. Group-level control variables $x_{2ct}$ include annual county-level controls (real per capita income, government transfers to individuals, medical spending, and retirement and disability payments) and 1960 county-level characteristics (county population and the shares of urban population, black population, and of farmland) interacted with a linear time trend. Further, $x_{2ct}$ also includes county, state-year, and time fixed effects. 

For the estimation, we drop groups that have less than 25 degrees of freedom.\footnote{If there are $K_1$ individual level variables we drop all groups with less than $K_1 + 1 + 25$ observations. Since some variables might vary at the individual level in some groups only, this threshold is group-specific.} The estimations are performed using a sample of 2,822,091 individual observations divided into 19,482 groups for blacks and 16,038,235 individual births divided into 80,289 groups for whites.\footnote{We have a different number of groups compared to \cite{Almond2011} due to multiple reasons. First, they give higher weights to births in groups where only 50\% of the births are coded in the natality data; thus, when they drop groups with less than 25 births, the number of births in these groups is inflated. Second, since they take the group average, they keep births with missing values for birth weight. We drop those births as we work with individual-level data.} 
Figure \ref{fig:fsp results} illustrates the results. The results for black mothers are in the left panel, while the results for white mothers are in the right panel. The effect is substantially larger among black mothers. The results suggest a positive effect of the food stamp program on the lower tail of the conditional distribution. For instance, the estimates suggest that the food stamp program is associated with an increase in birth weight by almost 30 grams for blacks at the 5th percentile of the conditional distribution. For whites, there seems to be an effect only in the left tail of the distribution, and the effects are small. For blacks, the coefficients are large in the left tail and remain positive, albeit of small magnitude, until the 75\% percentile. However, for higher quantiles, the effects are not statistically different from zero. 

To compare our MD estimator with the CLP estimator on observational data, we also perform the estimation using the CLP procedure. For this comparison, we focus on the sample of black mothers. Before implementing the CLP estimator on this data, we need to solve a few issues that affect the CLP estimator but not the MD estimator.
First, the treatment variable exhibits within-group variation in around 1\% of the groups. If there are variables that vary only within some groups, the intercept is not comparable over groups, thus creating a problem for the CLP estimator. Hence, for the comparison, we drop groups for which the variable $fsp$ is varying within the group.
A second but similar issue is due to the covariate controlling for the legitimacy status of the birth, as there are groups that do not contain all levels of the categorical variable. To solve this issue, we drop the variable from the conditioning set.
Third, for one observation (out of over 2 million), one of the group-level control variables is minimally different from the value assigned to other members of the group. This is most likely due to a coding error. In this case, we replace the values with the group average. While these features of the data do not pose any problem with our estimation approach, they can lead to substantially different results using the CLP estimation method.
\begin{figure}
    \centering
\resizebox{.8\columnwidth}{!}{%
\input{Figures/plot_comparison_clp_md}}
    \caption{ Impact of Food Stamp Introduction estimated with the MD and CLP estimators}
    \label{fig:empirical_application_comparison_clp_md}
      \floatfoot{The figure shows the impact of the food stamp introduction on the conditional distribution of birth weight in the sample of black mothers computed with the MD (in black) and the CLP (in blue) estimators. The sample consists of 2,787,509 divided into 17,141 groups. The panels show point estimates and 95\% confidence bands (shaded area) computed using standard errors clustered at the county level. For this estimation, we drop groups for which the treatment variables vary within groups and do not control for the legitimacy status of the birth.}
\end{figure}
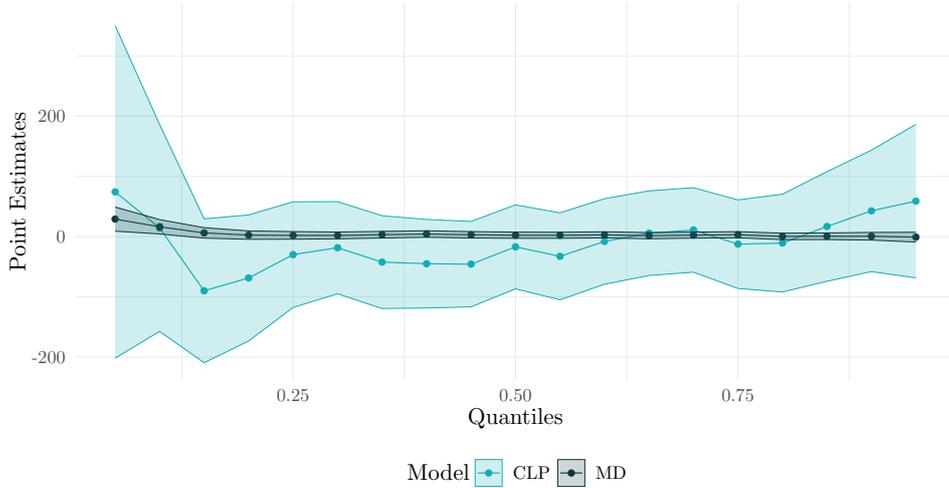
Figure \ref{fig:empirical_application_comparison_clp_md} shows the estimation results performed with our approach compared to the CLP approach using the sample of black mothers. Given that we estimate a slightly different model on a marginally different sample, we re-estimate the effects using the MD estimator to ensure a meaningful comparison. 
The results obtained using the MD estimator remain almost identical to those in Figure \ref{fig:fsp results}. On the other hand, using the CLP estimator, we obtain substantially different results from those attained using the MD estimator, and the confidence intervals are, on average, 14 times larger, precluding any possibility of drawing informative conclusions. Further, the figure shows that the MD point estimates and confidence intervals are always inside the confidence intervals of the CLP estimator.





\section{Conclusion}\label{sec:conclusion}

To summarize, our paper makes the following key contributions:
First, we propose a new class of minimum distance estimators for quantile panel data models, applicable to both classical panel data settings, where units are observed over time, and grouped data settings, where individuals are divided into groups and treatment varies at the group level. This class of estimators provides quantile analogs of established panel data methods, including fixed effects, random effects, between, and Hausman-Taylor estimators. Additionally, it substantially improves upon the existing grouped instrumental variable estimator of \cite{Chetverikov2016} and remains computationally fast.
Second, we establish uniform asymptotic properties under a general framework accommodating arbitrary degrees of group-level heterogeneity. We also introduce adaptive inference procedures that remain valid regardless of whether group effects are zero, bounded, or vanish at arbitrary rates. The adaptive estimator of the variance of the moments leads to a GMM estimator that is uniformly efficient across unknown relative convergence rates. 
Third, we demonstrate the practical relevance of our approach in an empirical application studying the impact of the Food Stamp Program on birth weight. Our results indicate that the policy’s positive effects are concentrated almost entirely in the lower tail of the birth weight distribution, particularly among black mothers.

While these advancements mark significant progress, our approach has limitations that invite further exploration. First, although our Monte Carlo simulations show that our estimator and suggested standard errors perform well in finite samples, our asymptotic framework relies on both the number of groups and the number of observations per group diverging to infinity. This assumption is reasonable in settings with large groups, such as our empirical application, but it is often violated in classical panel data contexts. Future work should explore finite $n$ inference or consider stronger assumptions for identification in such cases. Second, our estimator accommodates parametric time trends but not unrestricted time effects. Estimating time effects requires working with all units simultaneously, making a group-by-group computation strategy infeasible. Finally, while we focus on within-group heterogeneity, many applications also require understanding between-group heterogeneity. \cite{Pons2024} develops a framework that explicitly captures both, providing a promising avenue for further research.

\bibliographystyle{ecta}
\bibliography{references.bib}
\newpage

\appendix

\section{Proofs}\label{app: proofs asy. results}

\subsection{Proof of Lemma \ref{l:sampling error}: Sampling error}

\begin{proof}[Proof of lemma \ref{l:sampling error}]\label{p:sampling error}
Starting from the definition of the estimator we obtain 
\begin{align*}
\hat \delta (\tau) &= \left( X' Z \hat W(\tau) Z' X \right) ^{-1} X' Z \hat W(\tau) Z' \hat Y(\tau)\\
&=\left( S_{ZX}' \hat W(\tau) S_{ZX} \right) ^{-1} S_{ZX}' \hat W(\tau) \frac{1}{mn}\sum_{j = 1}^m\sum_{i = 1}^n z_{ij}\tilde x_{ij}' \hat \beta_j(\tau)\\
&=\left( S_{ZX}' \hat W(\tau) S_{ZX} \right) ^{-1} S_{ZX}' \hat W(\tau) \frac{1}{mn}\sum_{j = 1}^m\sum_{i = 1}^n z_{ij}\left(\tilde x_{ij}'\left(\hat \beta_j(\tau)-\beta_j(\tau)\right)+\tilde x_{ij}'\beta_j(\tau)\right)\\
&=\left( S_{ZX}' \hat W(\tau) S_{ZX} \right) ^{-1} S_{ZX}' \hat W(\tau) \frac{1}{mn}\sum_{j = 1}^m\sum_{i = 1}^n z_{ij}\left( \tilde x_{ij}'\left(\hat \beta_j(\tau)-\beta_j(\tau)\right)+ x_{ij}'\delta(\tau)+\alpha_j(\tau)\right)\\
&=\delta(\tau)+\left( S_{ZX}' \hat W(\tau) S_{ZX} \right) ^{-1} S_{ZX}' \hat W(\tau) \frac{1}{mn}\sum_{j = 1}^m\sum_{i = 1}^n z_{ij}\left( \tilde x_{ij}'\left(\hat \beta_j(\tau)-\beta_j(\tau)\right)+\alpha_j(\tau)\right).
\end{align*}
\end{proof}

\subsection{Proof of Theorem \ref{t:consistency}: Uniform consistency}

\subsubsection{Lemma \ref{l:uniform coefficients}}

As a preliminary step to prove the uniform consistency of our estimator, we show uniform (in $\tau$ and $j$) consistency of the group-level quantile regressions.
\begin{lemma}[\textbf{Uniform consistency of $\hat\beta_j(\tau)$}]
\label{l:uniform coefficients}
Under Assumptions \ref{a:sampling}-\ref{a:bounded density} and \ref{a:growth condition}(a), we have $$\underset{\tau\in\mathcal T} {\sup}\,\underset{1\leq j\leq m}{\max} \lVert \hat \beta_j(\tau)-\beta_j (\tau) \rVert = o_p(1).$$
\end{lemma}
\begin{proof}[Proof of Lemma \ref{l:uniform coefficients}]\label{p:uniform coefficients}
\cite{Angrist2006a} show uniform consistency of the quantile regression estimator in $\tau$ but not in $j$ (see their Theorem 3) while \cite{Galvao2015} show uniform consistency in $j$ but not in $\tau$ (see their Lemma 1). We show uniformity in both dimensions by following the steps of the proof in \cite{Galvao2015} and extending it. 

We define $\mathbb Q_{nj}(\beta, \tau)=\frac{1}{n}\sum_{i = 1}^n \rho_\tau(y_{ij}-\tilde x_{ij}'\beta) - \rho_\tau(y_{ij}-\tilde x_{ij}'\beta_j(\tau))$ and $Q_{j}(\beta,\tau)=\bE_{i|j}[ \rho_\tau(y_{ij}-\tilde x_{ij}'\beta) - \rho_\tau(y_{ij}-\tilde x_{ij}'\beta_j(\tau))]$. As shown in \cite{Angrist2006a}, the empirical processes $(\beta, \tau)\mapsto \mathbb Q_{nj}(\beta,\tau)$ for all groups $j$ are stochastically equicontinuous because $\left|\mathbb Q_{nj}(\beta',\tau')-\mathbb Q_{nj}(\beta'',\tau'') \right|\leq C_{1}\cdot |\tau'-\tau''|+C_{2}\cdot \lVert\beta'-\beta''\rVert$ where $C_{1}=2\cdot C\cdot \underset{\beta \in \mathcal{B}}{\sup}\lVert \beta \rVert$  for any compact set $\mathcal{B}$ and $C_{2}=2\cdot C$. The constant $C$ is defined in Assumption \ref{a:covariates}. Note that $C_{1}$ and $C_{2}$ are neither functions of $j$ nor $\tau$. 

Fix any $\zeta>0$. Let $B_j(\zeta,\tau)=\{\beta:\lVert\beta-\beta_j(\tau)\rVert \leq \zeta\}$, the ball with center $\beta_j(\tau)$ and radius $\zeta$. For each $\beta\notin B_j(\zeta,\tau)$, define $\tilde \beta=r_j\beta+ (1-r_j)\beta_j(\tau)$ where $r_j=\frac{\zeta}{\lVert \beta-\beta_j(\tau)\rVert}$. So $\tilde \beta\in \partial B_j(\zeta,\tau)=\{\beta:\lVert\beta-\beta_j(\tau)\rVert=\zeta\}$, the boundary of $B_j(\zeta,\tau)$. Since $\mathbb Q_{nj}(\beta,\tau)$ is convex in $\beta$ for all $\tau$, and $\mathbb Q_{nj}(\beta_j(\tau),\tau)=0$, we have
\begin{align} \label{eq:ineq_convexity}
r_j \mathbb Q _{nj}(\beta,\tau) \geq\mathbb Q_{nj}(\tilde \beta,\tau)=Q_{j}(\tilde\beta,\tau)+\mathbb Q_{nj}(\tilde\beta,\tau)-Q_{j}(\tilde\beta,\tau)>\epsilon_\zeta+\mathbb Q_{nj}(\tilde\beta,\tau)-Q_j(\tilde\beta,\tau)
\end{align}
uniformly in $j$ and $\tau$, where 
\begin{equation*}
\epsilon_\zeta=\underset{\tau\in\mathcal T}{\inf}\;\underset{1\leq j\leq m}{\inf}\;\underset{\lVert\beta-\beta_j(\tau)\rVert=\zeta}{\inf} \bE_{i|j}\left[\int_0^{\tilde x_{ij}'(\beta-\beta_j(\tau))}\left(1(y_{ij}-\tilde x_{ij}'\beta_j(\tau)\leq s)-1(y_{ij}-\tilde x_{ij}'\beta_j(\tau)\leq 0)\right)ds\right]
\end{equation*}
by the identity of \cite{Knight1998} and $\epsilon_\zeta>0$ by Assumptions \ref{a:conditional distribution} and \ref{a:bounded density}.
Thus, we have that
\begin{align*}
\left\{\underset{\tau \in \mathcal{T}}{\sup}\,\underset{1\leq j\leq m}{\max} \lVert\hat\beta_j(\tau)-\beta_j(\tau)\rVert>\zeta\right\} &\overset{(a)}{\subseteq}\{\exists \tau_j\in\mathcal T,\exists\beta_j\notin B_j(\zeta,\tau_j) :\mathbb Q_{nj}(\beta_j,\tau_j)\leq 0\}\\
&\overset{(b)}{\subseteq} \cup_{j=1}^m \left\{\underset{\tau\in\mathcal T}{\sup}\,\underset{\beta_j\in B_j(\zeta,\tau_j)}{\sup} |\mathbb Q_{nj}(\beta_j,\tau_j)-Q_j(\beta_j,\tau_j)| \geq\epsilon_\zeta\right\}.
\end{align*}
Relation (a) holds because, by definition, $\hat\beta_j(\tau)$ minimizes $\mathbb Q_{nj}(\beta,\tau)$, and $\mathbb Q_{nj}(\beta_j(\tau),\tau)=0$. Relation (b) holds by the rightmost inequality of line (\ref{eq:ineq_convexity}). Then, it follows that
\begin{align*}
\Pr\left\{\underset{\tau\in\mathcal T}{\sup}\, \underset{1\leq j\leq m}{\max} \lVert \hat \beta _j(\tau) -\beta_j(\tau)\rVert > \zeta\right\}
&\leq \Pr \left\{\cup_{j=1}^m\left\{\underset{\tau\in\mathcal T}{\sup}\, \underset{\beta_j\in B_j(\zeta,\tau)}{\sup} |\mathbb Q_{nj}(\beta_j,\tau)-Q_j(\beta_j,\tau)| \geq \epsilon_\zeta \right\} \right\}\\
&\leq \sum_{j = 1}^m \Pr\left\{ \underset{\tau \in \mathcal T}{\sup} \,\underset{\beta_j\in B_j(\zeta,\tau)}{\sup} |\mathbb Q_{nj}(\beta_j,\tau)-Q_j(\beta_j,\tau)| \geq\epsilon_\zeta\right\}\\
&\leq m \underset{1\leq j\leq m}{\max} \Pr\left\{ \underset{\tau\in\mathcal T}{\sup}\,\underset{\beta_j\in B_j(\zeta,\tau)}{\sup} |\mathbb Q_{nj}(\beta_j,\tau)-Q_j(\beta_j,\tau) |\geq \epsilon_\zeta \right\}
\end{align*}
Therefore, if we can show that 
\begin{equation*}
\underset{1\leq j\leq m}{\max} \Pr\left\{ \underset{\tau\in\mathcal T}{\sup}\,\underset{\beta_j\in B_j(\zeta,\tau)}{\sup} |\mathbb Q_{nj}(\beta_j,\tau)-Q_j(\beta_j,\tau) |\geq \epsilon_\zeta \right\}= o\left(\frac{1}{m}\right)\end{equation*}
the proof of the lemma will be completed.

Without loss of generality, we assume $\beta_j(\tau)=0$ for all $j$ and $\tau\in\mathcal T$. Then the balls $B_j(\zeta,\tau)$ for all $j$ and $\tau\in\mathcal T$ are identical and we denote them by $B(\delta)$. Because the closed ball $B(\zeta)$ is compact, there exist $K$ balls with center $\beta^k$, $k=1,...,K$, and radius $\frac{\epsilon_\zeta}{3C_2}$ such that the collection of them covers $B(\zeta)$. Since $\epsilon_\zeta>0$, we can find a finite $K$ that satisfies this condition and is independent of $j$ and $\tau$. Therefore, for any $\beta\in B(\delta)$, there is some $k\in \{1,...,K\}$ such that
\begin{align*}
|\mathbb Q_{nj}(\beta,\tau)-Q_j(\beta,\tau)|-|\mathbb Q_{nj}(\beta^k,\tau)-Q_j(\beta^k,\tau)|&\leq |\mathbb Q_{nj}(\beta,\tau)-Q_j(\beta,\tau)-\mathbb Q_{nj}(\beta^k,\tau)+Q_j(\beta^k,\tau)|\\
&\leq|\mathbb Q_{nj}(\beta,\tau)-\mathbb Q_{nj}(\beta^k,\tau)|+|Q_j(\beta,\tau)-Q_j(\beta^k,\tau)|\\
&\leq C_2 \frac{\epsilon_\zeta}{3C_2}+C_2 \frac{\epsilon_\zeta}{3C_2}=\frac{2\epsilon_\zeta}{3}
\end{align*}
uniformly in $j$ and $\tau\in\mathcal T$. Note that the third line is justified by the stochastic equicontinuity of $\mathbb Q_{nj}(\beta,\tau)$.

It then follows that, $$\underset{\tau\in\mathcal T}{\sup}\, \underset{\beta\in B(\zeta)}{\sup} |\mathbb Q_{nj}(\beta,\tau)-Q_j(\beta,\tau)|\leq \underset{\tau\in\mathcal T}{\sup}\,\underset{1\leq k\leq K}{\max}|\mathbb Q_{nj}(\beta^k,\tau)-Q_j(\beta^k,\tau)| +\frac{2\epsilon_\zeta}{3},$$
and 
\begin{align*}
\Pr\left\{\underset{\tau\in\mathcal T}{\sup}\,\underset{\beta\in B(\delta)}{\sup} |\mathbb Q_{nj}(\beta,\tau)-Q_j(\beta,\tau)>\epsilon_\zeta\right\}&\leq\Pr\left\{\underset{\tau\in\mathcal T}{\sup} \,\underset{1 \leq k \leq K} {\max} |\mathbb Q_{nj}(\beta^k,\tau)-Q_j(\beta^k,\tau)|+\frac{2\epsilon_\zeta}{3} > \epsilon_\zeta \right\}\\
&=\Pr\left\{\underset{\tau\in\mathcal T}{\sup} \,\underset{1 \leq k \leq K} {\max} |\mathbb Q_{nj}(\beta^k,\tau)-Q_j(\beta^k,\tau)|>\frac{\epsilon_\zeta}{3}\right\}\\
&\leq \underset{\tau\in\mathcal T}{\sup} \sum_{k=1}^K \Pr \left\{|\mathbb Q_{nj}(\beta^k,\tau)-Q_j(\beta^k,\tau)|>\frac{\epsilon_\zeta}{3} \right\}.
\end{align*}
For each $\tau\in\mathcal T$, $\mathbb Q_{nj}(\beta^k,\tau)$ is the sample mean of $n$ i.i.d. terms bounded in absolute values by $2\cdot C\cdot \zeta$. By Hoeffding’s inequality, it follows that 
\begin{align*}\sum_{k=1}^K \Pr \left\{|\mathbb Q_{nj}(\beta^k,\tau)-Q_j(\beta^k,\tau)|>\frac{\epsilon_\zeta}{3}\right\} &\leq 2K\exp\left\{-\frac{2n\epsilon_\zeta^2}{3^22^2C^2\zeta^2} \right\}\\
& =2K\exp\left\{-\frac{n\epsilon_\zeta^2}{18C^2\zeta^2} \right\}\\
&=O(\exp(-n)).\end{align*}
Since $\frac{\log m}{n}\rightarrow0$ by Assumption \ref{a:growth condition}(a), it follows that $O_p(\exp(-n))=o_p(1/m)$.
Note that the bound $\epsilon_\zeta/3>0$ is uniform in $\tau$, and $K$ is finite. By stochastic equicontinuity of $\mathbb Q_{nj}(\beta,\tau)$, as $n\rightarrow\infty$, uniformly in $\tau\in\mathcal T$,
\begin{equation*}
    \mathbb Q_{nj}(\beta^k,\tau)=Q_j(\beta^k,\tau)+o_p(1).
\end{equation*}
It follows that 
\begin{equation*}
    \underset{\tau\in\mathcal T}{\sup} \sum_{k=1}^K \Pr \left\{|\mathbb Q_{nj}(\beta^k,\tau)-Q_j(\beta^k,\tau)|>\frac{\epsilon_\zeta}{3} \right\}=o_p\left(\frac{1}{m}\right).
\end{equation*}
\end{proof}

\subsubsection{Lemma \ref{l:consistency of G}}

\begin{lemma}[\textbf{Uniform consistency of $\hat G(\tau)$ with a full-rank weighting matrix}]\label{l:consistency of G}
Assumptions \ref{a:sampling}, \ref{a:covariates}, \ref{a:instruments}, and \ref{a:weighting} hold. Then,
\begin{equation}
\underset{\tau\in\mathcal{T}}{\sup}\left\lVert \left(S_{ZX}'\hat W(\tau) S_{ZX}\right)^{-1}S_{ZX}'\hat W(\tau)-\left(\Sigma_{ZX}'W(\tau) \Sigma_{ZX}'\right)^{-1}\Sigma_{ZX}'W(\tau)\right\rVert=o_p(1)
\end{equation}
and $\left(\Sigma_{ZX}'W(\tau) \Sigma_{ZX}'\right)^{-1}\Sigma_{ZX}'W(\tau)$ is uniformly bounded and continuous.
\end{lemma}

\begin{proof}[Proof of Lemma \ref{l:consistency of G}]
First, it follows from Assumptions \ref{a:sampling}(ii), \ref{a:covariates}(i) and \ref{a:instruments}(i) that $\Var\left(\frac{1}{n}\sum_{i = 1}^n z_{ij}x_{ij}'\right)=o_p\left(\frac{1}{n}\right)$ and $\bE\left[\frac{1}{n}\sum_{i = 1}^n z_{ij}x_{ij}'\right]=\bE[z_{ij}x_{ij}']$. Hence, by Assumption \ref{a:sampling}(i), $\Var\left(\frac{1}{m}\sum_{j = 1}^m\frac{1}{n}\sum_{i = 1}^n x_{ij}z_{ij}'\right)=o_p\left(\frac{1}{mn}\right)$. By Chebyshev’s inequality,
$$\frac{1}{m}\sum_{j = 1}^m\left(\frac{1}{n}\sum_{i = 1}^n z_{ij}x_{ij}'-\bE[z_{ij}x_{ij}']\right)\underset{p}{\rightarrow}0.$$
In addition, by Assumption \ref{a:instruments}(iii), $m^{-1}\sum_{j = 1}^m\bE[z_{ij}x_{ij}']\rightarrow\Sigma_{ZX}$. It follows that
\begin{equation*}\label{eq:sigmaXZ}
   S_{ZX}\underset{p}{\rightarrow}\Sigma_{ZX}.
\end{equation*}

By assumption \ref{a:weighting}, uniformly in $\tau \in \mathcal{T}$, $\hat W(\tau)\underset{p}{\rightarrow}W(\tau)$ where $W(\tau)$ is uniformly continuous and strictly positive definite. Since $\Sigma_{ZX}$ is bounded and has full column rank, it implies that $\Sigma_{ZX}'W(\tau)\Sigma_{ZX}$ is also uniformly continuous and invertible. The result of the lemma follows.
\end{proof}

\subsubsection{Theorem \ref{t:consistency}}

\begin{proof}[Proof of Theorem \ref{t:consistency}]\label{p:consistency}

By Lemma \ref{l:sampling error},
\begin{equation*}
\hat \delta (\tau) - \delta(\tau) = \left( S_{ZX}' \hat W(\tau) S_{ZX} \right) ^{-1} S_{ZX}' \hat W(\tau) \frac{1}{mn}\sum_{j = 1}^m\sum_{i = 1}^n z_{ij}\left( \tilde x_{ij}'\left(\hat \beta_j(\tau)-\beta_j(\tau)\right)+\alpha_j(\tau)\right).
\end{equation*}
By Lemma \ref{l:consistency of G}, the first factor converges uniformly to $G(\tau)$:
\begin{equation}\label{eq:consistency1}
\left(S_{ZX}'\hat W(\tau) S_{ZX}\right)^{-1}S_{ZX}'\hat W(\tau) \underset{p}{\rightarrow}\left(\Sigma_{ZX}'W(\tau) \Sigma_{ZX}'\right)^{-1}\Sigma_{ZX}'W(\tau)=G(\tau).
\end{equation}
By lemma \ref{l:uniform coefficients}, $\hat\beta_j(\tau)$ is consistent for $\beta_j(\tau)$ uniformly in $j$ and $\tau$. Together with the boundedness of $x_{ij}$ in Assumption \ref{a:covariates}(i) and of $z_{ij}$ in Assumption \ref{a:instruments}(i), it follows that
\begin{equation}\label{eq:consistency2}
\underset{\tau\in\mathcal T}{\sup} \ \frac{1}{mn} \sum_{j = 1}^m\sum_{i = 1}^n z_{ij} \tilde x_{ij}'(\hat \beta_j(\tau)-\beta_j(\tau))\underset{p}{\rightarrow}0 .
\end{equation}
By Assumption \ref{a:instruments}(ii), $\bE[z_{ij}\alpha_j(\tau)]=0$ uniformly in $\tau$. By Assumption \ref{a:instruments}, $\Var(z_{ij}\alpha_j(\tau))$ is uniformly bounded. In addition, $z_{ij}$ is bounded and $\alpha_j(\tau)$ is uniformly continuous in $\tau$. Hence,
\begin{equation}\label{eq:consistency3}
\underset{\tau\in\mathcal T}{\sup} \ \frac{1}{mn} \sum_{j = 1}^m\sum_{i = 1}^n z_{ij} \alpha_j(\tau)\underset{p}{\rightarrow}0 .
\end{equation}
The result of the Theorem follows from equations (\ref{eq:consistency1}), (\ref{eq:consistency2}), and (\ref{eq:consistency3}).
\end{proof}

\subsection{Proof of Lemma \ref{l:asymptotic moments} - Asymptotic distribution of sample moments}
\begin{proof}[Proof of Lemma \ref{l:asymptotic moments}]
\textbf{Part (i)} 

Lemma 3 in \cite{Galvao2020} provides the uniform Bahadur representation for the group-level quantile regression coefficient under our assumptions:
\begin{equation}
\hat\beta_j(\tau)-\beta_j(\tau)=\frac{1}{n}\sum_{i = 1}^n\phi_{j,\tau}(\tilde x_{ij},y_{ij})+R_{nj}^{(1)}(\tau)+R_{nj}^{(2)}(\tau),
\end{equation}
where
\begin{equation}
\phi_{j,\tau}(\tilde x_{ij},y_{ij})=-B_{j,\tau}^{-1} \tilde x_{ij}(1(y_{ij}\leq \tilde x_{ij}\beta_j(\tau))-\tau),
\end{equation}
with $B_{j,\tau}=\bE_{i|j}[f_{y|x}(Q_{y|x, \nu_j}(\tau|\tilde x_{ij} )|\tilde x_{ij})\tilde x_{ij} \tilde x_{ij}']$ and 
\begin{equation}\label{r2}\underset{j}{\sup}\;\underset{\tau\in \mathcal{T}}{\sup}\left\lVert R_{nj}^{(2)}(\tau)\right\lVert=O_p\left(\frac{\log n}{n}\right),\end{equation}
\begin{equation}\label{r1}\underset{j}{\sup}\;\underset{\tau\in \mathcal{T}}{\sup}\left\lVert \bE_{i|j}\left[R_{nj}^{(1)}(\tau)\right]\right\lVert=O\left(\frac{\log n}{n}\right),\end{equation}
\begin{equation} \label{var_r1} \underset{j}{\sup}\;\underset{\tau\in \mathcal{T}}{\sup}\left\lVert \bE_{i|j} \left[\left(R_{nj}^{(1)}(\tau)-\bE_{i|j}[R_{nj}^{(1)}(\tau)]\right)\left(R_{nj}^{(1)}(\tau)-\bE_{i|j}[R_{nj}^{(1)}(\tau)]\right)' \right]\right\lVert =O\left(\left(\frac{\log n}{n}\right)^{3/2}\right).\end{equation}
It follows that 
\begin{align}\label{score}
\frac{1}{mn} \sum_{j = 1}^m\sum_{i = 1}^n z_{ij}\tilde x_{ij}' \left(\hat \beta_j(\tau)-\beta_j(\tau)\right)=& \frac{1}{m}\sum_{j = 1}^m\left(\frac{1}{n}\sum_{i = 1}^n z_{ij}\tilde x_{ij}'\right)\left(\frac{1}{n}\sum_{i = 1}^n\phi_{j,\tau}(\tilde x_{ij},y_{ij})\right) \\
\label{score_r1}&+ \frac{1}{m}\sum_{j = 1}^m\left(\frac{1}{n}\sum_{i = 1}^n z_{ij}\tilde x_{ij}'\right)R_{nj}^{(1)}(\tau) \\
&\label{score_r2}+\frac{1}{m}\sum_{j = 1}^m\left(\frac{1}{n}\sum_{i = 1}^n z_{ij}\tilde x_{ij}'\right)R_{nj}^{(2)}(\tau).
\end{align}

Consider first the third term (\ref{score_r2}). By assumptions \ref{a:covariates}(i) and \ref{a:instruments}(i), $x_{ij}$ and $z_{ij}$ are bounded by $C$ such that the sample mean of their product is also bounded. Therefore, (\ref{r2}) implies that
\begin{equation}
\label{r2_sup}
\underset{\tau\in\mathcal T}{\sup} \ \frac{1}{m}\sum_{j = 1}^m\left(\frac{1}{n}\sum_{i = 1}^n z_{ij}\tilde x_{ij}'\right)R_{nj}^{(2)}(\tau)=O_p\left(\frac{\log n}{n} \right).
\end{equation} 
Consider now the second term (\ref{score_r1}). Since $\Var\left(R_{nj}^{(1)}(\tau)\right)=o\left(\frac{1}{n}\right)$ by (\ref{var_r1}),  $x_{ij}$ and $z_{ij}$ are bounded by assumptions \ref{a:covariates}(i) and \ref{a:instruments}(i), and observations are independent across groups, it follows that $\Var\left(\frac{1}{m}\sum_{j = 1}^m\left(\frac{1}{n}\sum_{i = 1}^n z_{ij}\tilde x_{ij}'\right)R_{nj}^{(1)}(\tau)\right)=o_p\left(\frac{1}{mn}\right)$. In addition, by (\ref{r1}), $\underset{j}{\sup}\ \underset{\tau\in \mathcal T} {\sup} \ \bE_{i|j} \left [ R_{nj}^{(1)}(\tau) \right ]=O\left(\frac{\log n}{n}\right)$ such that $\underset{\tau\in\mathcal T}{\sup}\  \bE_{i|j} \left[\frac{1}{m}\sum_{j = 1}^m\left(\frac{1}{n}\sum_{i=1}^n z_{ij}\tilde x_{ij}'\right)R_{nj}^{(1)}(\tau)\right]=O\left(\frac{\log n}{n}\right)$. Putting this together, by the Chebyshev inequality and under Assumption \ref{a:growth condition}(c),
\begin{equation}
\label{r1_sup}
\underset{\tau\in\mathcal T}{\sup} \frac{1}{m}\sum_{j = 1}^m\left(\frac{1}{n}\sum_{i = 1}^n z_{ij}\tilde x_{ij}'\right)R_{nj}^{(1)}(\tau) = o_p\left(\frac{1}{\sqrt{mn}}\right).
\end{equation} 
It follows that both remainder terms are $o_p\left(\frac{1}{\sqrt{mn}}\right)$ uniformly over $\tau$. 

Let $\Sigma_{ZXj}=\bE_{i|j}[z_{ij}\tilde x_{ij}']$ and consider now the term (\ref{score}):
\begin{align}
 \frac{1}{m}\sum_{j = 1}^m \left(\frac{1}{n}\sum_{i=1}^n z_{ij}\tilde x_{ij}'\right) &\left(\frac{1}{n}\sum_{i = 1}^n\phi_{j,\tau}(\tilde x_{ij},y_{ij})\right) \nonumber \\ =& \frac{1}{m}\sum_{j = 1}^m\left(\frac{1}{n}\sum_{i = 1}^n z_{ij}\tilde x_{ij}'-\Sigma_{ZXj}\right) \left(\frac{1}{n}\sum_{i = 1}^n\phi_{j,\tau}(\tilde x_{ij},y_{ij})\right) \nonumber \\ & +\frac{1}{m}\sum_{j = 1}^m\Sigma_{ZXj} \left(\frac{1}{n}\sum_{i = 1}^n\phi_{j,\tau}(\tilde x_{ij},y_{ij})\right).
\end{align}
By the boundedness of $z_{ij}$ and $x_{ij}$ and the independence of the observations over time, it follows that $\left\lVert \frac{1}{n}\sum_{i=1}^n z_{ij}\tilde x_{ij}'-\Sigma_{ZXj}\right\rVert = o(1)$ uniformly in $j$. In addition, $\Var \left(\frac{1}{n}\sum_{i = 1}^n\phi_{j,\tau}(\tilde x_{ij},y_{ij})\right)=O\left(\frac{1}{n}\right)$. Hence,
\begin{equation*}
\Var\left(\frac{1}{m}\sum_{j = 1}^m\left(\frac{1}{n}\sum_{i = 1}^n z_{ij}\tilde x_{ij}'-\Sigma_{ZXj}\right) \left(\frac{1}{n}\sum_{i = 1}^n\phi_{i,\tau}(\tilde x_{ij},y_{ij})\right)\right) = o\left(\frac{1}{mn}\right).
\end{equation*}
The model in equation (\ref{eq:model}) and Assumption \ref{a:instruments}(iii) imply that $\bE_{i|j} \left[1(y_{ij}\leq\tilde x_{ij}\beta_j(\tau))|\tilde x_{ij},z_{ij}\right]=\tau$, which implies that $$\bE \left[\frac{1}{m}\sum_{j = 1}^m\left(\frac{1}{n}\sum_{i = 1}^n z_{ij}\tilde x_{ij}'-\Sigma_{ZXj}\right) \left(\frac{1}{n}\sum_{i = 1}^n\phi_{j,\tau}(\tilde x_{ij},y_{ij})\right)\right]=0$$ uniformly in $\tau$. Therefore, by Chebyshev's inequality, 
\begin{equation}
\frac{1}{m}\sum_{j = 1}^m\left(\frac{1}{n}\sum_{i = 1}^n z_{ij}\tilde x_{ij}'-\Sigma_{ZXj}\right) \left(\frac{1}{n}\sum_{i = 1}^n\phi_{j,\tau}(\tilde x_{ij},y_{ij})\right) =o_p\left(\frac{1}{\sqrt{mn}}\right)
\end{equation}
uniformly in $\tau$.

Since all other terms are $o_p\left(\frac{1}{\sqrt{mn}}\right)$ uniformly over $\tau$, the limiting distribution of the process  $\frac{1}{mn} \sum_{j = 1}^m\sum_{i = 1}^n z_{ij}\tilde x_{ij}' \left(\hat \beta_j(\tau)-\beta_j(\tau)\right)$ is the same as the limiting distribution of
\begin{multline}\label{eq:s_j}
\frac{1}{m}\sum_{j = 1}^m \Sigma_{ZXj}\left(\frac{1}{n}\sum_{i = 1}^n\phi_{j,\tau}(\tilde x_{ij},y_{ij})\right)\\=\frac{1}{m}\sum_{j = 1}^m \Sigma_{ZXj}\left(\frac{-B_{j,\tau}^{-1}}{n}\sum_{i = 1}^n \tilde x_{ij}(1(y_{ij}\leq \tilde x_{ij}\beta_j(\tau))-\tau)\right)=\frac{1}{mn}\sum_{j = 1}^m \sum_{i = 1}^n s_{ij}(\tau).\end{multline}
This is a sample mean over $mn$ independent (but not necessarily identically distributed) observations denoted by $s_{ij}(\tau)$. The model in equation (\ref{eq:model}) and Assumption \ref{a:instruments}(iii) imply that $\bE \left[1(y_{ij}\leq\tilde x_{ij}\beta_j(\tau))|\tilde x_{ij},z_{ij},v_j\right]=\tau$, which implies that $\bE[s_{ij}(\tau)]=0$. In addition,
\begin{align}
\Var(s_{ij}(\tau))= \bE[\Sigma_{ZXj}\Var(\phi_{j,\tau})\Sigma_{ZXj}']= \bE [\Sigma_{ZXj}B_{j,\tau}^{-1}\tau(1-\tau)\bE_{i|j}[\tilde x_{ij} \tilde x_{ij}']B_{j,\tau}^{-1}\Sigma_{ZXj}'].
\end{align}   
Pointwise asymptotic normality follows from the Lindeberg CLT.

Next we note that $\left\{\Sigma_{ZXj}\left(\frac{-B_{j,\tau}^{-1}}{n}\sum_{i = 1}^n \tilde x_{ij}(1(y_{ij}\leq \tilde x_{ij}\beta)-\tau)\right), \tau\in\mathcal{T},\beta\in\mathcal B\right\}$ is a Donsker class for any compact set $\mathcal B$. This follows by noting that $\left\{1(y_{ij}\leq\tilde x_{ij}\beta_j(\tau)), \tau\in\mathcal T,\beta\in\mathcal B\right\}$ is a VC subgraph class and hence a bounded Donsker class. Hence, $$\left\{\frac{1}{n}\sum_{i = 1}^n\tilde x_{ij}(1(y_{ij}\leq\tilde x_{ij}\beta)-\tau), \tau\in\mathcal T,\beta\in\mathcal B\right\}$$is also bounded Donsker with a square-integrable envelope $2\cdot\max_{i\in 1,...,n}\lVert \tilde x_{ij}\rVert \leq 2\cdot C$. The whole function is then Donsker by the boundedness of $\Sigma_{ZXj}$ and $B_{j,\tau}^{-1}$. The weak convergence result follows by application of the functional central limit theorem for independent but not identically distributed random variables, see, for instance, Theorem 3 in \cite{Brown1971}.

\textbf{Part (ii)} Follows directly by Lemma 3 in \cite{Chetverikov2016}. 

\textbf{Part (iii)} The first moment is asymptotically equivalent to (up to a term, which is uniformly $o_p\left(\frac{1}{\sqrt{mn}}\right)$):  
$$\frac{1}{m}\sum_{j = 1}^m \Sigma_{ZXj}\left(\frac{-B_{j,\tau}^{-1}}{n}\sum_{i = 1}^n \tilde x_{ij}(1(y_{ij}\leq \tilde x_{ij}\beta_j(\tau))-\tau)\right).$$
We have already shown that both moments have a mean of zero. By Assumption \ref{a:sampling}, the observations are independent across $i$ and $j$ such that we only need to consider the correlation between both moments for the same individual and group: 
\begin{align*}
    \Cov&(\tilde x_{ij}(1(y_{ij}\leq \tilde x_{ij}\beta_j(\tau))-\tau),z_{ij}\alpha_j(\tau'))\\
    &=\bE[\tilde x_{ij}(1(y_{ij}\leq \tilde x_{ij}\beta_j(\tau))-\tau)z_{ij}'\alpha_j(\tau')]\\
    &=\bE[\tilde x_{ij}\bE_{i|j}[(1(y_{ij}\leq \tilde x_{ij}\beta_j(\tau))-\tau)|x_{ij},z_{ij}]z_{ij}'\alpha_j(\tau)]
    =0.
\end{align*}
It then directly follows that
\begin{align*}
\underset{\tau,\tau'\in\mathcal T}{\sup} \left\lVert \Cov \left (\bar g^{(1)}_{mn}(\hat\delta, \tau), \bar g^{(2)}_{mn}(\hat\delta, \tau')\right )\right\rVert=o_p\left(\frac{1}{\sqrt{mn}}\right).
\end{align*}
\end{proof}

\subsection{Lemma \ref{l:consistency G mat uniform}: Uniform consistency of the G matrix with a heterogeneous weighting matrix}

\begin{lemma}\label{l:consistency G mat uniform}
Assumptions \ref{a:sampling}, \ref{a:covariates}, \ref{a:instruments}, and \ref{a:weighting 2} hold.
Then, uniformly in $\tau$, we can partition $\hat G(\tau)$ into $\hat G_{11}(\tau)$ (with dimensions $M_1 \times L_1$), $\hat G_{12}(\tau)$ ($M_1 \times L_2$), $\hat G_{21}(\tau)$ ($M_2 \times L_1$), and $\hat G_{22}(\tau)$ ($M_2 \times L_2$) such that
\begin{equation*}
\hat G(\tau)=\begin{pmatrix}
        \hat G_{11}(\tau) & \hat G_{12}(\tau) \\ \hat G_{21}(\tau) & \hat G_{22}(\tau)
    \end{pmatrix}=\begin{pmatrix}
        G_{mn,11}(\tau) & G_{mn,12}(\tau) \\ G_{mn,21}(\tau) & G_{mn,22}(\tau)
    \end{pmatrix}+ \begin{pmatrix}
        o_p\left(1\right) & o_p \left (\sqrt{a_n(\tau)}\right ) \\ o_p \left (1/\sqrt{a_n(\tau)} \right) & o_p\left(1\right)
    \end{pmatrix},
\end{equation*}
where $G_{mn}(\tau)$ is defined in equation (\ref{eq:Gmn matrix}), $\underset{\tau\in \mathcal{T}}{\sup} \ G_{mn,12}(\tau)=O_p\left(a_n(\tau)\right)$ and the other elements of $G_{mn}(\tau)$ are $O_p\left(1\right)$ uniformly in $\tau$. 

\end{lemma}

\begin{proof}
To simplify the notation, we omit the dependency on $\tau$. We partition $S_{ZX}'$  into four submatrices that have the same dimensions as the submatrices of $\hat G$. As in proposition \ref{p:weight matrix adaptive}, we also partition $\hat W$ into four submatrices such that the diagonal submatrices are square matrices of dimensions $L_1$ and $L_2$. 
Note that
\begin{align*}
    S'_{ZX} \hat W =\begin{pmatrix}S_{11}' & S_{21}' \\
0 & S_{22}'
\end{pmatrix}\begin{pmatrix}  \hat W_{11} &  \hat W_{12} \\
 \hat W_{21} &  \hat W_{22}
\end{pmatrix}= \begin{pmatrix}S_{11}' \hat W_{11}+ S_{21}'\hat W_{21} & S_{11}'\hat W_{12}+S_{21}'\hat W_{22} \\
 S_{22}'\hat  W_{21} & S_{22}'\hat  W_{22}
\end{pmatrix}.
\end{align*}
Furthermore,
\begin{align*}
    S_{ZX}'\hat WS_{ZX}&=
    \begin{pmatrix}S_{11}' & S_{21}' \\
0 & S_{22}'
\end{pmatrix}\begin{pmatrix}  \hat W_{11} &  \hat W_{12} \\
 \hat W_{21} &  \hat W_{22}
\end{pmatrix}\begin{pmatrix}S_{11} & 0 \\
S_{21} & S_{22}
\end{pmatrix}\\
&=\begin{pmatrix} \hat A_{11}+\hat B_{11} &  \hat B_{12} \\
 \hat B_{21} &  \hat B_{22}
\end{pmatrix},
\end{align*}
where $\hat A_{11}=S_{11}'\hat W_{11}S_{11}$, $\hat B_{11}=S_{11}'\hat W_{12}S_{21} + S_{21}' \hat W_{21} S_{11} + S_{21}'\hat W_{22}S_{21}$, $\hat B_{12} = S_{11}' \hat W_{12}
S_{22} + S_{21}' \hat W_{22} S_{22}$, $ \hat B_{21}= S_{22}'\hat W_{21} S_{11} +S_{22}'\hat W_{22}S_{21}$, and $\hat B_{22}=S_{22}'\hat W_{22}S_{22}$.

For the non-zero elements of the matrix $S_{ZX}$, we have
\begin{align*}
    S_{11} - \Sigma_{11} = o_p(1), \\
    S_{22} - \Sigma_{22} = o_p(1) ,\\
    S_{21} - \Sigma_{21} = o_p(1).
\end{align*}
Together with proposition \ref{p:weight matrix adaptive}, this implies that uniformly over $\tau$,
\begin{align}
    \hat A_{11} =& \Sigma_{11}'W_{11}\Sigma_{11} + o_p(1) ,
\\[1em]
   \hat B_{11}= & a_n\Sigma_{11}'W_{12}\Sigma_{21} + a_n \Sigma_{21}' W_{21} \Sigma_{11} + a_n \Sigma_{21}'W_{22}\Sigma_{21} + o_p(\sqrt{a_n}) =  a_n B_{11} + o_p(\sqrt{a_n}) ,    \\[1em]
    \hat B_{12} = & a_n \Sigma_{11}'W_{12}
\Sigma_{22} +  a_n \Sigma_{21}' W_{22} \Sigma_{22} + o_p(\sqrt{a_n})  =  a_n B_{12} + o_p(\sqrt{a_n}),     \\[1em] \label{eq:B21}
\hat B_{21} = & a_n \Sigma_{22}' W_{21} \Sigma_{11} +a_n \Sigma_{22}' W_{22}\Sigma_{21}  + o_p(\sqrt{a_n})=  a_n B_{21} + o_p(\sqrt{a_n}),\\[1em] 
\hat B_{22}=& a_n\Sigma_{22}'W_{22}\Sigma_{22} + o_p(a_n) = a_n B_{22} + o_p(a_n).
    \end{align}

\noindent By the inverse of a partitioned matrix
\begin{equation*}
    \left(S_{ZX}' \hat WS_{ZX}\right)^{-1}=\\
\begin{pmatrix}\hat C+\hat  C \hat B_{12}\hat D \hat B_{21}\hat C & -\hat C \hat B_{12}\hat D \\
-\hat D \hat B_{21}\hat  C & \hat D
\end{pmatrix},
\end{equation*}
where $\hat C=\left(\hat  A_{11}+ \hat 
 B_{11}\right)^{-1}$ and $\hat  D=\left(\hat B_{22}-\hat  B_{21}\hat C\hat  B_{12}\right)^{-1}$. Hence uniformly over $\tau$,
\begin{align}
    \hat C = & C + o_p(1), \label{eq:Chat} \\
    \hat D = & a_n^{-1} D + o_p \left ( a_n^{-1} \right), \label{eq:Dhat}\\
    \hat D \hat B_{21} \hat C =& D B_{21} C + o_p(a_n^{-1/2}), \label{eq:DBC}
\end{align}
where $C$ and $D$ are strictly positive and bounded. Equation (\ref{eq:Dhat}) holds because 
$$a_n \hat D =  \left ( a_n^{-1} \hat B_{22}-a_n^{-1}\hat  B_{21}\hat C\hat  B_{12} \right)^{-1} =  \left( B_{22}- a_n B_{21} C  B_{12} \right)^{-1}  + o_p(1) =  D + o_p(1).$$ Combining equation (\ref{eq:B21}) with equations (\ref{eq:Chat})-(\ref{eq:Dhat}), we obtain similarly
\begin{equation}
    \hat C \hat B_{12} \hat D = C B_{12} D + o_p(a_n^{-1/2}). \label{eq:CBD}
\end{equation}
 
Now, we can consider each submatrix of $\hat G$ separately and derive their convergence rate to the corresponding element of the $G_{mn}$ matrix defined in equation (\ref{eq:Gmn matrix}): 
$$G_{mn}(\tau) = \left ( \Sigma_{ZX}' W_{mn} (\tau) \Sigma_{ZX}  \right) ^{-1} \Sigma_{ZX}' W_{mn}(\tau).$$
For the upper left term, we have
\begin{align*}
       \hat G_{ 11}& =
         ( \hat C+  \hat C  \hat B_{12}\hat D  \hat B_{21}\hat C  )(S_{11}'\hat W_{11}+ S_{21}'\hat W_{21} )    -  \hat C  \hat B_{12} \hat D S_{22}' \hat W_{21}\\
    &= \left( C+  a_n C B_{12} D B_{21} C + o_p(1) \right) \left (\Sigma_{11}'W_{11}+a_n \Sigma_{21}'W_{21} + o_p(1)  \right) \\ & \quad + \left ( C  B_{12} D + o_p(a_n^{-1/2}) \right) \left (a_n  \Sigma_{22}' W_{21} + o_p(\sqrt{a_n}) \right) \\&= ( C+  a_n C  B_{12} D  B_{21} C  )(\Sigma_{11}'W_{11}+a_n \Sigma_{21}'W_{21} )    -   a_n  C  B_{12} D\Sigma_{22}' W_{21} + o_p(1)\\
    & = G_{mn, 11} + o_p(1),
\end{align*}
uniformly over $\tau$.
For the upper right term, we have uniformly over $\tau$,
\begin{align*}
    \hat G_{12} =&  ( \hat C+ \hat  C \hat   
 B_{12} \hat  D \hat  B_{21}\hat  C ) \left(S_{11}'\hat 
 W_{12}+S_{21}' \hat  W_{22}\right) - \hat  C  \hat  B_{12}\hat  D  S_{22}' \hat  W_{22} \\ =& \left( C+  a_n C  B_{12} D  B_{21} C + o_p(1) \right) \left (a_n\left(\Sigma_{11}'W_{12}+\Sigma_{21}'W_{22}\right)  + o_p(\sqrt{a_n}) \right) \\ &  - \left ( C  B_{12} D + o_p(a_n^{-1/2}) \right) \left (a_n \Sigma_{22}' W_{22} + o_p(a_n) \right)\\ = &a_n( C+ a_n C  B_{12} D  B_{21} C ) \left(\Sigma_{11}'W_{12}+\Sigma_{21}'W_{22}\right) - a_n C  B_{12} D   \Sigma_{22}' W_{22} +  o_p(\sqrt{a_n}) \\
 =&  G_{mn,12} + o_p(\sqrt{a_n}).
\end{align*}
For the lower right term, 
\begin{align*}
    \hat G_{22} = &  -\hat  D  \hat B_{21} \hat  C \left(S_{11}'\hat W_{12}+S_{21}'\hat W_{22}\right) + \hat D S_{22}' \hat W_{22}  \\ =&-\left(  D  B_{21}  C + o_p({a_n}^{-1/2}) \right) \left(a_n\Sigma_{11}'W_{12}+a_n\Sigma_{21}'W_{22} + o_p(\sqrt{a_n} )\right) \\ = &- a_n D  B_{21}  C \left(\Sigma_{11}'W_{12}+\Sigma_{21}'W_{22}\right) + D \Sigma_{22}' W_{22} + o_p(1) \\
    = & G_{mn,22} + o_p(1)  ,
\end{align*}
uniformly over $\tau$.
Finally, for the lower left term, uniformly of $\tau$,
\begin{align*}
            \hat G_{21} = & -\hat D  \hat B_{21} \hat C (S_{11}'\hat W_{11}+ S_ {21}'\hat W_{21} ) +\hat  D  S_{22}' \hat W_{21} \nonumber \\  =&  -\left (D B_{21} C + o_p(a_n^{-1/2}) \right)   \left (\Sigma_{11}' W_{11} +  a_n \Sigma_{21}' W_{21}  + o_p(\sqrt{a_n}) \right) \nonumber \\ &+\left (a_n^{-1} D + o_p(a_n^{-1} ) \right)  \left( \Sigma_{22}'  a_n W_{21} + o_p(\sqrt{a_n}) \right) \nonumber \\
       =&  - D B_{21} C \left ( \Sigma_{11}' W_{11}  -  \Sigma_{21}'a_n W_{21} \right) +      D\Sigma_{22}'   W_{21} + o_p({a_n}^{-1/2}) \\
    = & G_{mn,21} + o_p \left ( a_n ^{-1/2}\right ).
\end{align*}

Hence,
\begin{equation*}
\sup_\tau \left \rVert \hat G(\tau)- G_{mn}(\tau) \right \rVert =  \begin{pmatrix}
        o_p(1) & o_p \left (\sqrt{a_n(\tau)}\right ) \\ o_p \left (1/\sqrt{a_n(\tau)} \right ) & o_p(1)
    \end{pmatrix}
\end{equation*}
where $G_{mn,12}(\tau)=O_p(a_n(\tau))$ and the other elements of $G_{mn}(\tau)$ are $O_p(1)$ uniformly over $\tau$.

\end{proof}

\subsection{Proof of Theorem \ref{t:first-order}: Degree of heterogeneity is known}
\begin{proof}[Proof of Theorem \ref{t:first-order}]\label{proof:first_order}
From the definition of the estimator,
\begin{align}\label{eq:error representation}
\hat\delta(\tau)-\delta(\tau)&=\left(S_{ZX}'\hat W(\tau)S_{ZX}\right)^{-1}S_{ZX}'\hat W(\tau) \bar g_{mn}(\hat \delta,\tau).\end{align}

In part (i) Lemma \ref{l:consistency G mat uniform} applies uniformly in $\tau$ with $a_n(\tau)=O_p(1/n)$ such that
\begin{align*}
\left(S_{ZX}'\hat W(\tau) S_{ZX}\right)^{-1}S_{ZX}'\hat W(\tau)&=\begin{pmatrix}
    \hat G_{11}(\tau) & \hat G_{12}(\tau)\\
    \hat G_{21}(\tau) & \hat G_{22}(\tau)
\end{pmatrix}\\
&=\begin{pmatrix}
        G_{11}(\tau)+o_p(1) & o_p(1/\sqrt n) \\ o_p(\sqrt n) & G_{22}(\tau)+o_p(1)
    \end{pmatrix}.
\end{align*}

By the definition of the fast instruments, the first $L_1$ elements of $\bar g^{(2)}_{mn}(\hat\delta,\tau)$ are equal to zero. Together with Lemma \ref{l:asymptotic moments}(i), this implies that
\begin{equation}\label{eq:Z11}
    \sqrt{mn}\bar g_{mn,1}(\hat \delta,\cdot)\rightsquigarrow \mathbb Z_{11}(\cdot),
\end{equation}
where $\bar g_{mn,1}(\hat \delta,\cdot)$ contains the first $L_1$ elements of $\bar g_{mn}(\hat \delta,\cdot)$. For the remaining $L_2$ elements, we have $\sqrt{m}\bar g_{mn,2}(\hat \delta,\cdot) = O_p(1)$. It follows that 
\begin{align*}
    \sqrt{mn}\left(\hat\delta_1(\cdot)-\delta_1(\cdot)\right)&=\hat G_{11}\sqrt{mn}\bar g_{mn,1}(\hat\delta,\cdot)+\hat G_{12}\bar g_{mn,2}(\hat\delta,\cdot)\\
    &=G_{11}\sqrt{mn}\bar g_{mn,1}(\hat\delta,\cdot)+o_p(1)\sqrt{mn}\bar g_{mn,1}(\hat\delta,\cdot)+o_p(1/\sqrt{n})\sqrt{mn}\bar g_{mn,2}(\hat\delta,\cdot)\\
    &=G_{11}\sqrt{mn}\bar g_{mn,1}(\hat\delta,\cdot)+o_p(1)\sqrt{mn}\bar g_{mn,1}(\hat\delta,\cdot)+o_p(1)\sqrt{m}\bar g_{mn,2}(\hat\delta,\cdot)\\
    &\rightsquigarrow G_{11}  \mathbb Z_{11}(\cdot)+o_p(1),
\end{align*}
which proves part (i)-(a) of the theorem.

For the remaining $M_2$ coefficients, applying the same results, we obtain
\begin{align*}
\sqrt{m}\left(\hat\delta_2(\cdot)-\delta_2(\cdot)\right)&=\hat G_{21}\sqrt{m}\bar g_{mn,1}(\hat\delta,\cdot)+\hat G_{22}\sqrt{m}\bar g_{mn,2}(\hat\delta,\cdot)\\
&=o_p(\sqrt{n})\sqrt{m}\bar g_{mn,1}(\hat\delta,\cdot)+G_{22}\sqrt{m}\bar g_{mn,2}(\hat\delta,\cdot)+o_p(1)\sqrt{m}\bar g_{mn,2}(\hat\delta,\cdot)\\
&\rightsquigarrow G_{22}(\cdot) \mathbb Z_{22}(\cdot)+o_p(1).
\end{align*}
The result of part (i)-(b) of the theorem follows.

In part (ii) and (iii), Lemma \ref{l:consistency of G} applies such that \begin{equation}\label{eq:G}
\left(S_{ZX}'\hat W(\tau) S_{ZX}\right)^{-1}S_{ZX}'\hat W(\tau)\underset{p}{\rightarrow}G(\tau) 
\end{equation}
uniformly in $\tau\in\mathcal{T}$. In part (ii), $\bar g_{mn}(\hat \delta,\cdot)=\bar g_{mn}^{(1)}(\hat \delta,\cdot)$. It follows from Lemma \ref{l:asymptotic moments}(i) that
\begin{equation}\label{eq:Z1}
    \sqrt{mn}\bar g_{mn}(\hat \delta,\cdot)\rightsquigarrow \mathbb Z_{1}(\cdot).
\end{equation}
The result of part (ii) of the theorem follows from equations (\ref{eq:error representation}), (\ref{eq:G}), and (\ref{eq:Z1}).

In part (iii), $\bar g_{mn}^{(1)}(\hat\delta,\tau)$ and $\bar g_{mn}^{(2)}(\hat\delta,\tau)$ converge at the same rate such that Lemma \ref{l:asymptotic moments} implies
\begin{equation}\label{eq:Z}
    \sqrt{mn}\bar g_{mn}(\hat \delta,\cdot)\rightsquigarrow \mathbb Z(\cdot),
\end{equation}
where $\mathbb Z(\cdot)$ is a mean-zero Gaussian process with uniformly continuous sample paths and covariance function $\Omega_1(\tau,\tau')+\bar\Omega_1(\tau,\tau')$. The result of part (iii) of the theorem follows from equations (\ref{eq:error representation}), (\ref{eq:G}), and (\ref{eq:Z}).

\end{proof}

\subsection{Proof of Theorem \ref{t:adaptive}: Adaptive Asymptotic Distribution}

\begin{proof}[Proof of Theorem \ref{t:adaptive}]
We prove the theorem when Assumption \ref{a:weighting 2} holds. The proof when instead Assumption \ref{a:weighting} holds is simpler because all coefficients converge at the same rate; it can actually be considered as a special case of the proof below when we set $a_n(\tau)=1$.

By Lemma \ref{l:consistency G mat uniform}, uniformly in $\tau\in\mathcal T$,
\begin{equation*}
\hat G(\tau)= G_{mn}(\tau)+ \begin{pmatrix}
        o_p(1) & o_p \left (\sqrt{a_n(\tau)}\right ) \\ o_p \left (1/\sqrt{a_n(\tau)} \right ) & o_p(1)
    \end{pmatrix}
\end{equation*}
where $G_{mn,12}(\tau)=O_p(a_n(\tau))$ and the other elements are $O_p(1)$.

Since the rate of convergence may differ across the quantile index and the regressors, we first consider the scalar $\hat\delta_k(\tau)$, which is the $k$-th element of $\hat\delta(\tau)$ for $k\in\{1,\dots,K\}$ and $\tau \in\mathcal{T}$. From Lemma \ref{l:sampling error},
\begin{equation*}
\hat\delta_k(\tau)-\delta_k(\tau)=\hat G_k(\tau)\bar g_{mn}(\hat\delta,\tau)=\hat G_k(\tau)\left(\bar g_{mn}^{(1)}(\hat\delta,\tau)+\bar g_{mn}^{(2)}(\hat\delta,\tau)\right)
\end{equation*}
where $\hat G_k(\tau)$ is the $k$-th row of $\hat G(\tau)$. From the proof of Lemma \ref{l:asymptotic moments}(i), we have 
\begin{equation*}
    \bar g_{mn}^{(1)}(\hat\delta,\tau)=\frac{1}{m}\sum_{j = 1}^m \Sigma_{ZXj}\frac{1}{n}\sum_{i = 1}^n\phi_{j,\tau}(\tilde x_{ij},y_{ij})+o_p\left(\frac{1}{\sqrt{mn}}\right),
\end{equation*}
uniformly in $\tau\in\mathcal T$.

We deal separately with the fast and slow coefficients and show that the result holds for both types of coefficients. For $k\in\{1,\dots,M_1\}$ (fast coefficients), $\hat G_k(\tau)$ is uniformly bounded in probability, which implies that
\begin{align*}
\hat G_k(\tau)\bar g_{mn}^{(1)}(\hat\delta,\tau)&=\hat G_k(\tau)\frac{1}{m}\sum_{j = 1}^m \Sigma_{ZXj}\frac{1}{n}\sum_{i = 1}^n\phi_{j,\tau}(\tilde x_{ij},y_{ij})+o_p\left(\frac{1}{\sqrt{mn}}\right),
\end{align*}
uniformly in $\tau\in\mathcal T$. In addition, \ref{l:asymptotic moments}(i) implies that $\frac{1}{m}\sum_{j = 1}^m \Sigma_{ZXj}\frac{1}{n}\sum_{i = 1}^n\phi_{j,\tau}(\tilde x_{ij},y_{ij})=O_p\left(\frac{1}{\sqrt{mn}}\right)$ uniformly in $\tau\in\mathcal T$. For these coefficients, $\underset{\tau\in\mathcal{T}}{\sup}\left\lVert \hat G_k(\tau)-G_{mn,k}(\tau)\right\rVert=o_p\left(1\right)$. This, in turn, implies that
\begin{equation*}
    \left(\hat G_k(\tau)-G_{mn,k}(\tau)\right)\frac{1}{m}\sum_{j = 1}^m \Sigma_{ZXj}\frac{1}{n}\sum_{i = 1}^n\phi_{j,\tau}(\tilde x_{ij},y_{ij})=o_p\left(\frac{1}{\sqrt{mn}}\right),
\end{equation*}
uniformly in $\tau\in\mathcal T$. Combining the last two displayed equations, we obtain
\begin{align}\label{eq:convergence g1}
\hat G_k(\tau)\bar g_{mn}^{(1)}(\hat\delta,\tau)=G_{mn,k}(\tau)\frac{1}{m}\sum_{j = 1}^m \Sigma_{ZXj}\left(\frac{1}{n}\sum_{i = 1}^n\phi_{j,\tau}(\tilde x_{ij},y_{ij})\right)+o_p\left(\frac{1}{\sqrt{mn}}\right),
\end{align}
uniformly in $\tau\in\mathcal T$.

For $k\in\{M_1+1,\dots,K\}$ (slow coefficients), by Lemma \ref{l:consistency G mat uniform}, uniformly in $\tau$, $\hat G_k(\tau)-G_{mn,k}(\tau)=o_p(1/\sqrt{a_n(\tau)})$ and $G_{mn,k}=O_p(1)$ such that $\hat G_{k}=O_p(1/\sqrt{a_{n}(\tau)})$. In addition, Lemma \ref{l:asymptotic moments}(i) implies that $$\frac{1}{m}\sum_{j = 1}^m \Sigma_{ZXj} \frac{1}{n} \sum_{i = 1}^n \phi_{j,\tau}(\tilde x_{ij},y_{ij})=O_p\left(\frac{1}{\sqrt{mn}}\right).$$ It follows that
\begin{align*}
\hat G_k(\tau)\bar g_{mn}^{(1)}(\hat\delta,\tau)&=\hat G_k(\tau)\frac{1}{m}\sum_{j = 1}^m \Sigma_{ZXj}\frac{1}{n}\sum_{i = 1}^n\phi_{j,\tau}(\tilde x_{ij},y_{ij})+o_p\left(\frac{1}{\sqrt{mna_n(\tau)}}\right),
\end{align*}
and 
\begin{equation*}
    \left(\hat G_k(\tau)-G_{mn,k}(\tau)\right)\frac{1}{m}\sum_{j = 1}^m \Sigma_{ZXj}\frac{1}{n}\sum_{i = 1}^n\phi_{j,\tau}(\tilde x_{ij},y_{ij})=o_p\left(\frac{1}{\sqrt{mna_n(\tau)}}\right),
\end{equation*}
uniformly in $\tau\in\mathcal T$. Combining the last two displayed equations, we obtain
\begin{align*}
\hat G_k(\tau)\bar g_{mn}^{(1)}(\hat\delta,\tau)=G_{mn,k}(\tau)\frac{1}{m}\sum_{j = 1}^m \Sigma_{ZXj}\left(\frac{1}{n}\sum_{i = 1}^n\phi_{j,\tau}(\tilde x_{ij},y_{ij})\right)+o_p\left(\frac{1}{\sqrt{mna_n(\tau)}}\right),
\end{align*}
uniformly in $\tau\in\mathcal T$. 

Lemma \ref{l:asymptotic moments}(ii) implies that 
\begin{equation}\bar g_{mn,l}^{(2)}(\delta,\tau)=\frac{1}{m}\sum_{j = 1}^m \bar z_{j,l}\alpha_j(\tau)=O_p\left(\frac{1}{\sqrt m}\right)\sqrt{\Omega_{2,ll}(\tau)}\end{equation}
uniformly in $\tau\in\mathcal T$, where $g_{mn,l}^{(2)}(\delta,\tau)$ is the $l$-th element of $g_{mn}^{(2)}(\delta,\tau)$, $z_{j,l}$ is the $l$-th element of $z_{j}$, and $\Omega_{2,ll}(\tau)$ it the element in the $l$-th row and $l$-th column of $\Omega_{2}(\tau)$. Let $\hat G_{kl}(\tau)$ and $G_{mn,kl}(\tau)$ be the element in the $k$-th row and $l$-th column of $\hat G(\tau)$ and $G_{mn}(\tau)$, respectively. Note that $\hat G_{kl}(\tau)-G_{mn,kl}(\tau)=o_p(\sqrt{G_{mn,kl}(\tau)})$ for $l\in\{L_1+1,\dots,L\}$ and $k\in\{1,\dots,K\}$. Thus, we have
\begin{align}\label{eq:convergence g2}
(\hat G_k(\tau)-G_{mn}(\tau)) \bar g_{mn}^{(2)}(\hat\delta,\tau)&=\sum_{l=L_1+1}^{L}(\hat G_{kl}(\tau)-G_{mn,kl}(\tau))\bar g_{mn,l}^{(2)}(\hat \delta,\tau)\nonumber\\
&=\sum_{l=L_1+1}^{L}o_p \left (\sqrt{G_{mn,kl}} \right )O_p\left(\frac{1}{\sqrt m}\right)\sqrt{\Omega_{2,ll}(\tau)}\nonumber\\
&=o_p\left(\frac{\sqrt{\Var(\alpha_j(\tau))}}{\sqrt{m}}\right)\sum_{l=L_1+1}^{L} \sqrt{G_{mn,kl}(\tau))},\nonumber
\end{align}
uniformly in $\tau\in\mathcal T$. 

If $k\in\{1,\dots,M_1\}$, define
\begin{equation}\label{eq:zetak1}
    \zeta(k,\tau)=\frac{1}{\sqrt{mn}}+\frac{\sqrt{\Var(\alpha_j(\tau))}}{\sqrt m}\sum_{l=L_1 + 1}^L\sqrt{G_{mn,kl}(\tau)}.
\end{equation}
If $k\in\{M_1+1,\dots,K\}$, define
\begin{equation}\label{eq:zetak2}
    \zeta(k,\tau)=\frac{1}{\sqrt{mna_n(\tau)}}+\frac{\sqrt{\Var(\alpha_j(\tau))}}{\sqrt m}\sum_{l=L_1+1}^L\sqrt{G_{mn,kl}(\tau)}.
\end{equation}
Thus, uniformly in $\tau\in\mathcal T$,
\begin{align*}
    \hat\delta_k(\tau)-\delta_k(\tau)=&\hat G_k(\tau)\bar g^{(1)}_{mn}(\hat\delta,\tau)+\hat G_k(\tau)\bar g^{(2)}_{mn}(\hat\delta,\tau)\\
    =&G_{mn,k}(\tau)\left(\frac{1}{m}\sum_{j = 1}^m \Sigma_{ZXj}\left(\frac{1}{n}\sum_{i = 1}^n\phi_{j,\tau}(\tilde x_{ij},y_{ij})\right)+\frac{1}{m}\sum_{j=1}^m \bar z_{j}\alpha_j(\tau)\right)+o_p\left(\zeta(k,\tau)\right)
\end{align*}
Hence, we have
\begin{equation*}
    \hat\delta_k(\tau)-\delta_k(\tau)=\sum_{j=1}^m d_j(k,\tau)+o_p\left(\zeta(k,\tau)\right),
\end{equation*}
where
\begin{equation*}
    d_j(k,\tau)=G_{mn,k}(\tau)\left(\frac{1}{m}\Sigma_{ZXj}\left(\frac{1}{n}\sum_{i = 1}^n\phi_{j,\tau}(\tilde x_{ij},y_{ij})\right)+\frac{1}{m}\bar z_{j}\alpha_j(\tau)\right).
\end{equation*}
Let $D_j$ be the $TK\times 1$ vector $(\diag(\Sigma_{mn}(\tau_1))^{-1/2} d_j(\tau_1),\dots, \diag(\Sigma_{mn}(\tau_T))^{-1/2} d_j(\tau_T))'$ where $d_j(\tau)=(d_j(1,\tau),d_j(2,\tau),\dots,d_j(K,\tau))'$. It follows that
\begin{equation*}
\Var\left(\sum_{j=1}^m D_j\right)=H_{mn}
\end{equation*}
Then, from the proof of Lemma \ref{l:asymptotic moments} and Assumption \ref{a:kernel},
\begin{equation*}
H_{mn}^{-1/2}\sum_{j=1}^mD_j\underset{d}{\rightarrow}N(0,I_{TK}).
\end{equation*}
By Slutsky's theorem,
\begin{equation*}
H^{-1/2}\sum_{j=1}^mD_j=H_{mn}^{-1/2}\sum_{j=1}^mD_j+o_p(1)\underset{d}{\rightarrow}N(0,I_{TK})
\end{equation*}
In the proof of Lemma \ref{l:asymptotic moments} we show that $\sum_j \diag(\Sigma_{mn}(\cdot))^{-1/2} d_j(\cdot))$ is asymptotically tight in $\ell^\infty(\mathcal{T})$. It follows that the process $\sum_j \diag(\Sigma_{mn}(\cdot))^{-1/2} d_j(\cdot))$ weakly converges to $\mathbb Z$, a centered Gaussian process with covariance kernel $H(\tau,\tau')$.

Finally, we show that $o_p(1)\underset{\tau\in\mathcal T, k\in\{1,\dots,K\}}{\sup}\zeta(k,\tau)\Sigma_{mn,k}(\tau)^{-1/2}=o_p(1)$, where $\Sigma_{mn,k}(\tau)$ is the $(k,k)$ element of $\Sigma_{mn}(\tau)$. We consider separately the fast and slow coefficients For $k\in\{1,\dots,M_1\}$ (fast coefficients), note that
$$\Sigma_{mn,k}(\tau)^{1/2}=O_p\left(\frac{1}{\sqrt{mn}} + \frac{\sqrt{\Var(\alpha_j(\tau))}}{n\sqrt{m}}\right)=O_p\left(\frac{1}{\sqrt{mn}}\right),$$
and 
\begin{align*}
\zeta(k,\tau)&=O_p\left(\frac{1}{\sqrt{mn}}+\frac{\sqrt{\Var(\alpha(\tau))}}{\sqrt m}\sum_{l=L_1}^L\sqrt{G_{mn,kl}(\tau)}\right)\\
&=O_p\left(\frac{1}{\sqrt{mn}}+\frac{\sqrt{\Var(\alpha_j(\tau))}}{\sqrt{m}}\frac{1}{\sqrt{1+\Var(\alpha_j(\tau))n}}\right)\\
&=O_p\left(\frac{1}{\sqrt{mn}}\right),
\end{align*}
both uniformly in $\tau$. It follows that $o_p(1) \ \underset{\tau\in\mathcal T}{\sup}\ \zeta(k,\tau)\Sigma_{mn,k}(\tau)^{-1/2}=o_p(1)$.

For $k\in\{M_1+1,\dots,K\}$ (slow coefficients), note that
$$\Sigma_{mn,k}(\tau)^{1/2}=O_p\left(\frac{1}{\sqrt{mn}} + \frac{\sqrt{\Var(\alpha_j(\tau))}}{\sqrt{m}}\right),$$
and 
$$\zeta(k,\tau)=\frac{1}{\sqrt{mna_n(\tau)}}+\frac{\sqrt{\Var(\alpha(\tau))}}{\sqrt m}\sum_{l=L_1}^L\sqrt{G_{mn,kl}(\tau)}.$$
For the first term of $\zeta(k,\tau)$, we obtain
\begin{align*}
    \frac{1}{\sqrt{mna_n(\tau)}}\frac{1}{\Sigma_{mn,k}(\tau)^{1/2}}&=O_p\left(\frac{1}{\sqrt{mna_n(\tau)}}\frac{1}{\frac{1}{\sqrt{mn}}+\frac{\sqrt{\Var(\alpha_j(\tau))}}{\sqrt{m}}}\right)\\
    &=O_p\left(\frac{1}{\sqrt{a_n(\tau)}}\frac{1}{1+\sqrt{\Var(\alpha_j(\tau))n}}\right)\\
    &=O_p\left(\frac{\sqrt{1+\Var(\alpha_j(\tau))n}}{1+\sqrt{\Var(\alpha_j(\tau))n}}\right)\\
    &=O_p(1).
\end{align*}
For the second term of $\zeta(k,\tau)$, we obtain 
\begin{align*}
    \frac{\sqrt{\Var(\alpha_j(\tau))}}{\sqrt m}\sum_{l=L_1}^L\sqrt{G_{mn,kl}(\tau)}\Sigma_{mn,k(\tau)}^{-1/2}&=O_p\left(\frac{\sqrt{\Var(\alpha_j(\tau))}}{\sqrt m}\frac{1}{\sqrt{mn}+\frac{\sqrt{\Var(\alpha_j(\tau))}}{\sqrt m}}\right)\\
   &=O_p(1).
\end{align*}
It follows that, also in this second case, $o_p(1)\ \underset{\tau\in\mathcal T}  {\sup}\ \zeta(k,\tau)\Sigma_{mn,k}(\tau)^{-1/2}=o_p(1)$.

Hence, uniformly in $\tau$,
\begin{align*}
    \diag(\Sigma_{mn}(\tau))^{-1/2}(\hat\delta(\tau)-\delta(\tau))&=\diag(\Sigma_{mn}(\tau))^{-1/2}\left(\sum_{j=1}^m d_j(\tau)+o_p\left(\zeta(\tau)\right)\right)\\
    &=\diag(\Sigma_{mn}(\tau))^{-1/2} \sum_{j=1}^m d_j(\tau)+o_p(1)\\ &\rightsquigarrow \mathbb G(\tau),
\end{align*}
where $\zeta(\tau)=(\zeta(1,\tau),\dots,\zeta(K,\tau))'$.

\end{proof}

\subsection{Proof of Propositions \ref{p:omegamatrix} and \ref{p:omegamatrix2}: Properties of $\hat\Omega(\tau,\tau')$}

\subsubsection{Proposition \ref{p:omegamatrix}}

\begin{proof}[Proof of Proposition \ref{p:omegamatrix}]
    
\allowdisplaybreaks
We use $\hat u_{ij}(\tau) = \tilde x_{ij}' \hat \beta_j(\tau) - x_{ij}' \hat \delta (\tau)=  \tilde x_{ij}' (\hat \beta_j(\tau) -  \beta_j(\tau)  )  + x_{ij}' (\delta(\tau) -  \hat \delta(\tau))+ \alpha_j(\tau)$ to obtain
\begin{align*}
   \hat \Omega(\tau, \tau')&=\frac{1}{m} \sum_{j = 1}^m \left\{\left(\frac{1}{n}\sum_{i=1}^n z_{ij} \hat u_{ij}(\tau)\right)  \left(\frac{1}{n}\sum_{i=1}^n z_{ij}\hat u_{ij}(\tau')\right)'\right\}\\
    &=\frac{1}{m} \sum_{j = 1}^m  \Bigg\{ \left(\frac{1}{n}\sum_{i=1}^n z_{ij} \left[\tilde x_{ij}' (\hat \beta_j(\tau) -  \beta_j(\tau)  )  + x_{ij}' (\delta(\tau) -  \hat \delta(\tau))+ \alpha_j(\tau)\right]\right) \\ 
    & \hphantom{{} + \frac{1}{m} \sum_{j = 1}^m \Bigg\{ }\left(\frac{1}{n}\sum_{i=1}^n z_{ij}\left[\tilde x_{ij}' (\hat \beta_j(\tau') -  \beta_j(\tau')  )  + x_{ij}' (\delta(\tau') -  \hat \delta(\tau'))+ \alpha_j(\tau')\right]\right)'  \Bigg\}\\
    &=\frac{1}{m} \sum_{j = 1}^m  \Bigg\{\left ( \frac{1}{n}\sum_{i = 1}^n z_{ij}\tilde x_{ij}'(\hat \beta_j(\tau) -\beta_j(\tau) )   \right)   \left ( \frac{1}{n}\sum_{i = 1}^n z_{ij}\tilde x_{ij}'(\hat \beta_j(\tau') -\beta_j(\tau') ) \right )'\\
    & \hphantom{{} + \frac{1}{m} \sum_{j = 1}^m \Bigg\{ } + \left(\frac{1}{n}\sum_{i = 1}^n z_{ij}\alpha_j(\tau)\right) \left(\frac{1}{n}\sum_{i = 1}^n z_{ij}\alpha_j(\tau')\right)' \\
    & \hphantom{{} + \frac{1}{m} \sum_{j = 1}^m \Bigg\{ } +\left ( \frac{1}{n}\sum_{i = 1}^n z_{ij} x_{ij}'(\hat \delta(\tau) -\delta(\tau) )   \right)   \left ( \frac{1}{n}\sum_{i = 1}^n z_{ij} x_{ij}'(\hat \delta(\tau') -\delta(\tau') ) \right )' \\
    & \hphantom{{} + \frac{1}{m} \sum_{j = 1}^m \Bigg\{ } -\left ( \frac{1}{n}\sum_{i = 1}^n z_{ij} \tilde x_{ij}'(\hat \beta_j(\tau) -\beta_j(\tau) )   \right)   \left ( \frac{1}{n}\sum_{i = 1}^n z_{ij} x_{ij}'(\hat \delta(\tau') -\delta(\tau') ) \right )'\\
    & \hphantom{{} + \frac{1}{m} \sum_{j = 1}^m \Bigg\{ } -\left ( \frac{1}{n}\sum_{i = 1}^n z_{ij} x_{ij}'(\hat \delta(\tau) -\delta(\tau) )   \right)   \left ( \frac{1}{n}\sum_{i = 1}^n z_{ij}\tilde x_{ij}'(\hat \beta_j(\tau') -\beta_j(\tau') ) \right )' \\
    & \hphantom{{} + \frac{1}{m} \sum_{j = 1}^m \Bigg\{ } +\left ( \frac{1}{n}\sum_{i = 1}^n z_{ij} \tilde x_{ij}'(\hat \beta_j(\tau) -\beta_j(\tau) )   \right)   \left ( \frac{1}{n}\sum_{i = 1}^n z_{ij}\alpha_j(\tau') \right )'\\
    & \hphantom{{} + \frac{1}{m} \sum_{j = 1}^m \Bigg\{ } +\left ( \frac{1}{n}\sum_{i = 1}^n z_{ij} \alpha_j(\tau)   \right)   \left ( \frac{1}{n}\sum_{i = 1}^n z_{ij} \tilde x_{ij}'(\hat \beta_j(\tau') -\beta_j(\tau') ) \right )'\\
    & \hphantom{{} + \frac{1}{m} \sum_{j = 1}^m \Bigg\{ } -\left ( \frac{1}{n}\sum_{i = 1}^n z_{ij} x_{ij}'(\hat \delta(\tau) -\delta(\tau) )   \right)   \left ( \frac{1}{n}\sum_{i = 1}^n z_{ij} \alpha_j(\tau') \right )' \\
    & \hphantom{{} + \frac{1}{m} \sum_{j = 1}^m \Bigg\{ } -\left ( \frac{1}{n}\sum_{i = 1}^n z_{ij} \alpha_j(\tau)   \right)   \left ( \frac{1}{n}\sum_{i = 1}^n z_{ij} x_{ij}'(\hat \delta(\tau') -\delta(\tau') ) \right )' \Bigg\}.
\end{align*}\allowdisplaybreaks[0]
We will show that the first term converges to $\Omega_1(\tau, \tau')/n$, the second term to $\Omega_2(\tau,\tau')$, and the remaining terms vanish at a rate faster than the leading term.

By the proof of Lemma \ref{l:asymptotic moments}(i), it follows for the first term that
\begin{align*}
\frac{1}{m} \sum_{j = 1}^m \Bigg( \frac{1}{n}\sum_{i = 1}^n &z_{ij}\tilde x_{ij}'(\hat \beta_j(\tau) -\beta_j(\tau) )   \Bigg)   \left ( \frac{1}{n}\sum_{i = 1}^n z_{ij}\tilde x_{ij}'(\hat \beta_j(\tau') -\beta_j(\tau') ) \right ) '  \\ 
= & \bE \left [ \left ( \Sigma_{ZXj} \sumtT \phi_{i,\tau}(\tilde x_{ij}, z_{ij})   \right)   \left ( \Sigma_{ZXj} \sumtT \phi_{i,\tau'}(\tilde x_{ij}, z_{ij}) \right )'  \right]  +  o_p \left ( \frac{1}{mn}  \right ) \\ 
=& \frac{\Omega_1(\tau, \tau')}{n} + O_p \left ( \frac{1}{mn}  \right ),
\end{align*}
uniformly in $\tau$.
For the second term, we consider each element of the $L\times L$ matrix separately. For $l,l'\in\{1,\dots,L\}$, we have
\begin{align*}
\frac{1}{m} \sum_{j = 1}^m  \left(\frac{1}{n}\sum_{i = 1}^n z_{ijl}\alpha_j(\tau)\right) \left(\frac{1}{n}\sum_{i = 1}^n z_{ijl'}\alpha_j(\tau')\right)'&=\frac{1}{m} \sum_{j = 1}^m \bar z_{jl}\bar z_{jl'}\alpha_j(\tau)\alpha_j(\tau').
\end{align*}
If either $l$ or $l'$ (or both) is in $ \{1,\dots,L_1\}$, then $\Omega_{2ll'}(\tau,\tau')=0$ and $\sum_{j = 1}^m \bar z_{jl}\bar z_{jl'}\alpha_j(\tau)\alpha_j(\tau')=0$. Thus, we need only to consider the case that both $\bar z_{jl}$ and $\bar z_{jl'}$ are not equal to $0$ uniformly across all groups. We apply Theorem 9.2 in \cite{hansen2022probability}. His condition (9.3) is satisfied by the boundedness of $z_{ij}$ in Assumption \ref{a:instruments}(i) and the uniform boundedness of the $4+\varepsilon_C$ moment of $\alpha_j(\tau)$ in Assumption \ref{a:group effects}(i). Condition (9.5) in \cite{hansen2022probability} is satisfied by the assumption in equation (\ref{a:variance}).
It follows that
\begin{equation*}
\sqrt m\left(\frac{1}{m} \sum_{j = 1}^m \bar z_{jl}\bar z_{jl'}\alpha_j(\tau)\alpha_j(\tau')-\Omega_{2ll'}(\tau,\tau')\right)\underset{d}{\rightarrow}N(0,C_{l,l'}(\tau,\tau')).
\end{equation*}

Then, since the condition for Theorem 18.3 in \cite{hansen2022probability} are satisfied, we have that $\sup_{\tau, \tau'} \left \rVert \frac{1}{m} \sum_{j = 1}^m \bar z_{jl}\bar z_{jl'}\alpha_j(\tau)\alpha_j(\tau') -\Omega_{2ll'}(\tau,\tau')\right \rVert = O_p\left(\frac{\sqrt{C_{l,l'}(\tau,\tau')}}{\sqrt m} \right)$.
Let $C_l$ denotes a uniform bound on $|z_{ijl}|$.
From the definition of $C_{l,l'}(\tau,\tau')$ we have 
\begin{equation*}
    C_{l,l'}(\tau,\tau')\leq C_l^2C_{l'}^2 \Var(\alpha_j(\tau)\alpha_j(\tau'))\leq C_l^2C_{l'}^2\bE[\alpha_j(\tau)^2\alpha_j(\tau')^2]\leq C_l^2C_{l'}^2\sqrt{\bE[\alpha_j(\tau)^4]\bE[\alpha_j(\tau')^4]} 
\end{equation*}
Finally, note that $\frac{\bE[\alpha_j(\tau)^4]}{Var(\alpha_j(\tau))^2}$ is bounded. It follows that $\frac{1}{m} \sum_{j = 1}^m \bar z_{jl}\bar z_{jl'}\alpha_j(\tau)\alpha_j(\tau')=\Omega_{2ll'}(\tau,\tau')+O_p\left(\frac{\sqrt{\Var(\alpha_j(\tau))\Var(\alpha_j(\tau'))}}{\sqrt m} \right)=\Omega_{2ll'}(\tau,\tau')+O_p\left(\frac{\sqrt{\Omega_{2,ll}(\tau,\tau)\Omega_{2,l'l'}(\tau'\tau')}}{\sqrt m} \right)$.

For the third term, we also consider each element $l,l'$ of the matrix separately:
\begin{multline*}
\frac{1}{m} \sum_{j = 1}^m  \left ( \frac{1}{n}\sum_{i = 1}^n z_{ijl} x_{ij}'(\hat \delta(\tau) -\delta(\tau) )   \right)   \left ( \frac{1}{n}\sum_{i = 1}^n z_{ijl'} x_{ij}'(\hat \delta(\tau') -\delta(\tau') ) \right )=\\
\frac{1}{m} \sum_{j = 1}^m  \left ( \frac{1}{n}\sum_{i = 1}^n \sum_{k=1}^K z_{ijl} x_{ijk}(\hat \delta_k(\tau) -\delta_k(\tau) )   \right)   \left ( \frac{1}{n}\sum_{i = 1}^n \sum_{k=1}^K z_{ijl'} x_{ijk}(\hat \delta_k(\tau') -\delta_k(\tau') ) \right ).
\end{multline*}
We split $\sum_{k=1}^K z_{ijl'} x_{ijk}(\hat \delta_k(\tau') -\delta_k(\tau') )$ into the the individual-level and group-level variables. Since the estimation error of the coefficients on the individual-level variables is $O_p(1/\sqrt{mn})$ and $z_{ijl}$ and $x_{ijk}$ are bounded, we obtain
\begin{align*}
    \sum_{k=1}^{K_1} \frac{1}{n}\sum_{i = 1}^n  z_{ijl} x_{1ijk}(\hat \delta_k(\tau) -\delta_k(\tau) )   & =  O_p \left (  \frac{1}{\sqrt {mn} }   \right ) ,
\end{align*}
uniformly in $\tau$.
For the group-level variables, we obtain
\begin{align*}
  \sum_{k=K_1 + 1}^{K}  \frac{1}{n}\sum_{i = 1}^n  z_{ijl} x_{2jk}(\hat \delta_k(\tau) -\delta_k(\tau) ) & =    \sum_{k=K_1+1}^K \bar z_{jl} x_{2jk}(\hat \delta_k(\tau) -\delta_k(\tau) ).
\end{align*}
It follows that $\sum_{k=K_1 + 1}^{K}  \frac{1}{n}\sum_{i = 1}^n  z_{ijl} x_{jk}(\hat \delta_k(\tau) -\delta_k(\tau) )=0$ if $\bar z_{jl}=0$ for all $j=1,\dots,m$, i.e. if $l\in \{1,\dots,L_1\}$. If $l\in\{L_1+1,\dots,L\}$, uniformly in $\tau$, $\sum_{k=K_1 + 1}^{K}  \frac{1}{n}\sum_{i = 1}^n  z_{ijl} x_{jk}(\hat \delta_k(\tau) -\delta_k(\tau) )=O_p\left(\frac{1}{\sqrt{mn}}+\frac{\sqrt{Var(\alpha_j(\tau))}}{\sqrt m}\right)$. Since $\Omega_{2,ll}(\tau)=Var(\bar z_{jl}\alpha_j(\tau))$, in both cases we can write $$\sum_{k=1}^{K}  \frac{1}{n}\sum_{i = 1}^n  z_{ijl} x_{jk}(\hat \delta_k(\tau) -\delta_k(\tau) )=O_p\left(\frac{1}{\sqrt{mn}}+\frac{\sqrt{\Omega_{2,ll} (\tau)}}{\sqrt m}\right)$$ uniformly in $\tau$.
Combining these results, we get uniformly in $\tau, \tau'$
\begin{align}
    \frac{1}{m} \sum_{j = 1}^m & \left ( \frac{1}{n}\sum_{i = 1}^n z_{ijl} x_{ij}'(\hat \delta(\tau) -\delta(\tau) )   \right)   \left ( \frac{1}{n}\sum_{i = 1}^n z_{ijl'} x_{ij}'(\hat \delta(\tau') -\delta(\tau') ) \right )\\
     & = O_p \left (\frac{\sqrt{ \Omega_{mn,ll}(\tau) \Omega_{mn,l'l'}(\tau')}}{ m}  \right ).
\end{align}

For the fourth term, similar arguments and the fact that  $\sup_\tau \left \rVert \beta_j(\tau)-\beta_j(\tau)\right \rVert=O_p\left(\frac{1}{\sqrt n}\right)$ imply that uniformly in $\tau, \tau'$
\begin{align}
     \frac{1}{m} \sum_{j = 1}^m &   \left ( \frac{1}{n}\sum_{i = 1}^n z_{ijl} \tilde x_{ij}'(\hat \beta_j(\tau) -\beta_j(\tau) )   \right)   \left ( \frac{1}{n}\sum_{i = 1}^n z_{ijl'} x_{ij}'(\hat \delta(\tau') -\delta(\tau') ) \right)  \\ & =  O_p \left (\frac{1}{\sqrt{m}n} + \frac{ \sqrt{\Omega_{2,l'l'}(\tau')}} {\sqrt{mn}} \right ).
\end{align}
Similarly, for the fifth term
\begin{align}
     \frac{1}{m} \sum_{j = 1}^m &   \left ( \frac{1}{n}\sum_{i = 1}^n z_{ijl'} x_{ij}'(\hat \delta(\tau') -\delta(\tau') ) \right)\left ( \frac{1}{n}\sum_{i = 1}^n z_{ijl} \tilde x_{ij}'(\hat \beta_j(\tau) -\beta_j(\tau) )   \right)  \\ & =  O_p \left (\frac{1}{\sqrt{m}n} + \frac{ \sqrt{\Omega_{2,ll}(\tau)}} {\sqrt{mn}} \right ),
\end{align}
uniformly in $\tau, \tau'$.
Thus, the sum of the fourth and fifth terms is
    \begin{align*}
    O_p \left ( \frac{1}{\sqrt{m}n}+\frac{\sqrt{ \Omega_{2,ll}(\tau)}}{ \sqrt{m}\sqrt{n}} + \frac{\sqrt{  \Omega_{2,l'l'}(\tau')}}{ \sqrt{m} \sqrt{n}} \right )=O_p \left (\frac{ \sqrt{\Omega_{mn,ll}(\tau)\Omega_{mn,l'l'}(\tau') }} {\sqrt{m}}\right ),\end{align*}
uniformly in $\tau$.

For the sixth term, we obtain, uniformly in $\tau, \tau'$ 
\begin{align*}
     \frac{1}{m} \sum_{j = 1}^m    \Biggl( \frac{1}{n}\sum_{i = 1}^n z_{ijl} & \tilde x_{ij}'(\hat \beta_j(\tau) -\beta_j(\tau) )   \Biggr)   \Biggl( \frac{1}{n}\sum_{i = 1}^n z_{ijl'}\alpha_j(\tau') \Biggr) \\ 
     & =     \frac{1}{m} \sum_{j = 1}^m   \left ( \frac{1}{n}\sum_{i = 1}^n z_{ijl}  \tilde x_{ij}'(\hat \beta_j(\tau) -\beta_j(\tau) )   \right)   \bar z_{jl'}\alpha_j(\tau') \\ 
     & = O_p \left ( \frac{ \sqrt{\Omega_{2,l'l'}(\tau') }} {\sqrt{mn}} \right ),
\end{align*}
and similarly, we can show that the seventh term is $O_p \left ( \frac{ \sqrt{\Omega_{2,ll}(\tau) }} {\sqrt{mn}} \right )$.

For the eighth term, we obtain
\begin{align*}
\frac{1}{m} \sum_{j = 1}^m \Biggl( \frac{1}{n}\sum_{i = 1}^n z_{ijl} x_{ij}'&(\hat \delta(\tau) -\delta(\tau) )   \Biggr)   \left ( \frac{1}{n}\sum_{i = 1}^n z_{ijl'}' \alpha_j(\tau') \right )  \\ 
& =  \frac{1}{m} \sum_{j = 1}^m \left ( \frac{1}{n}\sum_{i = 1}^n z_{ijl} x_{ij}'(\hat \delta(\tau) -\delta(\tau) )   \right) \bar z_{jl'} \alpha_j(\tau') \\
&= O_p\left(\frac{\sqrt{\Omega_{2,ll}(\tau)\Omega_{2,l'l'}(\tau')}}{m} +\frac{\sqrt{\Omega_{2,l'l'}(\tau')}}{m\sqrt n}\right ) ,
\end{align*} 
and the ninth term is $O_p\left(\frac{\sqrt{\Omega_{2,ll}(\tau)\Omega_{2,l'l'}(\tau')}}{m} +\frac{\sqrt{\Omega_{mn,l'l'}(\tau')}}{m\sqrt n}\right )$ such that the sum of the eight and ninth term is $O_p\left(\frac{\sqrt{\Omega_{mn,ll}(\tau)\Omega_{mn,l'l'}(\tau')}}{m}\right)$, where both results are uniformly in $\tau, \tau'$.

Combining all the terms, we find that for each $l$, $l'$ entry
    \begin{align*}
        \hat \Omega(\tau, \tau') =& \frac{\Omega_{1,ll'}(\tau, \tau')}{n} + \Omega_{2,ll'}(\tau,\tau') + O_p \left (  \frac{\sqrt{\Omega_{mn,ll}(\tau) \Omega_{mn,l'l' }(\tau)}}{\sqrt{m}} \right)\nonumber \\
        = &\Omega_{mn,ll'}(\tau, \tau') + o_p \left (  {\sqrt{\Omega_{mn,ll}(\tau) \Omega_{mn,l'l' }(\tau')}} \right),
    \end{align*}
    uniformly in $\tau, \tau'$.

\end{proof}

\subsubsection{Proposition \ref{p:omegamatrix}$'$}

When the coefficients of some individual-level variables converge at a slow rate, while those of others converge at a fast rate, the third term in the proof of Proposition \ref{p:omegamatrix} may not be $O_p\left(\frac{\sqrt{\Omega_{mn,ll}(\tau)\Omega_{mn,l'l'}(\tau')}}{m}\right)$. The slowly converging coefficients can introduce an estimation error in the variance of the faster moments that diminishes more slowly than the true value. Proposition \ref{p:omegamatrix}$'$, stated below, provides a more general result that does not assume the coefficients of individual-level variables converge at the $O_p(1/\sqrt{mn})$ rate. Consequently, $\hat\Omega(\tau,\tau')$ is consistent as long as all individual-level variable coefficients converge at the same rate. One example of this is the between estimator for individual-level variables. Another example is the 2SLS estimator applied to the random effects model. In this case, the weighting matrix is full rank, and all coefficients converge at the rate of  $1/\sqrt{mn}+\sqrt{\Var(\alpha_j(\tau))}/\sqrt{m}$. 

\begin{propositionp}{\ref{p:omegamatrix}$'$}\label{p:omegamatrix2}
Let assumptions \ref{a:sampling}-\ref{a:group effects} and \ref{a:growth condition}(c) hold. As $m\rightarrow\infty$, for each $l,l'\in \{1,\dots,L\}$ and uniformly in $\tau,\tau'\in\mathcal{T}^2$,
\begin{equation*}
    m^{-1}\sum_{j=1}^m \bE \left[\left(\bar z_{jl} \bar z_{jl'}\alpha_j(\tau)\alpha_j(\tau')-\Omega_{2ll'}(\tau,\tau')\right)^2\right]\rightarrow C_{l,l'}(\tau,\tau')<\infty
\end{equation*}
The estimator used to compute $\hat u_{ij}(\tau)$ satisfies 
\begin{equation*}\hat\delta(\tau)-\delta(\tau)= O_p\left(\frac{1}{\sqrt{mn}} + \frac{\sqrt{\Var(\alpha_j(\tau))}}{\sqrt m} \right)\end{equation*} 
uniformly in $\tau$. Then, for any $ll'$ entry of the $\hat \Omega(\tau, \tau')$ matrix with $l,l' \in \{1, \dots , L \}$ we have 
\begin{align*}
   \hat \Omega (\tau, \tau')=\Omega_{mn,ll'}(\tau, \tau') + o_p \left (\frac{1}{n}+\frac{\sqrt{\Var(\alpha_j(\tau))}}{\sqrt n}+\frac{\sqrt{\Var(\alpha_j(\tau'))}}{\sqrt n}+\sqrt{\Var(\alpha_j(\tau))\Var(\alpha_j(\tau'))} \right) 
\end{align*}
\end{propositionp}

\begin{proof}
The proof follows the same steps as the proof of Proposition \ref{p:omegamatrix}, but requires some modifications each time $\hat\delta(\tau)-\delta(\tau)$ is involved. This term is now $O_p\left(\frac{1}{\sqrt{mn}}+\frac{\Var(\alpha_j(\tau))}{\sqrt m} \right)$ instead of $O_p\left(\frac{1}{\sqrt{mn}}+\frac{\Omega_{2}(\tau)}{\sqrt m} \right)$. For the third term, we obtain
\begin{align*}
    \frac{1}{m} \sum_{j = 1}^m & \left ( \frac{1}{n}\sum_{i = 1}^n z_{ijl} x_{ij}'(\hat \delta(\tau) -\delta(\tau) )   \right)   \left ( \frac{1}{n}\sum_{i = 1}^n z_{ijl'}' x_{ij}'(\hat \delta(\tau') -\delta(\tau') ) \right )\\
     & = O_p\left(\frac{1}{mn} + \frac{\sqrt{\Var(\alpha_j(\tau))\Var(\alpha_j(\tau'))}}{m} + \frac{\sqrt{\Var(\alpha_j(\tau))}}{\sqrt{n }m} + \frac{\sqrt{\Var(\alpha_j(\tau'))}}{\sqrt{n }m}\right).
\end{align*}
The sum of the fourth and fifth terms is now
\begin{align*}
O_p\left(\frac{1}{n\sqrt m}+\frac{\sqrt{\Var(\alpha_j(\tau))}}{\sqrt{n m}}+\frac{\sqrt{\Var(\alpha_j(\tau'))}}{\sqrt{n m}}\right).
\end{align*}
The sum of the eight and ninth terms is now $$O_p\left(\frac{\sqrt{\Omega_{2,ll}(\tau)}}{m\sqrt n}+\frac{\sqrt{\Omega_{2,l'l'}(\tau')}}{m\sqrt n}+\frac{\sqrt{\Var(\alpha(\tau'))\Omega_{2,ll}(\tau)}}{m}+\frac{\sqrt{\Var(\alpha(\tau))\Omega_{2,l'l'}(\tau')}}{m}\right).$$
The result of the lemma follows.
\end{proof}

\subsection{Proof of Proposition \ref{p:cov matrix}: Adaptive Inference}

\begin{proof}[Proof of Proposition \ref{p:cov matrix}]
    
We prove the results for $T = 2$ as the proof trivially extends to $T > 2$. We consider the $2K\times 1$ coefficient vector $\hat \delta(\tau, \tau')=(\hat\delta(\tau)',\hat\delta(\tau')')'$ 
whose asymptotic covariance matrix is
\begin{align*}
    \Sigma_{mn} = \begin{pmatrix}
        \Sigma_{mn}(\tau) & \Sigma_{mn}(\tau, \tau') \\ \Sigma_{mn}(\tau, \tau') & \Sigma_{mn}( \tau')
    \end{pmatrix}.
\end{align*}
This covariance matrix is estimated by
\begin{align*} \hat  \Sigma = & \frac{1}{m}\begin{pmatrix}
        \hat G(\tau) & 0\\ 0 & \hat G(\tau')
    \end{pmatrix} 
        \begin{pmatrix}
       \hat \Omega(\tau, \tau) & \hat \Omega(\tau,\tau') \\
        \hat \Omega(\tau', \tau) & \hat \Omega(\tau')
    \end{pmatrix} \begin{pmatrix}
        \hat G(\tau) & 0\\ 0 & \hat G(\tau')
    \end{pmatrix} ' \\= &\frac{1}{m} \begin{pmatrix}
        \hat G(\tau) \hat \Omega(\tau) \hat G(\tau) &  \hat G(\tau) \hat \Omega(\tau,\tau') \hat G(\tau') \\ \hat G(\tau') \hat \Omega(\tau',\tau) \hat G(\tau) &  \hat G(\tau') \hat \Omega(\tau') \hat G(\tau')
    \end{pmatrix}.
\end{align*}
For each $k,k'\in\{1,\dots,2K\}$, we will show that \begin{equation*}\hat\Sigma_{kk'}=\Sigma_{mn,kk'}+o_p\left(\sqrt{\Sigma_{mn,kk}\Sigma_{mn,k'k'}}\right).\end{equation*}
To simplify the notation, we show this result for the $K\times K$ top-left submatrix of $\Sigma_{mn}$ such that we can drop the dependence on $\tau$. The proof for the other parts is similar. Our analysis relies on two key ingredients: 

First, by Proposition \ref{p:omegamatrix}, for any $ll'$ entry of the $\hat \Omega$ matrix with $l,l' \in \{1, \dots , L \}$, we have 
\begin{align*}
        \hat \Omega_{ll'}= \Omega_{mn,ll'} +  o_p \left ( 
 \sqrt{ \Omega_{mn,ll} \Omega_{mn,l'l'} }   \right )
\end{align*}
where $\Omega$ can be split into four submatrices where the dimensions of the top-left part are $L_1\times L_1$ and the bottom-right are $L_2\times L_2$:
\begin{equation*}
    \Omega=\begin{pmatrix}
        \Omega_{mn,11} & \Omega_{mn,12} \\
        \Omega_{mn,21} & \Omega_{mn,22}
    \end{pmatrix}
\end{equation*}
such that $\Omega_{mn,11}$, $\Omega_{mn,12}$, and $\Omega_{mn,21}$ are $O_p\left(\frac{1}{n}\right)$ while $\Omega_{mn,22}=O_p\left(\frac{1}{na_n}\right)$.

Second, by Lemma \ref{l:consistency G mat uniform}, \begin{equation*}
\hat G= \begin{pmatrix}
        G_{mn,11} & G_{mn,12} \\ G_{mn21} & G_{mn,22}
    \end{pmatrix}+ \begin{pmatrix}
        o_p\left(1\right) & o_p \left (\sqrt{a_n}\right ) \\ o_p \left (1/\sqrt{a_n} \right) & o_p\left(1\right)
    \end{pmatrix},
\end{equation*}
where $G_{mn,12}=O_p\left(a_n\right)$ and the other elements of $G_{mn}$ are $O_p\left(1\right)$.

With this notation, 
\begin{align*}
\Sigma_{mn}&=\frac{1}{m}
    \begin{pmatrix}
        G_{mn,11} & G_{mn,12}\\
        G_{mn,21} & G_{mn,22}
    \end{pmatrix}
    \begin{pmatrix}
        \Omega_{mn,11} & \Omega_{mn,12}\\
        \Omega_{mn,21} & \Omega_{mn,22}
    \end{pmatrix}
    \begin{pmatrix}
        G_{mn,11}' & G_{mn,21}'\\
        G_{mn,12}' & G_{mn,22}'
    \end{pmatrix}\\
    &=\begin{pmatrix}
        \Sigma_{mn,11} & \Sigma_{mn,12}\\
        \Sigma_{mn,21} & \Sigma_{mn,22}
    \end{pmatrix},
\end{align*}
where \begin{align*}
\Sigma_{mn,11}&= \frac{1}{m} \left (G_{mn,11}\Omega_{mn,11}G_{mn,11}'+G_{mn,11}\Omega_{mn,12}G_{mn,12}'+G_{mn,12}\Omega_{mn,21}G_{mn,11}'+G_{mn,12}\Omega_{mn,22}G_{mn,12}'\right),\\
\Sigma_{mn,12}&=\frac{1}{m} \left (G_{mn,11}\Omega_{mn,11}G_{mn,21}'+G_{mn,11}\Omega_{mn,12}G_{mn,22}'+G_{mn,12}\Omega_{mn,21}G_{mn,21}'+G_{mn,12}\Omega_{mn,22}G_{mn,22}'\right),\\
\Sigma_{mn,21}&=\frac{1}{m} \left (G_{mn,21}\Omega_{mn,11}G_{mn,11}'+G_{mn,21}\Omega_{mn,12}G_{mn,12}'+G_{mn,22}\Omega_{mn,21}G_{mn,11}'+G_{mn,22}\Omega_{mn,22}G_{mn,12}' \right),\\
\Sigma_{mn,22}&=\frac{1}{m} \left (G_{mn,21}\Omega_{mn,11}G_{mn,21}'+G_{mn,21}\Omega_{mn,12}G_{mn,22}'+G_{mn,22}\Omega_{mn,21}G_{mn,21}'+G_{mn,22}\Omega_{mn,22}G_{mn,22}'\right).
\end{align*}
It follows that $\Sigma_{mn,11}=O_p\left(\frac{1}{mn}\right)$, $\Sigma_{mn,12}=O_p\left(\frac{1}{mn}\right)$, $\Sigma_{mn,21}=O_p\left(\frac{1}{mn}\right)$, and $\Sigma_{mn,22}=O_p\left(\frac{1}{mna_n}\right)$.

We consider the estimation error of these four terms separately, each of them being composed of four parts. For the first term,
\begin{align*}
    \hat G_{11}\hat \Omega_{11}\hat G_{11}'
    &=(G_{mn,11}+o_p\left(1\right))(\Omega_{mn,11}+o_p\left(\Omega_{mn,11}\right))(G_{mn,11}'+o_p\left(1\right))\\
    &=G_{mn,11}\Omega_{mn,11}G_{mn,11}'+o_p(\Omega_{mn,11})
    \\
    &=G_{mn,11}\Omega_{mn,11}G_{mn,11}'+o_p\left(\frac{1}{n}\right),\\
    \hat G_{11}\hat \Omega_{12}\hat G_{12}'
    &=(G_{mn,11}+o_p\left(1\right))\left (\Omega_{mn,12}+o_p\left(\sqrt{\Omega_{mn,11}\Omega_{mn,22}}\right)\right)\left (G_{mn,12}'+o_p\left(\sqrt{a_n}\right) \right )\\   &=G_{mn,11}\Omega_{mn,12}G_{mn12}'+o_p\left(\Omega_{mn,12}\right)+o_p\left(\sqrt{\Omega_{mn,11}\Omega_{mn,22}a_n}\right)\\
    &=G_{mn,11}\Omega_{mn,12}G_{mn,12}'+o_p\left(\frac{1}{n}\right),\\
    \hat G_{12}\hat\Omega_{21}\hat G_{11}'&=(\hat G_{11}\hat \Omega_{12}\hat G_{12}')'\\
    &=G_{mn,12}\Omega_{mn,21}G_{mn,11}'+o_p\left(\frac{1}{n}\right),\\
    \hat G_{12}\hat \Omega_{22}\hat G_{12}'
    &=(G_{mn,12}+o_p\left(\sqrt{a_n}\right))(\Omega_{mn,22}+o_p\left(\Omega_{mn,22}\right))(G_{mn,12}'+o_p\left(\sqrt{a_n}\right))\\   &=G_{mn,12}\Omega_{mn,22}G_{mn,12}'+o_p\left(a_n\Omega_{mn,22}\right)\\
    &=G_{mn,12}\Omega_{mn,22}G_{mn,12}'+o_p\left(\frac{1}{n}\right).
\end{align*}
It follows that $\hat\Sigma_{11}=\Sigma_{mn,11}+o_p\left(\frac{1}{mn}\right)=\Sigma_{mn,11}+o_p\left(\Sigma_{mn,11}\right)$.

For the second part,
\begin{align*}
    \hat G_{11}\hat \Omega_{11}\hat G_{21}'
    &=(G_{mn,11}+o_p\left(1\right))\left (\Omega_{mn,11}+o_p\left(\Omega_{mn,11}\right)\right )(G_{mn,21}'+o_p\left(1/\sqrt{a_n}\right))\\    &=G_{mn,11}\Omega_{mn,11}G_{mn,21}'+o_p\left(\Omega_{mn,11}/\sqrt{a_n}\right)\\    &=G_{mn,11}\Omega_{mn,11}G_{mn,21}'+o_p\left(\sqrt{\Omega_{mn,11}\Omega_{mn,22}}\right),\\
    \hat G_{11}\hat \Omega_{12}\hat G_{22}'
    &=(G_{mn,11}+o_p\left(1\right))\left(\Omega_{mn,12}+o_p\left(\sqrt{\Omega_{mn,11}\Omega_{mn,22}}\right) \right)(G_{mn,22}'+o_p\left(1\right) )\\   &=G_{mn,11}\Omega_{mn,12}G_{mn,22}'+o_p\left(\sqrt{\Omega_{mn,11}\Omega_{mn,22}}\right),\\
    \hat G_{12}\hat\Omega_{21}\hat G_{21}'&=(G_{mn,12}+o_p\left(\sqrt{a_n}\right))(\Omega_{mn,21}+o_p\left(\sqrt{\Omega_{mn,11}\Omega_{mn,22}}\right)(G_{mn,21}'+o_p\left(1/\sqrt{a_n}\right))'\\    &=G_{mn,12}\Omega_{mn,21}G_{mn,21}'+o_p\left(\sqrt{\Omega_{mn,11}\Omega_{mn,22}}\right),\\
    \hat G_{12}\hat \Omega_{22}\hat G_{22}'
    &=(G_{mn,12}+o_p\left(\sqrt{a_n}\right))(\Omega_{mn,22}+o_p(\Omega_{mn,22}))(G_{mn,22}'+o_p\left(1\right))\\   &=G_{mn,12}\Omega_{mn,22}G_{mn,22}'+o_p\left(\sqrt{a_n}\Omega_{mn,22}\right)\\    &=G_{mn,12}\Omega_{mn,22}G_{mn,22}'+o_p\left(\sqrt{\Omega_{mn,11}\Omega_{mn,22}}\right).
\end{align*}
It follows that $\hat\Sigma_{12}=\Sigma_{mn,12}+o_p(m^{-1}\sqrt{\Omega_{mn,11}\Omega_{mn,22}})=\Sigma_{mn,12}+o_p(\sqrt{\Sigma_{mn,11}\Sigma_{mn,22}})$. The third part, $\hat\Sigma_{21}$, is the transpose of $\hat\Sigma_{12}$.

For the fourth part,
\begin{align*}
    \hat G_{21}\hat \Omega_{11}\hat G_{21}'
    &=(G_{mn,21}+o_p\left(1/\sqrt{a_n}\right))(\Omega_{mn,11}+o_p\left(\Omega_{mn,11}\right))\left(G_{mn,21}'+o_p\left(1/\sqrt{a_n}\right) \right)\\    &=G_{mn,21}\Omega_{mn,11}G_{mn,21}'+o_p\left(\Omega_{mn,11}/a_n\right)\\    &=G_{mn,21}\Omega_{mn,11}G_{mn,21}'+o_p\left(\Omega_{mn,22}\right),\\
    \hat G_{21}\hat \Omega_{12}\hat G_{22}'
    &=(G_{mn,21}+o_p\left(1/\sqrt{a_n}\right))\left(\Omega_{mn,12}+o_p\left(\sqrt{\Omega_{mn,11}\Omega_{mn,22}}\right) \right) \left (G_{mn,22}'+o_p\left(1\right) \right )\\   &=G_{mn,21}\Omega_{mn,12}G_{mn,22}'+o_p\left(\sqrt{\Omega_{mn,11}\Omega_{mn,22}/a_n}\right)\\    &=G_{mn,21}\Omega_{mn,12}G_{mn,22}'+o_p\left(\Omega_{mn,22}\right),\\
    \hat G_{22}\hat\Omega_{21}\hat G_{21}'&=(G_{mn,22}+o_p\left(1\right))\Omega_{mn,21}+o_p\left(\sqrt{\Omega_{mn,11}\Omega_{mn,22}}\right)(G_{mn,21}'+o_p\left(1/\sqrt{a_n}\right))'\\ &=G_{mn,22}\Omega_{mn,21}G_{mn,21}'+o_p\left(\sqrt{\Omega_{mn,11}\Omega_{mn,22}/a_n}\right)\\    &=G_{mn,22}\Omega_{mn,21}G_{mn,21}'+o_p\left(\Omega_{mn,22}\right),\\
    \hat G_{22}\hat \Omega_{22}\hat G_{22}'
    &=(G_{mn,22}+o_p\left(1\right))(\Omega_{mn,22}+o_p\left(\Omega_{mn,22}\right))(G_{mn,22}'+o_p\left(1\right))\\   &=G_{mn,22}\Omega_{mn,22}G_{mn,22}'+o_p\left(\Omega_{mn,22}\right).
\end{align*}
It follows that $\hat\Sigma_{22}=\Sigma_{mn,22}+o_p\left(m^{-1}\Omega_{mn,22}\right)=\Sigma_{mn,22}+o_p\left(\Sigma_{mn,22}\right)$. The results for these four submatrices imply that
\begin{equation*}\hat\Sigma_{kk'}=\Sigma_{mn,kk'}+o_p(\sqrt{\Sigma_{mn,kk}\Sigma_{mn,k'k'}}),\end{equation*}
which implies that
\begin{equation*}
\eta'\hat\Sigma\eta=\eta'\Sigma_{mn}\eta+o_p\left(\eta'\Sigma_{mn}\eta\right).
\end{equation*}
The proposition follows from Theorem \ref{t:adaptive}.

\end{proof}

\subsection{Proof of Proposition \ref{p:weight matrix adaptive}: Properties of the weighting matrix}
 \begin{proof}

Since this proof focuses on a case where $\tau = \tau'$, we suppress the dependency on $\tau$ for simplicity. 
 We partition the $\Omega $ matrix as follows: 
\begin{equation*}
       \Omega_{mn} = \begin{pmatrix}
         \Omega_{mn,11} &  \Omega_{mn,12} \\  \Omega_{mn,21} &  \Omega_{mn,22}
    \end{pmatrix}
\end{equation*}
where $\Omega_{mn,11}$ is $L_1 \times L_1, \Omega_{mn,12} $ is $L_1 \times L_2$, $\Omega_{mn,21}$ is $L_2 \times L_1$ and $\Omega_{mn,22}$ is $L_2\times L_2$.

From Proposition \ref{p:omegamatrix} we have, uniformly over $\tau$

\begin{align*}
   \hat \Omega   =      \begin{pmatrix}
       \frac{ \Omega_{1,11}}{n}  &  \frac{ \Omega_{1,12}}{n}   \\
        \frac{ \Omega_{1,21}}{n}    &   \Omega_{mn,22}   
    \end{pmatrix} + 
     \begin{pmatrix}
         o_p \left (   \parallel\Omega_{mn,11}\parallel  \right )  &  + o_p \left (  \frac{\sqrt{ \parallel\Omega_{mn,22}\parallel  }}{\sqrt{n}}\right )   \\
         o_p \left (  \frac{\sqrt{ \parallel\Omega_{mn,22}\parallel  }}{\sqrt{n}}\right )    &    o_p \left (   \parallel\Omega_{mn,22}\parallel  \right )  
    \end{pmatrix}.
\end{align*}
By the inverse of a partitioned matrix
\begin{equation}\label{eq:Wmat}
 \hat W = \frac{1}{n} \cdot   \begin{pmatrix}
        \hat \Omega_{11} & \hat \Omega_{12} \\ \hat \Omega_{21} & \hat \Omega_{22}
    \end{pmatrix}^{-1} =  \begin{pmatrix}
       \hat  \Psi & - \hat \Psi \hat \Omega_{12} \hat \Omega_{22}^{-1} \\ -  \hat \Omega_{22}^{-1} \hat \Omega_{21} \hat \Psi & n^{-1} \cdot \hat \Omega_{22}^{-1} + \hat \Omega_{22}^{-1} \hat \Omega_{21} \Psi \hat \Omega_{12} \hat \Omega_{22}^{-1}
    \end{pmatrix}
\end{equation}
where $\hat \Psi = ( n \cdot \hat \Omega_{11} - n \cdot \hat \Omega_{12} \hat \Omega_{22}^{-1}\hat \Omega_{21} ) ^{-1}$.
To address the possibility that $\Omega_{mn,22}$ could converge to zero, we rescale both elements in $\hat \Omega_{12} \hat \Omega_{22}^{-1}$, which ensures that both elements are well behaved, and show that the expression is self-normalizing. Hence, uniformly in $\tau$,
\begin{align*}
  \hat \Omega_{12} \hat \Omega_{22}^{-1}  =  & \left (  \parallel \Omega_{mn,22} \parallel^{-1}\hat \Omega_{12} \right) \left ( \parallel \Omega_{mn,22} \parallel^{-1}   \hat \Omega_{mn,22} \right) ^{-1} \\ =&  \Omega_{mn,12}  \Omega_{mn,22}^{-1} +   o_p \left (\sqrt{ n^{-1} \parallel \Omega_{mn,22} \parallel^{-1}} \right)\\  =& n^{-1}   \Omega_{1,12}  \Omega_{mn,22}^{-1} + o_p(\sqrt{a_n}).
\end{align*}
Note that the second line follows as
\begin{align*}
    \parallel \Omega_{mn,22} \parallel^{-1} \hat \Omega_{12} = &   { n^{-1}\parallel \Omega_{mn,22} \parallel^{-1} \Omega_{1,12}} + o_p \left (\sqrt{ n^{-1} \parallel \Omega_{mn,22} \parallel^{-1}} \right),\\
     \parallel \Omega_{mn,22} \parallel^{-1} \hat \Omega_{22} =&  \parallel \Omega_{mn,22} \parallel^{-1} \Omega_{mn,22} + o_p(1),
\end{align*}   
 and $ \parallel \Omega_{mn,22} \parallel^{-1}   \hat \Omega_{22}  $ is invertible.
It then follows that, uniformly in $\tau$,
\begin{align*} 
 \hat \Omega_{11} -  \hat \Omega_{12} \hat \Omega_{22}^{-1}\hat \Omega_{21}
 = & 
 \frac{\Omega_{1,11}}{n} + o_p \left( \frac{1}{n}\right) - \left  (n^{-1}   \Omega_{mn,12}  \Omega_{mn,22}^{-1} +   o_p \left (\sqrt{ n^{-1} \parallel \Omega_{mn,22} \parallel^{-1}} \right) \right )\\ & \left(    \frac{ \Omega_{1,21}}{n} + o_p \left (  \frac{\sqrt{\parallel \Omega_{mn,22}\parallel}}{\sqrt{n}}  \right )\right  ) \\ &
= \frac{\Omega_{1,11}}{n}  - \frac{ \Omega_{1,12}}{n}\Omega_{mn,22}^{-1}\frac{ \Omega_{1,21}}{n} + o_p \left (\frac{1}{n} \right)  .
\end{align*}
Hence, uniformly in $\tau$, 
\begin{align*}
\hat \Psi^{-1} = n \cdot  \hat \Omega_{11} - n \cdot  \hat \Omega_{12} \hat \Omega_{22}^{-1}\hat \Omega_{21} = {\Omega_{1,11}}  - \frac{1}{n} \Omega_{1,12}\Omega_{mn,22}^{-1}{ \Omega_{1,21}} + o_p \left (1 \right) ,
 \end{align*}
 so that for the submatrix in the upper left of equation (\ref{eq:Wmat}), we have 
 $$ \sup_{\tau } \left \rVert \hat \Psi - \Psi \right \rVert= o_p(1) ,$$
 where $\Psi$ is strictly positive definite.\\
For the upper right term of $\hat W$, we have uniformly in $\tau$
\begin{align*}
    -  \hat \Psi \hat \Omega_{12} \hat \Omega_{22}^{-1} = -  \Psi  \frac{\Omega_{1,12}}{n}  \Omega_{mn,22}^{-1} + o_p(\sqrt{a_n}),
\end{align*}
where $ \Psi  \frac{\Omega_{1,12}}{n}  \Omega_{mn,22}^{-1} = O_p(a_n)$.\\
Similarly, 
$$ \hat \Omega_{22}^{-1} \hat \Omega_{21} \hat \Psi =   \Omega_{mn,22}^{-1}  \frac{\Omega_{1,21}}{n}  \Psi + o_p(\sqrt{a_n}).$$
Finally, for the lower right term, using an appropriate normalization, we find that uniformly in $\tau$
\begin{align*}
    n^{-1 }{\hat \Omega_{22}}^{-1} = n^{-1} \Omega_{mn,22}^{-1} + o_p \left ( \frac{\parallel \Omega_{mn,22} \parallel^{-1}}{n} \right) = n^{-1}\Omega_{mn,22}^{-1}  + o_p\left (a_n  \right),
\end{align*} 
and combining this with the results for the other components for the lower right submatrix, we obtain
$$ \sup_\tau\left \rVert n^{-1} \cdot \hat \Omega_{22}^{-1} + \hat \Omega_{22}^{-1} \hat \Omega_{21} \Psi \hat \Omega_{12} \hat \Omega_{22}^{-1} - n^{-1} \cdot  \Omega_{mn,22}^{-1} +  \Omega_{mn,22}^{-1}  \frac{\Omega_{1,21}}{n} \Psi  \frac{\Omega_{1,12}}{n}  \Omega_{mn,22}^{-1} \right \rVert =  o_p({a_n}) . $$ 
Hence, uniformly in $\tau\in\mathcal T$,
\begin{equation*}
    \hat W(\tau) = \begin{pmatrix}
        W_{11}(\tau) & a_n(\tau) W_{12}(\tau) \\ a_n(\tau) W_{21}(\tau) & a_n(\tau) W_{22}(\tau)
    \end{pmatrix} + \begin{pmatrix}
        o_p(1) & o_p \left (\sqrt{a_n(\tau)}\right ) \\ o_p \left (\sqrt{a_n(\tau)} \right ) & o_p({a_n}(\tau))
    \end{pmatrix}.
\end{equation*}

 \end{proof}

\subsection{Proof of Proposition \ref{prop:hausmant test}: Overidentification Test}
\begin{proof}[Proof of Proposition \ref{prop:hausmant test}]
We prove this result for $T =1$ as the proof trivially extends to $T > 1$. We suppress the dependency on $\tau$ for simplicity as this proof focuses on a case where $\tau = \tau'$. 
First, we want to rewrite the J-statistics in a way that accounts for the different convergence rates of the moment conditions. Let $d_{\Omega, mn} = I_L \cdot  \diag (\Omega_{mn})$. Then we can write
\begin{align*}
    J(\hat \delta) &=  m \bar g_{mn}( \hat \delta)' \hat \Omega^{-1} \bar g_{mn}( \hat \delta)
    \\ &=   \left (  d_{\Omega,mn}^{-1/2} \sqrt{m}\bar g_{mn}( \hat \delta) \right )'  \left [ d_{\Omega,mn}^{-1/2} \hat \Omega d_{\Omega,mn}^{-1/2} \right ]^{-1}  \left ( d_{\Omega,mn}^{-1/2} \sqrt{m} \bar g_{mn}( \hat \delta).  \right)
\end{align*}

Second, we want to show that for some matrix $\hat B$, $\bar g_{mn}(\hat \delta) = \hat B  \bar g_{mn}(\delta)$. Recall that $\hat Y_j = X_j \delta + \alpha_j + \tilde X_j ( \hat \beta_j - \beta_j)$. Hence, we can write
\begin{align*}
   Z_j' \hat Y_j &=Z_j' X_j \delta + Z_j'\alpha_j +  Z_j'\tilde X_j ( \hat \beta_j - \beta) \\ 
   S_{Z\hat Y} &= S_{ZX} \delta + \sumiN \sumtT z_{ij} \alpha_j + \sumiN \sumtT z_{ij} \tilde x_{ij}' ( \hat \beta_j - \beta_j)
   \\ 
    &= S_{ZX} \delta + \bar g_{mn}(\delta) .
\end{align*}
\\
Then, note that
\begin{align*}
    \bar g_{mn}( \hat \delta ) &= \sumiN \sumtT z_{ij} (\hat y_{ij} - x_{ij}'\hat \delta) \\
    &= S_{Z \hat Y} - S_{ZX} \hat \delta \\
    &= S_{Z \hat Y} - S_{ZX} \left ( S_{ZX} \hat \Omega^{-1} S_{ZX}\right )^{-1}S_{ZX} \hat \Omega^{-1}S_{Z \hat Y} = \hat B S_{Z \hat Y},
\end{align*}
where $\hat B =\left (I_L - S_{ZX} \left ( S_{ZX}' \hat \Omega^{-1} S_{ZX} \right )^{-1} S_{ZX}' \hat \Omega^{-1} \right )$.\\
Thus,
\begin{align*}
    \bar g_{mn} ( \hat \delta) &= \hat B S_{Z\hat Y} \\ 
    &= \left (I_L - S_{ZX} \left ( S_{ZX}' \hat \Omega^{-1} S_{ZX} \right )^{-1} S_{ZX}' \hat \Omega^{-1} \right ) \left ( S_{ZX} \delta + \bar g_{mn}(\delta) \right ) \\ &= \hat B  \bar g_{mn}(\delta)  .
\end{align*}
 Since $\left [ d_{\Omega,mn}^{-1/2}  \hat\Omega d_{\Omega,mn}^{-1/2} \right ] $ is positive definite there exist a matrix $\hat Q$ such that $$\left [ d_{\Omega,mn}^{-1/2} \hat \Omega d_{\Omega,mn}^{-1/2} \right ]^{-1} =  \hat Q'\hat Q.$$
We define $\hat A = \hat Q d_{\Omega,mn}^{-1/2} S_{ZX}'$ and $ \hat M = I_L - \hat A(\hat A'\hat A)^{-1}\hat  A' $.

In this third part, we show that 
\begin{align*}
    \hat B' d_{\Omega,mn}^{-1/2} \left [ d_{\Omega,mn}^{-1/2} \hat \Omega d_{\Omega,mn}^{-1/2} \right ]^{-1} d_{\Omega,mn}^{-1/2} \hat  B = d_{\Omega,mn}^{-1/2} \hat Q ' \hat M \hat Q d_{\Omega,mn}^{-1/2} ,
\end{align*}
where $ \hat B' d_{\Omega,mn}^{-1/2} \left [ d_{\Omega,mn}^{-1/2} \hat \Omega d_{\Omega,mn}^{-1/2} \right ]^{-1} d_{\Omega,mn}^{-1/2} \hat  B =  \hat B' d_{\Omega,mn}^{-1/2} \hat Q' \hat Q  d_{\Omega,mn}^{-1/2} \hat  B  $.
Note that
\begin{align*}
   \hat Q  d_{\Omega,mn}^{-1/2} \hat B & = \hat Q d_{\Omega,mn}^{-1/2}  \left (I_L - S_{ZX} \left ( S_{ZX}' \hat \Omega^{-1} S_{ZX} \right )^{-1} S_{ZX}' \hat \Omega^{-1} \right )
   \\ &=    \left (\hat Q d_{\Omega,mn}^{-1/2} - \hat Q d_{\Omega,mn}^{-1/2} S_{ZX} \left ( S_{ZX}' \hat \Omega^{-1} S_{ZX} \right )^{-1} S_{ZX}' \hat \Omega^{-1} \right )
   \\ &=  \left (\hat Q  d_{\Omega,mn}^{-1/2} - \hat Q d_{\Omega,mn}^{-1/2} S_{ZX} \left ( S_{ZX}' d_{\Omega,mn}^{-1/2} \hat Q'\hat Q d_{\Omega,mn}^{-1/2} S_{ZX} \right )^{-1} S_{ZX}' d_{\Omega,mn}^{-1/2} \hat Q' \hat Q d_{\Omega,mn}^{-1/2} \right )
          \\ &=  \left (I_L - \hat A \left (\hat A'\hat A \right )^{-1} \hat A' \right )\hat Q d_{\Omega,mn}^{-1/2}
           \\ &= \hat M \hat Q d_{\Omega,mn}^{-1/2}.
\end{align*}
Where the third line uses $ d_{\Omega,mn}^{1/2} \hat \Omega^{-1}  d_{\Omega,mn}^{1/2} = \hat Q' \hat Q$ and in the last two lines, we use the definitions of $\hat A$ and $\hat M$. \\
$\hat M$ is symmetric and idempotent. Thus
\begin{align*}
    \hat B' d_{\Omega,mn}^{-1/2} \hat \Omega^{-1} d_{\Omega,mn}^{-1/2} \hat  B &= \hat B' d_{\Omega,mn}^{-1/2} \hat Q'\hat Q d_{\Omega,mn}^{-1/2} \hat  B \\ &= \left ( \hat Q d_{\Omega,mn}^{-1/2} \hat  B \right)'\hat Q d_{\Omega,mn}^{-1/2} \hat  B \\ & = \left (  \hat M \hat Q d_{\Omega,mn}^{-1/2} \right)' \hat M \hat Q d_{\Omega,mn}^{-1/2} \\ &= d_{\Omega,mn}^{-1/2} \hat Q' \hat M \hat Q d_{\Omega,mn}^{-1/2} .
\end{align*}
The rank of $\hat M$ is the trace of $\hat M$, which is $L - K$.
\\
Since $ d_{\Omega,mn}^{-1/2} \Omega_{mn} d_{\Omega,mn}^{-1/2}$ is positive definite, there exist a matrix $Q_{mn}$ such that
$$ Q_{mn}'Q_{mn} = \left [ d_{\Omega,mn}^{-1/2} \Omega_{mn} d_{\Omega,mn}^{-1/2}\right]^{-1}.$$ It is easy to show that

\begin{align*}
    d_{\Omega}^{1/2} \hat \Omega^{-1}  d_{\Omega}^{1/2} = d_{\Omega}^{1/2}  \Omega_{mn}^{-1}  d_{\Omega}^{1/2} + o_p(1),
\end{align*}
where
\begin{align*}
    \left (d_{\Omega}^{-1/2}  \Omega_{mn}  d_{\Omega}^{-1/2} \right) ^{-1}=  
      & \begin{pmatrix}
        1 
        & \frac{\Omega_{1,12}}{\sqrt{n} \sqrt{\Omega_{1,11} \Omega_{22}}}  
        \\ 
      \frac{\Omega_{1,21}}{\sqrt{n} \sqrt{ \Omega_{22} \Omega_{1,11}}} 
        & 1
    \end{pmatrix} .
\end{align*}
It then follows directly that 
$\hat Q = Q_{mn} + o_p(1)$.

Further, we have that
\begin{equation*}
d_{\Omega,mn}^{-1/2} \sqrt{m}  \bar g_{mn}(\delta)  \xrightarrow{d} N \left ( 0,  d_{\Omega,mn}^{-1/2} \Omega_{mn}  d_{\Omega,mn}^{-1/2}  \right) .
\end{equation*}   
\begin{commentP}
\begin{align*}
    \left( \hat Q - Q_{mn} \right) d_{\Omega,mn}^{-1/2}  \sqrt{m} \bar g_{mn}(\delta)  
\end{align*}
\begin{small}
\begin{align*}
d_{\Omega}^{1/2} \hat \Omega^{-1}  d_{\Omega}^{1/2} =  & \begin{pmatrix}
        \left ( \frac{\Omega_{1,11}}{n} \right )^{-1/2}  \hat \Omega_{1,11}  \left ( \frac{\Omega_{1,11}}{n} \right )^{-1/2} 
        & \left ( \frac{\Omega_{1,11}}{n} \right )^{-1/2}  { \hat \Omega_{1,12} }  \Omega_{22}^{-1/2} \\ 
        \Omega_{22}^{-1/2} \hat \Omega_{12}  \left ( \frac{\Omega_{1,11}}{n} \right )^{-1/2}    
        & \Omega_{22}^{-1/2} \hat \Omega_{22} \Omega_{22}^{-1/2}  \end{pmatrix} \\
         & = \begin{pmatrix}
        \left ( \frac{\Omega_{1,11}}{n} \right )^{-1/2}  \left ( \frac{ \Omega_{1,11}}{n} + o_p(1/n) \right)   \left ( \frac{\Omega_{1,11}}{n} \right )^{-1/2} 
        & \left ( \frac{\Omega_{1,11}}{n} \right )^{-1/2}  \left ( \frac{ \Omega_{1,21}}{n}+ o_p \left (  \frac{\sqrt{ \parallel\Omega_{22}\parallel  }}{\sqrt{n}}\right ) \right)  \Omega_{22}^{-1/2} 
        \\ 
       \Omega_{22}^{-1/2} \left ( \frac{ \Omega_{1,21}}{n}+ o_p \left (  \frac{\sqrt{ \parallel\Omega_{22}\parallel  }}{\sqrt{n}}\right ) \right)   \left ( \frac{\Omega_{1,11}}{n} \right )^{-1/2} 
        & \Omega_{22}^{-1/2} \left (  \Omega_{22}   + o_p \left (   \parallel\Omega_{22}\parallel  \right ) \right)\Omega_{22}^{-1/2} 
    \end{pmatrix} \\
      & = \begin{pmatrix}
        1 + o_p(1)
        & \frac{\Omega_{1,21}}{\sqrt{n} \sqrt{\Omega_{1,11} \Omega_{22}}} + o_p(1) 
        \\ 
      \frac{\Omega_{1,21}}{\sqrt{n} \sqrt{\Omega_{1,11} \Omega_{22}}} + o_p(1)
        & 1 + o_p(1)
    \end{pmatrix}
\end{align*}
\end{small}

\end{commentP}
Define $\hat v_{mn} = \sqrt{m} \hat Q d_{\Omega,mn}^{-1/2}  \bar g_{mn}(\delta)  $ and note that 
$$\hat v_{mn} \xrightarrow{ d} N(0, Q_{mn}  d_{\Omega,mn}^{-1/2} \Omega_{mn}  d_{\Omega,mn}^{-1/2}  Q_{mn}') = N(0, Q_{mn} (Q_{mn}'Q_{mn})^{-1} Q_{mn}') = N(0, I_L). $$
Now we can come back to our test statistic:
\begin{align*}
     J(\hat \delta) &=  m \bar g_{mn}( \hat \delta)' \hat \Omega^{-1} \bar g_{mn}( \hat \delta)
    \\ &=   \left (\sqrt{m} d_{\Omega,mn}^{-1/2} \bar g_{mn}( \hat \delta) \right )'  \left [ d_{\Omega,mn}^{-1/2} \hat \Omega d_{\Omega,mn}^{-1/2} \right ]^{-1}  \left ( \sqrt{m} d_{\Omega,mn}^{-1/2} \bar g_{mn}( \hat \delta)  \right)     
    \\ &=   \left (\sqrt{m} d_{\Omega,mn}^{-1/2} \hat B \bar g_{mn}(  \delta) \right )'  \left [ d_{\Omega,mn}^{-1/2} \hat \Omega d_{\Omega,mn}^{-1/2} \right ]^{-1}  \left (\sqrt{m} d_{\Omega,mn}^{-1/2}  \hat B \bar g_{mn}(  \delta)  \right) 
            \\ &=   \left (\sqrt{m}   \bar g_{mn}(  \delta) \right )'  \hat B'  \hat \Omega ^{-1} \hat B \left (\sqrt{m}   \bar g_{mn}(  \delta)  \right) 
        \\ &=   \left (\sqrt{m}   \bar g_{mn}(  \delta) \right )' d_{\Omega,mn}^{-1/2} \hat Q ' \hat M \hat Q d_{\Omega,mn}^{-1/2} \left (\sqrt{m}   \bar g_{mn}(  \delta)  \right) 
\\ &=   \left (\sqrt{m}  \hat Q d_{\Omega,mn}^{-1/2}  \bar g_{mn}(  \delta) \right )'  \hat M  \left (\sqrt{m} \hat Q d_{\Omega,mn}^{-1/2}  \bar g_{mn}(  \delta)  \right)\\
& = \hat v_{mn}' \hat M \hat v_{mn}.
\end{align*}

Since $\hat M$ is idempotent with rank $L-K$, it immediately follows that
$$  J(\hat \delta) \xrightarrow{d} \chi^2_{L-K}. $$

\end{proof} 

\section{Least Squares Panel Data Models}\label{app:linear models}
\subsection{Formal results}\label{sec:formal ls}
This section complements subsection \ref{subsec:least squares estimators} by discussing more in detail the relationship between the least squares estimator and the minimum distance approach. Throughout the section, we define the $n \times  (K_1 + 1 ) $ matrix of first-stage regressors $\tilde X_{1j} = (  \tilde x_{1j} ,  \tilde x_{2j}, \dots , \tilde x_{nj}  )'$, and the $mn \times K_1$ matrix of individual-level regressors $X_1 = (X_{1j}', \dots, X_{1m}')'$. Further, we use the matrices $P_j = l(l'l)^{-1} l'$ and $Q_j = I_j- P_j$, where $l$ is a $n \times  1$ vector of ones. Thus, $P_j X_j = \bar X_j $ and $Q_j X_{1j} =\dot X_{1j}  $. We consider a linear version of our estimator, where OLS instead of quantile regression is used in the first stage and we focus on model (\ref{FELS2}).
In this section, we show that mean models can be estimated using a two-step procedure. Notation is the same as in the paper, except that the fitted values are computed using an OLS regression. More precisely, the vector of fitted values of group $j$ is
\begin{equation*}
    \hat Y_j = \tilde X_j \hat \beta_j = \tilde X_j \left ( \tilde X_j' \tilde X_j \right )^{-1} \tilde X_j' Y_j.
\end{equation*}
The following Proposition states the equivalence of the two-step procedure using the fitted values and the conventional one-step estimator in mean models.
\begin{proposition}
\label{prop:md = gmm} Denote $\hat \delta_{GMM}^{MD}$ the coefficient vector of a linear GMM regression of $\hat Y$ on $X$ with instrument $Z$. Let $\hat \delta_{GMM}$ be the coefficient vector of the same GMM regression but with regressand $Y$. Assume that for each $j$, $Z_j $ lies in the column space of $ \tilde X_j$, then $\hat \delta^{MD}_{GMM} = \hat \delta_{GMM}$.
\end{proposition}
The proof of this Proposition and all subsequent proofs are in Appendix \ref{app:proofs linear models}. 
Proposition \ref{prop:md = gmm} implies that any linear model can be computed by a two-step estimator as long as the matrix of instruments of each group $j$, $Z_j$ lies in the column space of the matrix of first-stage regressors of group $j$, $\tilde X_j$.\footnote{Since $\tilde X_j$ includes a constant, the presence of group-level variables in $Z_j$ will not affect its column space.} This result applies to a wide range of estimators. Since OLS is a special case of GMM, the result for pooled OLS follows directly, while the result for the within estimator is summarized in the following Corollary.
\begin{corollary}\label{cor:FE}
Denote $\hat \delta_{FE}^{MD}$ the coefficient vector of an IV regression of $\hat Y$ on $X_1$ with instruments $\dot X_1$. Let $\hat \delta_{FE}$ be the coefficient vector of the within estimator, that is, of a regression of $\dot Y$ on $\dot X_1$. Then $\hat \delta_{FE}^{MD} = \hat \delta_{FE}$.
\end{corollary}
The between estimator is usually computed by regressing $\bar Y$ on $\bar X$. Alternatively, it can be estimated by an IV regression of $Y$ (or $\hat Y$) on $X$ using $\bar X$ as an instrument, where it exploits only the variation between individuals. 
\begin{corollary}\label{cor:BE}
Denote $\hat \delta_{BE}^{MD}$ the coefficient vector of an IV regression of $\hat Y$ on $X$ with instruments $\bar X$. Let $\hat \delta_{BE}$ be the coefficient vector of the between estimator, that is, of a regression of  $\bar Y$ on $\bar X$. Then $\hat \delta_{BE}^{MD} = \hat \delta_{BE}$.
\end{corollary}
It is worth noting that the IV approach to these panel data estimators also works in one stage with $Y$ as the dependent variable. Further, it is clearly possible to estimate between (within) models using average (demeaned) fitted values and regressors. 

The pooled OLS and the between estimators can estimate both $\beta $ and $\gamma$ but are not efficient. 
The random effects estimator optimally combines between and the within variation to find a more efficient estimator. While FGLS is the most common estimator for the random effects model, \cite{Im1999} show that the overidentified 3SLS estimator, with instruments $Z_j = (\dot X_{1j},\bar X_j)$, is identical to the random effects estimator. 
The 3SLS estimator is a special case of GMM with weighting matrix $ W = \mathbb \bE[ Z_j'\tilde \Omega Z_j]$ where $\tilde \Omega$ follows the usual random effects covariance structure. 
Thus, by Proposition \ref{prop:md = gmm}, the random effects estimator can also be computed in two steps using the fitted values in the second stage.
\begin{corollary}\label{cor:RE}
Denote $\hat \delta_{RE}^{MD}$ the coefficient vector of a 3SLS regression of $\hat Y$ on $X$ with instruments $(\dot X_{1j},\bar X_j)$. Let $\hat \delta_{RE}$ be the coefficient vector of a random effects regression of $Y$ on $X$. Then $\hat \delta_{RE}^{MD} = \hat \delta_{RE}$.
\end{corollary}
Alternatively, the random effects estimator can be implemented using the theory of optimal instruments and a just identified 2SLS regression. Starting from a conditional moment restriction, the idea of optimal instruments is to select an instrument and weights that minimize the asymptotic variance (see, e.g. \citealp{Newey1993}). Relevant to our two-step procedure, under homoskedasticity of the errors, the conditional moments $\bE[Y_j - X_j \delta |X_j] = 0$ and $\bE[ \hat Y_j - X_j \delta |X_j] = 0$ imply the same optimal instrument and the problem simplifies to the usual random effects estimator. 

\begin{proposition}\label{prop: optimal instruments}
Assume $\bE[\varepsilon_{ij}^2 |X_j] = \sigma^2_\varepsilon$ and $\bE[\alpha_j^2 |X_j] = \sigma^2_\alpha$. The conditional moments $\bE[\hat Y_j - X_j \delta |X_j] = 0$ and $\bE[  Y_j - X_j \delta |X_j] = 0$ imply the same optimal instrument. 
\end{proposition}

The Hausman-Taylor model \citep{Hausman1981} is a middle ground between the fixed effects and the random effects models where some regressors are assumed to be uncorrelated with $\alpha_j$. The matrix of regressors $X$ is partitioned as $X = [X_1^{ex} \ X_1^{en} \ X_2^{ex} \ X_2^{en} ]$ where $X_1^{ex}$ and $X_2^{ex}$ are orthogonal to $\alpha_j$. No assumption is placed on the relationship between $\alpha_j$ and $X_1^{en}$ and $X_2^{en}$. The model can be estimated by IV using instruments $Z = (\dot X_1^{ex}, \dot X_1^{en},$ $ \bar X_1^{ex}, X_2^{ex})$ (see, e.g., \citealp{Hansen2021}). Thus, it follows by Proposition \ref{prop:md = gmm} that the Hausman-Taylor model can be estimated in two stages. 
\begin{corollary}\label{cor:HT}
Denote $\hat \delta_{HT}^{MD}$ the coefficient vector of a 2SLS regression of $\hat Y$ on $X$ with instruments $(\dot X_1^{ex}, \dot X_1^{en},$ $ \bar X_1^{ex}, X_2^{ex})$. Let $\hat \delta_{HT}$ be the coefficient vector of the Hausman-Taylor Estimator based on a regression $Y$ on $X$. Then $\hat \delta_{HT}^{MD} = \hat \delta_{HT}$.
\end{corollary}

Finally, we show that not only the point estimates but also the standard errors can be obtained using the two-stage minimum distance approach. This requires clustering the standard errors in the second stage at a level weakly higher than the group $j$. Let $g = 1, \dots , G$ index the clusters and assume that each of the clusters has $N_g$ observations.
This nests the case where one wishes to cluster at the individual level or at a higher level. For example, if $j$ are county-year combinations, one might cluster at the county level. For an estimator $\hat \delta$, the clustered covariance matrix is estimated by
\begin{align}\label{eq:cov_estimator}
    \hat V_\delta = & \left ( \sumgG X_g' Z_g \hat W \sumgG Z_g' X_g  \right ) ^{-1} \sumgG X_g' Z_g \hat W \left ( \sumgG Z_g' \tilde u_g \tilde u_g' Z_g\right ) \nonumber \\ & \cdot \hat W \sumgG Z_g' X_g  \left ( \sumgG X_g' Z_g \hat W \sumgG Z_g' X_g  \right ) ^{-1},
\end{align}
where $\tilde u_g$ is a $N_g$-dimensional vector of estimated errors for the observations in cluster $g$, $X_g$ is the $K \times N_g$ matrix of regressors of cluster $g$, and $Z_g$ is the $L \times N_g$ matrix of instruments.

\begin{proposition}\label{prop:clustered se OLS}
Denote $\hat V_\delta$ the clustered covariance matrix of $\hat \delta$ estimated by a GMM regression of $Y$ on $X$ with instrument $Z$. Let $\hat V_{ \delta^{MD}}$ be the clustered covariance matrix of $\hat \delta^{MD}$ estimated by GMM regression of $\hat Y$ on $X$ with instrument $Z$, where $\hat Y$ are estimated by an OLS first-stage. Let the clusters be at weakly higher level than $j$. Then, $\hat V_{ \delta^{MD}} = \hat V_{ \delta}$.
\end{proposition}

\subsection{Proofs of the least squares results}\label{app:proofs linear models}

\begin{proof}[Proof of Proposition \ref{prop:md = gmm}]
Define the  projection matrix $\tilde P_j = \tilde X_j ( \tilde X_j' \tilde X_j)^{-1} \tilde X_j'$. Since $Z_j$ is in the column space of $\tilde X_j$, 
\begin{equation}\label{eq:projection}
\tilde P_j Z_j = Z_j
\end{equation}
The MD estimator with a GMM second stage is: 
\begin{align*}
    \hat \delta^{MD}_{GMM} &= \left (  X'  Z  \hat W  Z'  X \right )^{-1 }  X'  Z \hat  W  Z'  \hat { Y }.
    \end{align*}
For $\hat \delta^{MD}_{GMM}$ to be equal to $\hat \delta_{GMM}$, it suffices that $  {Z' \hat Y} = {Z'Y}$. Note that
    \begin{align*}
             Z'  \hat { Y }  & = \sum_{j = 1}^m Z_j' \hat Y_j \\
       & =   \sum_{j = 1}^m Z_j' \tilde X_j \hat \beta_j \\ 
    & =  \sum_{j = 1}^m Z_j' \tilde X_j  (\tilde X_j' \tilde X_j)^{-1} \tilde X_j' Y_j \\ 
   & =  \sum_{j = 1}^m (\tilde P_j Z_j)'  Y_j \\ 
         & =   \sum_{j = 1}^m Z_j'  Y_j =  {Z'Y},\\ 
\end{align*}
where the third line uses $\hat Y_j = \tilde X_j \hat \beta_j$, the fourth line uses the definition of the first-step OLS estimator, and the last line uses equation (\ref{eq:projection}). It then follows directly that $\hat \delta_{MD}$ equals $\hat \delta_{GMM}$.
\end{proof}

\begin{proof}[Proof of Corollary \ref{cor:FE}]
First, note that since $Q_j  X_{1j} = \dot X_{1j}$, $\dot X_{1j} $ lies in the column space of $ X_{1j}$. Then, we apply Proposition \ref{prop:md = gmm}. It follows that a IV regression of $\hat Y$ on $X_{1j}$ with instrument $Z_j$ is algebraically identical to a IV regression with $Y_j$ as dependent variable.
Then,
\begin{align*}
    \hat \delta_{FE}^{MD} =&  \left ( \sum_{j = 1}^m Z_j ' {  X_{1j}} \right)^{-1}  \sum_{j = 1}^m Z_j' Y_j \\
     =&  \left ( \sum_{j = 1}^m \dot X_{1j} ' {  X_{1j}} \right)^{-1}  \sum_{j = 1}^m \dot X_{1j}' Y_j \\
          =&  \left ( \sum_{j = 1}^m {  X_{1j}'} Q_j {  X_{1j}} \right)^{-1}  \sum_{j = 1}^m  X_{1j}' Q_j Y_j \\
     =&  \left ( \sum_{j = 1}^m \dot {  X_{1j}'}  \dot {  X_{1j}} \right)^{-1}  \sum_{j = 1}^m \dot {  X_{1j}'} \dot Y_j = \hat \delta_{FE},\\
\end{align*}
where the second line follows since $Z_j = \dot X_{1j}$, the third and last line by $Q_j X_{1j} = \dot X_{1j}$, $Q_j Y_j = \dot Y_j $ and since $Q_j$ is idempotent.
\end{proof}

\begin{proof}[Proof of Corollary \ref{cor:BE}]
First, note that since $P_j \tilde X_j = \bar X_j$, $\bar X_j $ lies in the column space of $\tilde X_j$. Then, we apply Proposition \ref{prop:md = gmm}.
It follows that an IV regression of $\hat Y_j$ on $X_j$ with instrument $Z_j$ is algebraically identical to an IV regression with $Y_j$ as dependent variable. 
Then,
\begin{align*}
    \hat \delta_{BE}^{MD} =&  \left ( \sum_{j = 1}^m Z_j ' X_j \right)^{-1}  \sum_{j = 1}^m Z_j' Y_j \\
     =&  \left ( \sum_{j = 1}^m \bar X_j ' X_j \right)^{-1}  \sum_{j = 1}^m \bar X_j' Y_j \\
          =&  \left ( \sum_{j = 1}^m X_j' P_j  X_j \right)^{-1}  \sum_{j = 1}^m  X_j' P_j Y_j \\
     =&  \left ( \sum_{j = 1}^m \bar X_j ' \bar X_j \right)^{-1}  \sum_{j = 1}^m \bar X_j' \bar Y_j = \hat \delta_{BE}\\
\end{align*}
where the second line follows since $Z_j = \bar X_j$, the third and last line by $P_j X_j = \bar X_j$, $P_j Y_j = \bar Y_j $ and, since $P_j$ is idempotent.
\end{proof}

\begin{proof}[Proof of Proposition \ref{prop: optimal instruments}]
The optimal instrument takes the form $Z^*_j = \bE [g_j(\delta)g_j(\delta)' | Z_j] ^{-1} R_j(\delta, \tau)$, where $R_j(\delta, \tau) = \bE [ \frac{\partial}{\partial \delta} g_j(\delta, \tau) | Z_j]$.
For both moment conditions, $R_j(\delta, \tau) $ is identical. 
Then for the first moment restriction, we have:
\begin{align}\label{eq:optimal instrument}
 \bE [( \hat Y_j - X_j \delta)( \hat Y_j - X_j \delta)' | X_j]  &= \bE [( \tilde X_j (\hat \beta_j  -  \beta)   + \tilde X_j\beta  - X_j \delta) ( \tilde X_j (\hat \beta_j  -  \beta)   + \tilde X_j\beta  - X_j \delta)' | X_j] \\ \nonumber &= \bE [( \tilde X_j (\hat \beta_j  -  \beta)   + \alpha_j)( \tilde X_j (\hat \beta_j  -  \beta)   + \alpha_j)' | X_j] \\ \nonumber  &= \tilde X_j \frac{V_j}{n} \tilde X_j' + \mathbf l_n \mathbf l_n'\sigma^2_\alpha.
\end{align}
The matrix $\tilde X_j  \frac{V_j}{n} \tilde X_j' + \mathbf l_n \mathbf l_n'\sigma^2_\alpha$ is singular, so that we suggest using the Moore-Penrose inverse to construct the optimal instrument. \\
For the second moment restriction, we have:
\begin{align*}
     \bE [( Y_j - X_j \delta)(Y_j - X_j \delta)'| X_j] &= \bE [ (\alpha_j + \eps_{ij})(\alpha_j + \eps_{ij})' | X_j] \\ & = \left( \mathbf I_n \sigma_\eps^2 + \mathbf l_n \mathbf l_n' \sigma_\alpha^2 \right).
\end{align*}
Then note that $\left (\mathbf I_n \sigma_\eps^2 + \mathbf l_n \mathbf l_n' \sigma_\alpha^2 \right )^{-1} = ( \tilde X_j \tilde X_j^+  \sigma_\eps^2 + \mathbf l_n'\mathbf l_n \sigma_\alpha^2 )^{+} =
( \tilde X_j  (\tilde X_j' \tilde X_j)^{-1} \tilde X_j' \sigma_\eps^2 + 
\mathbf l_n'\mathbf l_n \sigma_\alpha^2 )^{+} =  ( \tilde X_j  \frac{V_j}{n} \tilde X_j' + \mathbf l_n'\mathbf l_n \sigma_\alpha^2 )^+  X_j$, where $V_j = ( \frac{1}{n} \tilde X_j '\tilde X_j )^{-1} \sigma_\eps^2$ and since for a full column rank matrix $\tilde X_j$, $\tilde X_j \tilde X_j^+ = I_n $ and $\tilde X_j^+ = (\tilde X_j' \tilde X_j )^{-1} \tilde X_j'$.
\end{proof}

\begin{proof}[Proof of Proposition \ref{prop:clustered se OLS}]
Define $Z_g = (z_{1g}, \dots , z_{n_g g})'$, $X_g = (x_{1g}, \dots , x_{n_g g})'$, $Y_g = (y_{1g}, \dots, y_{n_g g})'$ and  $\hat Y_g = (\hat y_{1g}, \dots, \hat y_{n_g g})'$. The first and third terms of expression (\ref{eq:cov_estimator}) are identical for both estimators. Hence, we focus on the middle term. Let $\hat u_g = Y_g  - X_g \hat \delta$ be the vector of residuals from the regression using $Y$ as dependent variable, and let $\hat u_g^{MD} = \hat Y_g - X_g \hat \delta^{MD}$ be the vector of residuals of the estimator using the fitted values as regressand. We show that $Z_g' \hat u_g = Z_g' \hat u_g^{MD}$ for all $g$. 
By Proposition \ref{prop:md = gmm}, $\hat \delta^{MD} = \hat \delta$. Thus, the fitted values of both estimators are identical. 
Next, define  $X_g = \operatorname{diag} ( X_j : j \in g)$ and recall that regressing $Y_g$ on $\breve X_g$ is the same as performing separate regression for each $j \in g$. Let $\breve{\beta}_g $ be the coefficient vector of an OLS regression of $Y_g$ on $\breve X_g$.
Note that $ Z_g$ is in the column space of $\tilde X_g$. Define the projection matrix $\breve P = \breve X_g ( \breve X_g' \breve X_g)^{-1} \breve X_g'$. Since $ Z_g$ is in the column space of $\breve X_g$, 

\begin{equation}\label{eq:col space}
\breve P Z_g = Z_g.
\end{equation}

Then,
\begin{align*}
Z_g' \hat u_g^{MD} & = Z_g' \left ( \hat Y_g - X_g \hat \delta^{MD} \right ) \\
&= Z_g ' \tilde X_g \breve \beta_g -   Z_g  X_g \hat \delta\\
&= Z_g'  \tilde X_g (\breve X_g' \breve X_g)^{-1}  \breve X_g' Y_g -   Z_g  X_g \hat \delta \\
& = Z_g' ( Y_g -X_g \hat \delta)  = Z_g' \hat u_g,
\end{align*}
where the fourth line follows by (\ref{eq:col space}).
Since this holds for all $g$, the desired result follows directly. 
\end{proof}

\section{Optimal Instruments and Minimum Distance}\label{app:opt inst MD}
In this section, we show that if $\alpha_j(\tau) = 0$ for all $j$ and $\tau$, efficient minimum distance can be implemented by optimal instruments. 
From equation (\ref{eq:optimal instrument}) we have that if $\alpha_j(\tau) = 0 $ for all $j$ and all $\tau$, $\bE[(\tilde X_j \hat \beta_j(\tau) - X_j \delta(\tau)) (\tilde X_j \hat \beta_j(\tau) - X_j \delta(\tau))'| X_j] = \tilde X_j \frac{V_j(\tau)}{n} \tilde X_j'$. This implies the optimal instrument $Z_j^* = (\tilde X_j \frac{V_j(\tau)}{n} \tilde X_j')^+ X_j $. Since $n$ is a scalar, using  $Z_j^*(\tau) = (\tilde X_j V_j(\tau) \tilde X_j')^+ X_j$ leads to the same results.

\begin{proposition} \label{prop:emd = giv}
The IV regression with instrument $Z_j^*(\tau) =  (\tilde X_j V_j(\tau) \tilde X_j')^+X_j$ is the efficient MD estimator.
\end{proposition}

\begin{proof}
\begin{align*}
    \hat \delta_{EMD}(\tau) =& \left ( \sum_{j = 1}^m R_j' \hat V_j^{-1}(\tau) R_j \right )^{-1} \left (  \sum_{j = 1}^m R_j' \hat V_j^{-1}(\tau) \hat \beta_j(\tau) \right )\\
       =& \left ( \sum_{j = 1}^m  X_j ' \tilde X_j \left(\tilde X_j ' \tilde X_j  \hat V_j(\tau) \tilde X_j ' \tilde X_j \right)^{-1} \tilde X_j '  X_j   \right )^{-1} \left (   X_j ' \tilde X_j \left(\tilde X_j ' \tilde X_j  \hat V_j(\tau)\tilde X_j ' \tilde X_j \right)^{-1} \tilde X_j ' \hat Y_j(\tau) \right )\\
       =& \left ( \sum_{j = 1}^m  X_j '  \left( \tilde X_j  \hat V_j(\tau) \tilde X_j '  \right)^{+}   X_j   \right )^{-1} \left (    X_j' \left(  \tilde X_j  \hat V_j(\tau)\tilde X_j ' \right)^{+}  \hat Y_j (\tau)\right ) = \hat \delta_{OI}(\tau).
\end{align*}
The second line follows as $\tilde X_j R_j = X_j$ and the third line follows since for a full column rank matrix $\tilde X_j$, $\tilde X_j^+ = (\tilde X_j' \tilde X_j)^{-1} \tilde X_j'$.
\end{proof}

\end{document}

%% file: Tables/clp_simul.tex
\begin{table}[ht]
\centering\centering
\caption{\label{tab:tab:grouped_bias}Bias, Standard Deviation and MSE of $\hat \gamma(\tau)$}
\centering
\fontsize{9}{11}\selectfont
\begin{threeparttable}
\begin{tabular}[t]{llllllllll}
\toprule
\multicolumn{1}{c}{ } & \multicolumn{3}{c}{No group-level heterogeneity} & \multicolumn{3}{c}{Exogenous} & \multicolumn{3}{c}{Endogenous} \\
\cmidrule(l{3pt}r{3pt}){2-4} \cmidrule(l{3pt}r{3pt}){5-7} \cmidrule(l{3pt}r{3pt}){8-10}
Quantile & MD & CLP & Rel. MSE & MD & CLP & Rel. MSE & MD & CLP & Rel. MSE\\
\midrule
\addlinespace[.5em]
\hline
\multicolumn{10}{c}{(m, n) = (25, 25)}\\
\hspace{1em}0.1 & 0.022 & -0.011 & 0.051 & 0.022 & -0.010 & 0.052 & 0.036 & 0.001 & 0.160\\
\hspace{1em} & (0.192) & (0.858) &  & (0.195) & (0.860) &  & (2.025) & (5.062) & \\
\hspace{1em}0.5 & -0.010 & -0.001 & 0.061 & -0.011 & 0.000 & 0.088 & -0.029 & 0.039 & 0.132\\
\hspace{1em} & (0.166) & (0.673) &  & (0.204) & (0.691) &  & (1.992) & (5.491) & \\
\hspace{1em}0.9 & -0.019 & -0.003 & 0.049 & -0.020 & -0.004 & 0.216 & -0.063 & -0.011 & 0.176\\
\hspace{1em} & (0.094) & (0.435) &  & (0.227) & (0.490) &  & (2.123) & (5.065) & \\
\addlinespace[.5em]
\hline
\multicolumn{10}{c}{(m, n) = (200, 25)}\\
\hspace{1em}0.1 & 0.024 & 0.003 & 0.060 & 0.024 & 0.004 & 0.063 & 0.023 & 0.006 & 0.057\\
\hspace{1em} & (0.066) & (0.284) &  & (0.067) & (0.285) &  & (0.106) & (0.456) & \\
\hspace{1em}0.5 & -0.006 & -0.001 & 0.059 & -0.006 & 0.000 & 0.086 & -0.009 & -0.003 & 0.071\\
\hspace{1em} & (0.056) & (0.232) &  & (0.069) & (0.238) &  & (0.097) & (0.366) & \\
\hspace{1em}0.9 & -0.017 & -0.004 & 0.060 & -0.017 & -0.003 & 0.223 & -0.022 & -0.009 & 0.142\\
\hspace{1em} & (0.031) & (0.145) &  & (0.075) & (0.164) &  & (0.086) & (0.234) & \\
\addlinespace[.5em]
\hline
\multicolumn{10}{c}{(m, n) = (25, 200)}\\
\hspace{1em}0.1 & 0.003 & -0.002 & 0.059 & 0.003 & -0.001 & 0.066 & -0.028 & -0.076 & 0.141\\
\hspace{1em} & (0.070) & (0.289) &  & (0.074) & (0.291) &  & (2.107) & (5.618) & \\
\hspace{1em}0.5 & -0.001 & -0.002 & 0.060 & -0.001 & -0.001 & 0.233 & -0.083 & -0.094 & 0.628\\
\hspace{1em} & (0.060) & (0.247) &  & (0.134) & (0.278) &  & (3.627) & (4.575) & \\
\hspace{1em}0.9 & -0.002 & 0.000 & 0.061 & -0.001 & 0.001 & 0.769 & -0.120 & -0.114 & 1.218\\
\hspace{1em} & (0.030) & (0.121) &  & (0.217) & (0.247) &  & (3.930) & (3.561) & \\
\addlinespace[.5em]
\hline
\multicolumn{10}{c}{(m, n) = (200, 200)}\\
\hspace{1em}0.1 & 0.003 & -0.003 & 0.057 & 0.003 & -0.003 & 0.062 & 0.002 & -0.004 & 0.058\\
\hspace{1em} & (0.024) & (0.100) &  & (0.025) & (0.101) &  & (0.039) & (0.162) & \\
\hspace{1em}0.5 & -0.001 & 0.000 & 0.059 & -0.001 & -0.001 & 0.222 & -0.004 & -0.004 & 0.141\\
\hspace{1em} & (0.020) & (0.084) &  & (0.044) & (0.093) &  & (0.051) & (0.136) & \\
\hspace{1em}0.9 & -0.002 & 0.000 & 0.067 & -0.003 & -0.001 & 0.762 & -0.009 & -0.007 & 0.617\\
\hspace{1em} & (0.010) & (0.040) &  & (0.071) & (0.082) &  & (0.074) & (0.095) & \\
\bottomrule
\end{tabular}
\begin{tablenotes}
\item \textit{Note: } 
\item The table reports mean bias, standard deviation and relative MSE from the simulations for $\gamma(\tau)$ from 10,000 Monte Carlo simulations using the MD estimator and the CLP estimator. The relative MSE gives the MSE of the MD estimator relative to that of the CLP estimator.
\end{tablenotes}
\end{threeparttable}
\end{table}

%% file: Tables/cilength_simul.tex
\begin{table}[ht]
\centering\centering
\caption{\label{tab:tab:cilength}Coverage Probability of the $95\%$ Confidence Intervals}
\centering
\fontsize{9}{11}\selectfont
\begin{threeparttable}
\begin{tabular}[t]{llllllllll}
\toprule
\multicolumn{1}{c}{ } & \multicolumn{3}{c}{No group-level heterogeneity} & \multicolumn{3}{c}{Exogenous} & \multicolumn{3}{c}{Endogenous} \\
\cmidrule(l{3pt}r{3pt}){2-4} \cmidrule(l{3pt}r{3pt}){5-7} \cmidrule(l{3pt}r{3pt}){8-10}
\multicolumn{1}{c}{ } & \multicolumn{1}{c}{Rel. length} & \multicolumn{2}{c}{Coverage Rate} & \multicolumn{1}{c}{Rel. length} & \multicolumn{2}{c}{Coverage Rate} & \multicolumn{1}{c}{Rel. length} & \multicolumn{2}{c}{Coverage Rate} \\
\cmidrule(l{3pt}r{3pt}){2-2} \cmidrule(l{3pt}r{3pt}){3-4} \cmidrule(l{3pt}r{3pt}){5-5} \cmidrule(l{3pt}r{3pt}){6-7} \cmidrule(l{3pt}r{3pt}){8-8} \cmidrule(l{3pt}r{3pt}){9-10}
Quantile & MD/CLP & MD & CLP & MD/CLP & MD & CLP & MD/CLP & MD & CLP\\
\midrule
\addlinespace[.5em]
\hline
\multicolumn{10}{c}{(m, n) = (25, 25)}\\
\hspace{1em}0.1 & 0.232 & 0.941 & 0.938 & 0.235 & 0.939 & 0.938 & 0.227 & 0.966 & 0.972\\
\hspace{1em}0.5 & 0.244 & 0.940 & 0.942 & 0.301 & 0.943 & 0.945 & 0.262 & 0.963 & 0.972\\
\hspace{1em}0.9 & 0.223 & 0.941 & 0.949 & 0.501 & 0.940 & 0.946 & 0.373 & 0.956 & 0.972\\
\addlinespace[.5em]
\hline
\multicolumn{10}{c}{(m, n) = (200, 25)}\\
\hspace{1em}0.1 & 0.230 & 0.932 & 0.947 & 0.233 & 0.932 & 0.948 & 0.231 & 0.942 & 0.953\\
\hspace{1em}0.5 & 0.245 & 0.946 & 0.944 & 0.296 & 0.945 & 0.946 & 0.267 & 0.952 & 0.949\\
\hspace{1em}0.9 & 0.220 & 0.926 & 0.947 & 0.475 & 0.941 & 0.945 & 0.368 & 0.953 & 0.952\\
\addlinespace[.5em]
\hline
\multicolumn{10}{c}{(m, n) = (25, 200)}\\
\hspace{1em}0.1 & 0.241 & 0.943 & 0.940 & 0.256 & 0.943 & 0.943 & 0.240 & 0.968 & 0.974\\
\hspace{1em}0.5 & 0.242 & 0.937 & 0.944 & 0.495 & 0.938 & 0.944 & 0.370 & 0.949 & 0.971\\
\hspace{1em}0.9 & 0.248 & 0.948 & 0.941 & 0.884 & 0.933 & 0.945 & 0.771 & 0.939 & 0.955\\
\addlinespace[.5em]
\hline
\multicolumn{10}{c}{(m, n) = (200, 200)}\\
\hspace{1em}0.1 & 0.241 & 0.944 & 0.944 & 0.254 & 0.947 & 0.945 & 0.246 & 0.951 & 0.950\\
\hspace{1em}0.5 & 0.244 & 0.946 & 0.945 & 0.483 & 0.952 & 0.948 & 0.377 & 0.957 & 0.951\\
\hspace{1em}0.9 & 0.246 & 0.942 & 0.953 & 0.872 & 0.950 & 0.950 & 0.772 & 0.954 & 0.955\\
\bottomrule
\end{tabular}
\begin{tablenotes}
\item \textit{Note: } 
\item Results based on 10,000 Monte Carlo simulations. The table provides the coverage rate and median length of the confidence intervals of $\gamma(\tau)$. The relative length provides the length of the confidence interval of the MD estimator relative to that of the CLP estimator. Robust standard errors are used for the CLP estimator, and clustered standard errors at the group level are used for the MD estimator.
\end{tablenotes}
\end{threeparttable}
\end{table}

%% file: Tables/panel_simul.tex
\begin{table}
\centering
\caption{\label{tab:panel_bias}Bias and Standard Deviation of $\hat \beta(\tau)$}
\centering
\fontsize{9}{11}\selectfont
\begin{threeparttable}
\begin{tabular}[t]{llllll}
\toprule
Quantile & Pooled & BE & FE & RE-GMM & RE-OI\\
\midrule
\addlinespace[.5em]
\hline
\multicolumn{6}{c}{(m, n) = (25, 10)}\\
\hspace{1em}0.1 & 0.009 & 0.002 & 0.037 & 0.019 & 0.044\\
\hspace{1em} & (0.193) & (0.235) & (0.261) & (0.183) & (0.177)\\
\hspace{1em}0.5 & 0.000 & 0.000 & -0.001 & -0.001 & 0.000\\
\hspace{1em} & (0.182) & (0.224) & (0.172) & (0.142) & (0.168)\\
\hspace{1em}0.9 & -0.010 & -0.003 & -0.039 & -0.021 & -0.045\\
\hspace{1em} & (0.195) & (0.235) & (0.259) & (0.184) & (0.181)\\
\addlinespace[.5em]
\hline
\multicolumn{6}{c}{(m, n) = (200, 10)}\\
\hspace{1em}0.1 & 0.011 & 0.005 & 0.040 & 0.021 & 0.046\\
\hspace{1em} & (0.068) & (0.080) & (0.092) & (0.061) & (0.067)\\
\hspace{1em}0.5 & 0.001 & 0.001 & 0.001 & 0.001 & 0.001\\
\hspace{1em} & (0.063) & (0.076) & (0.059) & (0.047) & (0.063)\\
\hspace{1em}0.9 & -0.010 & -0.003 & -0.040 & -0.019 & -0.045\\
\hspace{1em} & (0.067) & (0.080) & (0.091) & (0.060) & (0.068)\\
\addlinespace[.5em]
\hline
\multicolumn{6}{c}{(m, n) = (25, 25)}\\
\hspace{1em}0.1 & 0.003 & 0.000 & 0.015 & 0.011 & 0.016\\
\hspace{1em} & (0.175) & (0.222) & (0.141) & (0.122) & (0.120)\\
\hspace{1em}0.5 & -0.003 & -0.004 & 0.000 & -0.001 & -0.002\\
\hspace{1em} & (0.171) & (0.218) & (0.102) & (0.094) & (0.106)\\
\hspace{1em}0.9 & -0.009 & -0.007 & -0.017 & -0.014 & -0.018\\
\hspace{1em} & (0.177) & (0.223) & (0.138) & (0.121) & (0.120)\\
\addlinespace[.5em]
\hline
\multicolumn{6}{c}{(m, n) = (200, 25)}\\
\hspace{1em}0.1 & 0.006 & 0.004 & 0.015 & 0.012 & 0.017\\
\hspace{1em} & (0.061) & (0.075) & (0.049) & (0.041) & \vphantom{1} (0.042)\\
\hspace{1em}0.5 & 0.000 & 0.000 & 0.000 & 0.000 & \vphantom{1} 0.000\\
\hspace{1em} & (0.059) & (0.073) & (0.036) & (0.032) & (0.036)\\
\hspace{1em}0.9 & -0.006 & -0.004 & -0.015 & -0.012 & -0.017\\
\hspace{1em} & (0.061) & (0.075) & (0.049) & (0.041) & (0.042)\\
\addlinespace[.5em]
\hline
\multicolumn{6}{c}{(m, n) = (25, 200)}\\
\hspace{1em}0.1 & 0.001 & 0.002 & 0.002 & 0.002 & 0.002\\
\hspace{1em} & (0.163) & (0.211) & (0.049) & (0.049) & (0.047)\\
\hspace{1em}0.5 & 0.001 & 0.001 & 0.000 & 0.000 & 0.000\\
\hspace{1em} & (0.163) & (0.210) & (0.035) & (0.035) & (0.035)\\
\hspace{1em}0.9 & 0.000 & 0.001 & -0.002 & -0.002 & -0.002\\
\hspace{1em} & (0.163) & (0.211) & (0.049) & (0.048) & (0.046)\\
\addlinespace[.5em]
\hline
\multicolumn{6}{c}{(m, n) = (200, 200)}\\
\hspace{1em}0.1 & 0.000 & 0.000 & 0.002 & 0.002 & 0.002\\
\hspace{1em} & (0.058) & (0.073) & (0.017) & (0.017) & (0.016)\\
\hspace{1em}0.5 & 0.000 & 0.000 & 0.000 & 0.000 & 0.000\\
\hspace{1em} & (0.058) & (0.072) & (0.013) & (0.012) & (0.012)\\
\hspace{1em}0.9 & -0.001 & -0.001 & -0.002 & -0.002 & -0.002\\
\hspace{1em} & (0.058) & (0.073) & (0.017) & (0.017) & (0.017)\\
\bottomrule
\end{tabular}
\begin{tablenotes}
\item \textit{Note: } 
\item The table reports bias and standard deviation (in parentheses) of the simulations for $\beta(\tau)$ from 10,000 Monte Carlo simulations.
\end{tablenotes}
\end{threeparttable}
\end{table}

%% file: Tables/panel_coverage.tex
\begin{table}
\centering
\caption{\label{tab:panel_coverage}Coverage Probability of the 95\% Confidence Invervals}
\centering
\fontsize{9}{11}\selectfont
\begin{threeparttable}
\begin{tabular}[t]{llllll}
\toprule
Quantile & Pooled & BE & FE & RE-GMM & RE-OI\\
\midrule
\addlinespace[.5em]
\hline
\multicolumn{6}{c}{(m, n) = (25, 10)}\\
\hspace{1em}0.1 & 0.948 & 0.922 & 0.946 & 0.920 & 0.914\\
\hspace{1em}0.5 & 0.946 & 0.920 & 0.948 & 0.923 & 0.914\\
\hspace{1em}0.9 & 0.950 & 0.924 & 0.950 & 0.918 & 0.912\\
\addlinespace[.5em]
\hline
\multicolumn{6}{c}{(m, n) = (200, 10)}\\
\hspace{1em}0.1 & 0.942 & 0.941 & 0.927 & 0.930 & 0.874\\
\hspace{1em}0.5 & 0.947 & 0.943 & 0.952 & 0.948 & 0.943\\
\hspace{1em}0.9 & 0.946 & 0.942 & 0.932 & 0.934 & 0.877\\
\addlinespace[.5em]
\hline
\multicolumn{6}{c}{(m, n) = (25, 25)}\\
\hspace{1em}0.1 & 0.949 & 0.921 & 0.949 & 0.928 & 0.933\\
\hspace{1em}0.5 & 0.946 & 0.918 & 0.948 & 0.933 & 0.931\\
\hspace{1em}0.9 & 0.946 & 0.917 & 0.951 & 0.931 & 0.931\\
\addlinespace[.5em]
\hline
\multicolumn{6}{c}{(m, n) = (200, 25)}\\
\hspace{1em}0.1 & 0.947 & 0.945 & 0.942 & 0.939 & 0.928\\
\hspace{1em}0.5 & 0.950 & 0.945 & 0.954 & 0.948 & 0.948\\
\hspace{1em}0.9 & 0.948 & 0.946 & 0.938 & 0.938 & 0.927\\
\addlinespace[.5em]
\hline
\multicolumn{6}{c}{(m, n) = (25, 200)}\\
\hspace{1em}0.1 & 0.950 & 0.923 & 0.950 & 0.938 & 0.947\\
\hspace{1em}0.5 & 0.949 & 0.926 & 0.952 & 0.941 & 0.948\\
\hspace{1em}0.9 & 0.947 & 0.925 & 0.950 & 0.938 & 0.948\\
\addlinespace[.5em]
\hline
\multicolumn{6}{c}{(m, n) = (200, 200)}\\
\hspace{1em}0.1 & 0.948 & 0.943 & 0.947 & 0.946 & 0.950\\
\hspace{1em}0.5 & 0.947 & 0.943 & 0.951 & 0.951 & 0.950\\
\hspace{1em}0.9 & 0.948 & 0.944 & 0.949 & 0.947 & 0.949\\
\bottomrule
\end{tabular}
\begin{tablenotes}
\item \textit{Note: } 
\item Results based on 10,000 Monte Carlo simulations. The table reports the coverage probabilities of the confidence intervals of $\beta(\tau)$.
\end{tablenotes}
\end{threeparttable}
\end{table}

%% file: Tables/hausman_test.tex
\begin{table}[ht]

\caption{\label{tab:hausmantest}Hausman Test}
\centering
\fontsize{9}{11}\selectfont
\begin{threeparttable}
\begin{tabular}[t]{llllll}
\toprule
\multicolumn{1}{c}{ } & \multicolumn{5}{c}{$\lambda$ } \\
\cmidrule(l{3pt}r{3pt}){2-6}
Quantile & 0.0 & 0.1 & 0.2 & 0.3 & 0.4\\
\midrule
\addlinespace[.5em]
\hline
\multicolumn{6}{c}{(m,n) = (25, 10)}\\
\hspace{1em}0.1 & 0.052 & 0.058 & 0.076 & 0.116 & 0.178\\
\hspace{1em}0.5 & 0.048 & 0.060 & 0.095 & 0.162 & 0.268\\
\hspace{1em}0.9 & 0.050 & 0.066 & 0.094 & 0.144 & 0.218\\
\addlinespace[.5em]
\hline
\multicolumn{6}{c}{(m,n) = (200, 10)}\\
\hspace{1em}0.1 & 0.062 & 0.085 & 0.277 & 0.579 & 0.844\\
\hspace{1em}0.5 & 0.050 & 0.177 & 0.532 & 0.871 & 0.987\\
\hspace{1em}0.9 & 0.058 & 0.194 & 0.483 & 0.782 & 0.949\\
\addlinespace[.5em]
\hline
\multicolumn{6}{c}{(m,n) = (25, 25)}\\
\hspace{1em}0.1 & 0.046 & 0.058 & 0.098 & 0.173 & 0.281\\
\hspace{1em}0.5 & 0.043 & 0.060 & 0.114 & 0.204 & 0.340\\
\hspace{1em}0.9 & 0.045 & 0.062 & 0.110 & 0.186 & 0.306\\
\addlinespace[.5em]
\hline
\multicolumn{6}{c}{(m,n) = (200, 25)}\\
\hspace{1em}0.1 & 0.051 & 0.166 & 0.554 & 0.897 & 0.994\\
\hspace{1em}0.5 & 0.050 & 0.232 & 0.689 & 0.963 & 0.999\\
\hspace{1em}0.9 & 0.049 & 0.230 & 0.645 & 0.938 & 0.997\\
\addlinespace[.5em]
\hline
\multicolumn{6}{c}{(m,n) = (25, 200)}\\
\hspace{1em}0.1 & 0.043 & 0.062 & 0.123 & 0.230 & 0.399\\
\hspace{1em}0.5 & 0.043 & 0.063 & 0.125 & 0.242 & 0.412\\
\hspace{1em}0.9 & 0.044 & 0.063 & 0.124 & 0.235 & 0.403\\
\addlinespace[.5em]
\hline
\multicolumn{6}{c}{(m,n) = (200, 200)}\\
\hspace{1em}0.1 & 0.054 & 0.259 & 0.771 & 0.985 & 1.000\\
\hspace{1em}0.5 & 0.054 & 0.272 & 0.788 & 0.988 & 1.000\\
\hspace{1em}0.9 & 0.052 & 0.270 & 0.785 & 0.987 & 1.000\\
\bottomrule
\end{tabular}
\begin{tablenotes}
\item \textit{Note: } 
\item The table reports rejection rates of the Hausman test. The results are based on 10,000 Monte Carlo simulations. The first column shows the empirical size, while the other columns show the power of the test.
\end{tablenotes}
\end{threeparttable}
\end{table}

%% file: Figures/black.tex
\begin{tikzpicture}[x=1pt,y=1pt]
\definecolor{fillColor}{RGB}{255,255,255}
\path[use as bounding box,fill=fillColor,fill opacity=0.00] (0,0) rectangle (361.35,289.08);
\begin{scope}
\path[clip] ( 34.64, 30.69) rectangle (355.85,266.42);
\definecolor{drawColor}{gray}{0.92}

\path[draw=drawColor,line width= 0.3pt,line join=round] ( 34.64, 55.93) --
	(355.85, 55.93);

\path[draw=drawColor,line width= 0.3pt,line join=round] ( 34.64, 92.25) --
	(355.85, 92.25);

\path[draw=drawColor,line width= 0.3pt,line join=round] ( 34.64,128.58) --
	(355.85,128.58);

\path[draw=drawColor,line width= 0.3pt,line join=round] ( 34.64,164.90) --
	(355.85,164.90);

\path[draw=drawColor,line width= 0.3pt,line join=round] ( 34.64,201.22) --
	(355.85,201.22);

\path[draw=drawColor,line width= 0.3pt,line join=round] ( 34.64,237.55) --
	(355.85,237.55);

\path[draw=drawColor,line width= 0.3pt,line join=round] ( 73.58, 30.69) --
	( 73.58,266.42);

\path[draw=drawColor,line width= 0.3pt,line join=round] (154.69, 30.69) --
	(154.69,266.42);

\path[draw=drawColor,line width= 0.3pt,line join=round] (235.80, 30.69) --
	(235.80,266.42);

\path[draw=drawColor,line width= 0.3pt,line join=round] (316.92, 30.69) --
	(316.92,266.42);

\path[draw=drawColor,line width= 0.6pt,line join=round] ( 34.64, 37.77) --
	(355.85, 37.77);

\path[draw=drawColor,line width= 0.6pt,line join=round] ( 34.64, 74.09) --
	(355.85, 74.09);

\path[draw=drawColor,line width= 0.6pt,line join=round] ( 34.64,110.42) --
	(355.85,110.42);

\path[draw=drawColor,line width= 0.6pt,line join=round] ( 34.64,146.74) --
	(355.85,146.74);

\path[draw=drawColor,line width= 0.6pt,line join=round] ( 34.64,183.06) --
	(355.85,183.06);

\path[draw=drawColor,line width= 0.6pt,line join=round] ( 34.64,219.38) --
	(355.85,219.38);

\path[draw=drawColor,line width= 0.6pt,line join=round] ( 34.64,255.71) --
	(355.85,255.71);

\path[draw=drawColor,line width= 0.6pt,line join=round] (114.14, 30.69) --
	(114.14,266.42);

\path[draw=drawColor,line width= 0.6pt,line join=round] (195.25, 30.69) --
	(195.25,266.42);

\path[draw=drawColor,line width= 0.6pt,line join=round] (276.36, 30.69) --
	(276.36,266.42);
\definecolor{drawColor}{RGB}{17,51,58}

\path[draw=drawColor,line width= 0.6pt,line join=round] ( 49.25,180.11) --
	( 65.47,133.88) --
	( 81.69, 96.90) --
	( 97.91, 83.65) --
	(114.14, 81.98) --
	(130.36, 81.69) --
	(146.58, 86.06) --
	(162.80, 90.20) --
	(179.02, 85.57) --
	(195.25, 82.93) --
	(211.47, 82.65) --
	(227.69, 85.14) --
	(243.91, 80.31) --
	(260.14, 83.51) --
	(276.36, 85.92) --
	(292.58, 75.87) --
	(308.80, 76.83) --
	(325.03, 76.79) --
	(341.25, 71.33);
\definecolor{fillColor}{RGB}{17,51,58}

\path[draw=drawColor,line width= 0.4pt,line join=round,line cap=round,fill=fillColor] ( 49.25,180.11) circle (  1.67);

\path[draw=drawColor,line width= 0.4pt,line join=round,line cap=round,fill=fillColor] ( 65.47,133.88) circle (  1.67);

\path[draw=drawColor,line width= 0.4pt,line join=round,line cap=round,fill=fillColor] ( 81.69, 96.90) circle (  1.67);

\path[draw=drawColor,line width= 0.4pt,line join=round,line cap=round,fill=fillColor] ( 97.91, 83.65) circle (  1.67);

\path[draw=drawColor,line width= 0.4pt,line join=round,line cap=round,fill=fillColor] (114.14, 81.98) circle (  1.67);

\path[draw=drawColor,line width= 0.4pt,line join=round,line cap=round,fill=fillColor] (130.36, 81.69) circle (  1.67);

\path[draw=drawColor,line width= 0.4pt,line join=round,line cap=round,fill=fillColor] (146.58, 86.06) circle (  1.67);

\path[draw=drawColor,line width= 0.4pt,line join=round,line cap=round,fill=fillColor] (162.80, 90.20) circle (  1.67);

\path[draw=drawColor,line width= 0.4pt,line join=round,line cap=round,fill=fillColor] (179.02, 85.57) circle (  1.67);

\path[draw=drawColor,line width= 0.4pt,line join=round,line cap=round,fill=fillColor] (195.25, 82.93) circle (  1.67);

\path[draw=drawColor,line width= 0.4pt,line join=round,line cap=round,fill=fillColor] (211.47, 82.65) circle (  1.67);

\path[draw=drawColor,line width= 0.4pt,line join=round,line cap=round,fill=fillColor] (227.69, 85.14) circle (  1.67);

\path[draw=drawColor,line width= 0.4pt,line join=round,line cap=round,fill=fillColor] (243.91, 80.31) circle (  1.67);

\path[draw=drawColor,line width= 0.4pt,line join=round,line cap=round,fill=fillColor] (260.14, 83.51) circle (  1.67);

\path[draw=drawColor,line width= 0.4pt,line join=round,line cap=round,fill=fillColor] (276.36, 85.92) circle (  1.67);

\path[draw=drawColor,line width= 0.4pt,line join=round,line cap=round,fill=fillColor] (292.58, 75.87) circle (  1.67);

\path[draw=drawColor,line width= 0.4pt,line join=round,line cap=round,fill=fillColor] (308.80, 76.83) circle (  1.67);

\path[draw=drawColor,line width= 0.4pt,line join=round,line cap=round,fill=fillColor] (325.03, 76.79) circle (  1.67);

\path[draw=drawColor,line width= 0.4pt,line join=round,line cap=round,fill=fillColor] (341.25, 71.33) circle (  1.67);
\definecolor{fillColor}{RGB}{17,51,58}

\path[fill=fillColor,fill opacity=0.20] ( 49.25,252.82) --
	( 65.47,177.24) --
	( 81.69,128.47) --
	( 97.91,109.00) --
	(114.14,105.00) --
	(130.36,102.21) --
	(146.58,106.13) --
	(162.80,109.90) --
	(179.02,104.94) --
	(195.25,101.65) --
	(211.47,100.96) --
	(227.69,103.77) --
	(243.91, 99.10) --
	(260.14,101.74) --
	(276.36,104.61) --
	(292.58, 95.48) --
	(308.80, 96.81) --
	(325.03,100.01) --
	(341.25,100.93) --
	(341.25, 41.72) --
	(325.03, 53.58) --
	(308.80, 56.84) --
	(292.58, 56.27) --
	(276.36, 67.22) --
	(260.14, 65.27) --
	(243.91, 61.52) --
	(227.69, 66.52) --
	(211.47, 64.33) --
	(195.25, 64.20) --
	(179.02, 66.20) --
	(162.80, 70.50) --
	(146.58, 65.99) --
	(130.36, 61.17) --
	(114.14, 58.97) --
	( 97.91, 58.29) --
	( 81.69, 65.34) --
	( 65.47, 90.52) --
	( 49.25,107.41) --
	cycle;

\path[draw=drawColor,line width= 0.1pt,line join=round] ( 49.25,252.82) --
	( 65.47,177.24) --
	( 81.69,128.47) --
	( 97.91,109.00) --
	(114.14,105.00) --
	(130.36,102.21) --
	(146.58,106.13) --
	(162.80,109.90) --
	(179.02,104.94) --
	(195.25,101.65) --
	(211.47,100.96) --
	(227.69,103.77) --
	(243.91, 99.10) --
	(260.14,101.74) --
	(276.36,104.61) --
	(292.58, 95.48) --
	(308.80, 96.81) --
	(325.03,100.01) --
	(341.25,100.93);

\path[draw=drawColor,line width= 0.1pt,line join=round] (341.25, 41.72) --
	(325.03, 53.58) --
	(308.80, 56.84) --
	(292.58, 56.27) --
	(276.36, 67.22) --
	(260.14, 65.27) --
	(243.91, 61.52) --
	(227.69, 66.52) --
	(211.47, 64.33) --
	(195.25, 64.20) --
	(179.02, 66.20) --
	(162.80, 70.50) --
	(146.58, 65.99) --
	(130.36, 61.17) --
	(114.14, 58.97) --
	( 97.91, 58.29) --
	( 81.69, 65.34) --
	( 65.47, 90.52) --
	( 49.25,107.41);
\end{scope}
\begin{scope}
\path[clip] (  0.00,  0.00) rectangle (361.35,289.08);
\definecolor{drawColor}{gray}{0.30}

\node[text=drawColor,anchor=base east,inner sep=0pt, outer sep=0pt, scale=  0.88] at ( 29.69, 34.74) {-10};

\node[text=drawColor,anchor=base east,inner sep=0pt, outer sep=0pt, scale=  0.88] at ( 29.69, 71.06) {0};

\node[text=drawColor,anchor=base east,inner sep=0pt, outer sep=0pt, scale=  0.88] at ( 29.69,107.38) {10};

\node[text=drawColor,anchor=base east,inner sep=0pt, outer sep=0pt, scale=  0.88] at ( 29.69,143.71) {20};

\node[text=drawColor,anchor=base east,inner sep=0pt, outer sep=0pt, scale=  0.88] at ( 29.69,180.03) {30};

\node[text=drawColor,anchor=base east,inner sep=0pt, outer sep=0pt, scale=  0.88] at ( 29.69,216.35) {40};

\node[text=drawColor,anchor=base east,inner sep=0pt, outer sep=0pt, scale=  0.88] at ( 29.69,252.68) {50};
\end{scope}
\begin{scope}
\path[clip] (  0.00,  0.00) rectangle (361.35,289.08);
\definecolor{drawColor}{gray}{0.30}

\node[text=drawColor,anchor=base,inner sep=0pt, outer sep=0pt, scale=  0.88] at (114.14, 19.68) {0.25};

\node[text=drawColor,anchor=base,inner sep=0pt, outer sep=0pt, scale=  0.88] at (195.25, 19.68) {0.50};

\node[text=drawColor,anchor=base,inner sep=0pt, outer sep=0pt, scale=  0.88] at (276.36, 19.68) {0.75};
\end{scope}
\begin{scope}
\path[clip] (  0.00,  0.00) rectangle (361.35,289.08);
\definecolor{drawColor}{RGB}{0,0,0}

\node[text=drawColor,anchor=base,inner sep=0pt, outer sep=0pt, scale=  1.10] at (195.25,  7.64) {Quantiles};
\end{scope}
\begin{scope}
\path[clip] (  0.00,  0.00) rectangle (361.35,289.08);
\definecolor{drawColor}{RGB}{0,0,0}

\node[text=drawColor,rotate= 90.00,anchor=base,inner sep=0pt, outer sep=0pt, scale=  1.10] at ( 13.08,148.55) {Point Estimates};
\end{scope}
\begin{scope}
\path[clip] (  0.00,  0.00) rectangle (361.35,289.08);
\definecolor{drawColor}{RGB}{0,0,0}

\node[text=drawColor,anchor=base west,inner sep=0pt, outer sep=0pt, scale=  1.32] at ( 34.64,274.49) { };
\end{scope}
\end{tikzpicture}

%% file: Figures/white.tex
\begin{tikzpicture}[x=1pt,y=1pt]
\definecolor{fillColor}{RGB}{255,255,255}
\path[use as bounding box,fill=fillColor,fill opacity=0.00] (0,0) rectangle (361.35,289.08);
\begin{scope}
\path[clip] ( 34.64, 30.69) rectangle (355.85,266.42);
\definecolor{drawColor}{gray}{0.92}

\path[draw=drawColor,line width= 0.3pt,line join=round] ( 34.64, 55.93) --
	(355.85, 55.93);

\path[draw=drawColor,line width= 0.3pt,line join=round] ( 34.64, 92.25) --
	(355.85, 92.25);

\path[draw=drawColor,line width= 0.3pt,line join=round] ( 34.64,128.58) --
	(355.85,128.58);

\path[draw=drawColor,line width= 0.3pt,line join=round] ( 34.64,164.90) --
	(355.85,164.90);

\path[draw=drawColor,line width= 0.3pt,line join=round] ( 34.64,201.22) --
	(355.85,201.22);

\path[draw=drawColor,line width= 0.3pt,line join=round] ( 34.64,237.55) --
	(355.85,237.55);

\path[draw=drawColor,line width= 0.3pt,line join=round] ( 73.58, 30.69) --
	( 73.58,266.42);

\path[draw=drawColor,line width= 0.3pt,line join=round] (154.69, 30.69) --
	(154.69,266.42);

\path[draw=drawColor,line width= 0.3pt,line join=round] (235.80, 30.69) --
	(235.80,266.42);

\path[draw=drawColor,line width= 0.3pt,line join=round] (316.92, 30.69) --
	(316.92,266.42);

\path[draw=drawColor,line width= 0.6pt,line join=round] ( 34.64, 37.77) --
	(355.85, 37.77);

\path[draw=drawColor,line width= 0.6pt,line join=round] ( 34.64, 74.09) --
	(355.85, 74.09);

\path[draw=drawColor,line width= 0.6pt,line join=round] ( 34.64,110.42) --
	(355.85,110.42);

\path[draw=drawColor,line width= 0.6pt,line join=round] ( 34.64,146.74) --
	(355.85,146.74);

\path[draw=drawColor,line width= 0.6pt,line join=round] ( 34.64,183.06) --
	(355.85,183.06);

\path[draw=drawColor,line width= 0.6pt,line join=round] ( 34.64,219.38) --
	(355.85,219.38);

\path[draw=drawColor,line width= 0.6pt,line join=round] ( 34.64,255.71) --
	(355.85,255.71);

\path[draw=drawColor,line width= 0.6pt,line join=round] (114.14, 30.69) --
	(114.14,266.42);

\path[draw=drawColor,line width= 0.6pt,line join=round] (195.25, 30.69) --
	(195.25,266.42);

\path[draw=drawColor,line width= 0.6pt,line join=round] (276.36, 30.69) --
	(276.36,266.42);
\definecolor{drawColor}{RGB}{17,51,58}

\path[draw=drawColor,line width= 0.6pt,line join=round] ( 49.25,103.06) --
	( 65.47, 82.46) --
	( 81.69, 81.54) --
	( 97.91, 76.92) --
	(114.14, 77.73) --
	(130.36, 78.55) --
	(146.58, 79.34) --
	(162.80, 77.60) --
	(179.02, 76.06) --
	(195.25, 76.48) --
	(211.47, 77.05) --
	(227.69, 75.59) --
	(243.91, 74.89) --
	(260.14, 75.91) --
	(276.36, 75.81) --
	(292.58, 74.69) --
	(308.80, 74.99) --
	(325.03, 75.68) --
	(341.25, 76.82);
\definecolor{fillColor}{RGB}{17,51,58}

\path[draw=drawColor,line width= 0.4pt,line join=round,line cap=round,fill=fillColor] ( 49.25,103.06) circle (  1.67);

\path[draw=drawColor,line width= 0.4pt,line join=round,line cap=round,fill=fillColor] ( 65.47, 82.46) circle (  1.67);

\path[draw=drawColor,line width= 0.4pt,line join=round,line cap=round,fill=fillColor] ( 81.69, 81.54) circle (  1.67);

\path[draw=drawColor,line width= 0.4pt,line join=round,line cap=round,fill=fillColor] ( 97.91, 76.92) circle (  1.67);

\path[draw=drawColor,line width= 0.4pt,line join=round,line cap=round,fill=fillColor] (114.14, 77.73) circle (  1.67);

\path[draw=drawColor,line width= 0.4pt,line join=round,line cap=round,fill=fillColor] (130.36, 78.55) circle (  1.67);

\path[draw=drawColor,line width= 0.4pt,line join=round,line cap=round,fill=fillColor] (146.58, 79.34) circle (  1.67);

\path[draw=drawColor,line width= 0.4pt,line join=round,line cap=round,fill=fillColor] (162.80, 77.60) circle (  1.67);

\path[draw=drawColor,line width= 0.4pt,line join=round,line cap=round,fill=fillColor] (179.02, 76.06) circle (  1.67);

\path[draw=drawColor,line width= 0.4pt,line join=round,line cap=round,fill=fillColor] (195.25, 76.48) circle (  1.67);

\path[draw=drawColor,line width= 0.4pt,line join=round,line cap=round,fill=fillColor] (211.47, 77.05) circle (  1.67);

\path[draw=drawColor,line width= 0.4pt,line join=round,line cap=round,fill=fillColor] (227.69, 75.59) circle (  1.67);

\path[draw=drawColor,line width= 0.4pt,line join=round,line cap=round,fill=fillColor] (243.91, 74.89) circle (  1.67);

\path[draw=drawColor,line width= 0.4pt,line join=round,line cap=round,fill=fillColor] (260.14, 75.91) circle (  1.67);

\path[draw=drawColor,line width= 0.4pt,line join=round,line cap=round,fill=fillColor] (276.36, 75.81) circle (  1.67);

\path[draw=drawColor,line width= 0.4pt,line join=round,line cap=round,fill=fillColor] (292.58, 74.69) circle (  1.67);

\path[draw=drawColor,line width= 0.4pt,line join=round,line cap=round,fill=fillColor] (308.80, 74.99) circle (  1.67);

\path[draw=drawColor,line width= 0.4pt,line join=round,line cap=round,fill=fillColor] (325.03, 75.68) circle (  1.67);

\path[draw=drawColor,line width= 0.4pt,line join=round,line cap=round,fill=fillColor] (341.25, 76.82) circle (  1.67);
\definecolor{fillColor}{RGB}{17,51,58}

\path[fill=fillColor,fill opacity=0.20] ( 49.25,121.38) --
	( 65.47, 92.80) --
	( 81.69, 89.93) --
	( 97.91, 84.84) --
	(114.14, 85.20) --
	(130.36, 86.20) --
	(146.58, 87.14) --
	(162.80, 85.18) --
	(179.02, 83.55) --
	(195.25, 84.15) --
	(211.47, 84.71) --
	(227.69, 83.51) --
	(243.91, 82.68) --
	(260.14, 84.28) --
	(276.36, 84.18) --
	(292.58, 83.24) --
	(308.80, 84.25) --
	(325.03, 85.52) --
	(341.25, 87.97) --
	(341.25, 65.67) --
	(325.03, 65.84) --
	(308.80, 65.73) --
	(292.58, 66.14) --
	(276.36, 67.44) --
	(260.14, 67.55) --
	(243.91, 67.11) --
	(227.69, 67.68) --
	(211.47, 69.40) --
	(195.25, 68.81) --
	(179.02, 68.58) --
	(162.80, 70.03) --
	(146.58, 71.54) --
	(130.36, 70.90) --
	(114.14, 70.26) --
	( 97.91, 69.01) --
	( 81.69, 73.15) --
	( 65.47, 72.13) --
	( 49.25, 84.74) --
	cycle;

\path[draw=drawColor,line width= 0.1pt,line join=round] ( 49.25,121.38) --
	( 65.47, 92.80) --
	( 81.69, 89.93) --
	( 97.91, 84.84) --
	(114.14, 85.20) --
	(130.36, 86.20) --
	(146.58, 87.14) --
	(162.80, 85.18) --
	(179.02, 83.55) --
	(195.25, 84.15) --
	(211.47, 84.71) --
	(227.69, 83.51) --
	(243.91, 82.68) --
	(260.14, 84.28) --
	(276.36, 84.18) --
	(292.58, 83.24) --
	(308.80, 84.25) --
	(325.03, 85.52) --
	(341.25, 87.97);

\path[draw=drawColor,line width= 0.1pt,line join=round] (341.25, 65.67) --
	(325.03, 65.84) --
	(308.80, 65.73) --
	(292.58, 66.14) --
	(276.36, 67.44) --
	(260.14, 67.55) --
	(243.91, 67.11) --
	(227.69, 67.68) --
	(211.47, 69.40) --
	(195.25, 68.81) --
	(179.02, 68.58) --
	(162.80, 70.03) --
	(146.58, 71.54) --
	(130.36, 70.90) --
	(114.14, 70.26) --
	( 97.91, 69.01) --
	( 81.69, 73.15) --
	( 65.47, 72.13) --
	( 49.25, 84.74);
\end{scope}
\begin{scope}
\path[clip] (  0.00,  0.00) rectangle (361.35,289.08);
\definecolor{drawColor}{gray}{0.30}

\node[text=drawColor,anchor=base east,inner sep=0pt, outer sep=0pt, scale=  0.88] at ( 29.69, 34.74) {-10};

\node[text=drawColor,anchor=base east,inner sep=0pt, outer sep=0pt, scale=  0.88] at ( 29.69, 71.06) {0};

\node[text=drawColor,anchor=base east,inner sep=0pt, outer sep=0pt, scale=  0.88] at ( 29.69,107.38) {10};

\node[text=drawColor,anchor=base east,inner sep=0pt, outer sep=0pt, scale=  0.88] at ( 29.69,143.71) {20};

\node[text=drawColor,anchor=base east,inner sep=0pt, outer sep=0pt, scale=  0.88] at ( 29.69,180.03) {30};

\node[text=drawColor,anchor=base east,inner sep=0pt, outer sep=0pt, scale=  0.88] at ( 29.69,216.35) {40};

\node[text=drawColor,anchor=base east,inner sep=0pt, outer sep=0pt, scale=  0.88] at ( 29.69,252.68) {50};
\end{scope}
\begin{scope}
\path[clip] (  0.00,  0.00) rectangle (361.35,289.08);
\definecolor{drawColor}{gray}{0.30}

\node[text=drawColor,anchor=base,inner sep=0pt, outer sep=0pt, scale=  0.88] at (114.14, 19.68) {0.25};

\node[text=drawColor,anchor=base,inner sep=0pt, outer sep=0pt, scale=  0.88] at (195.25, 19.68) {0.50};

\node[text=drawColor,anchor=base,inner sep=0pt, outer sep=0pt, scale=  0.88] at (276.36, 19.68) {0.75};
\end{scope}
\begin{scope}
\path[clip] (  0.00,  0.00) rectangle (361.35,289.08);
\definecolor{drawColor}{RGB}{0,0,0}

\node[text=drawColor,anchor=base,inner sep=0pt, outer sep=0pt, scale=  1.10] at (195.25,  7.64) {Quantiles};
\end{scope}
\begin{scope}
\path[clip] (  0.00,  0.00) rectangle (361.35,289.08);
\definecolor{drawColor}{RGB}{0,0,0}

\node[text=drawColor,rotate= 90.00,anchor=base,inner sep=0pt, outer sep=0pt, scale=  1.10] at ( 13.08,148.55) {Point Estimates};
\end{scope}
\begin{scope}
\path[clip] (  0.00,  0.00) rectangle (361.35,289.08);
\definecolor{drawColor}{RGB}{0,0,0}

\node[text=drawColor,anchor=base west,inner sep=0pt, outer sep=0pt, scale=  1.32] at ( 34.64,274.49) { };
\end{scope}
\end{tikzpicture}

%% file: Figures/plot_comparison_clp_md.tex
\begin{tikzpicture}[x=1pt,y=1pt]
\definecolor{fillColor}{RGB}{255,255,255}
\path[use as bounding box,fill=fillColor,fill opacity=0.00] (0,0) rectangle (505.89,289.08);
\begin{scope}
\path[clip] ( 39.04, 67.14) rectangle (500.39,266.42);
\definecolor{drawColor}{gray}{0.92}

\path[draw=drawColor,line width= 0.3pt,line join=round] ( 39.04,111.16) --
	(500.39,111.16);

\path[draw=drawColor,line width= 0.3pt,line join=round] ( 39.04,174.73) --
	(500.39,174.73);

\path[draw=drawColor,line width= 0.3pt,line join=round] ( 39.04,238.29) --
	(500.39,238.29);

\path[draw=drawColor,line width= 0.3pt,line join=round] ( 94.96, 67.14) --
	( 94.96,266.42);

\path[draw=drawColor,line width= 0.3pt,line join=round] (211.47, 67.14) --
	(211.47,266.42);

\path[draw=drawColor,line width= 0.3pt,line join=round] (327.97, 67.14) --
	(327.97,266.42);

\path[draw=drawColor,line width= 0.3pt,line join=round] (444.47, 67.14) --
	(444.47,266.42);

\path[draw=drawColor,line width= 0.6pt,line join=round] ( 39.04, 79.38) --
	(500.39, 79.38);

\path[draw=drawColor,line width= 0.6pt,line join=round] ( 39.04,142.94) --
	(500.39,142.94);

\path[draw=drawColor,line width= 0.6pt,line join=round] ( 39.04,206.51) --
	(500.39,206.51);

\path[draw=drawColor,line width= 0.6pt,line join=round] (153.22, 67.14) --
	(153.22,266.42);

\path[draw=drawColor,line width= 0.6pt,line join=round] (269.72, 67.14) --
	(269.72,266.42);

\path[draw=drawColor,line width= 0.6pt,line join=round] (386.22, 67.14) --
	(386.22,266.42);
\definecolor{drawColor}{RGB}{15,174,182}

\path[draw=drawColor,line width= 0.6pt,line join=round] ( 60.01,166.60) --
	( 83.31,147.55) --
	(106.61,114.39) --
	(129.92,121.16) --
	(153.22,133.47) --
	(176.52,137.13) --
	(199.82,129.51) --
	(223.12,128.68) --
	(246.42,128.43) --
	(269.72,137.59) --
	(293.02,132.59) --
	(316.32,140.42) --
	(339.62,144.80) --
	(362.92,146.50) --
	(386.22,138.99) --
	(409.52,139.57) --
	(432.82,148.33) --
	(456.12,156.57) --
	(479.42,161.72);
\definecolor{drawColor}{RGB}{17,51,58}

\path[draw=drawColor,line width= 0.6pt,line join=round] ( 60.01,152.22) --
	( 83.31,148.18) --
	(106.61,144.94) --
	(129.92,143.78) --
	(153.22,143.63) --
	(176.52,143.61) --
	(199.82,143.99) --
	(223.12,144.35) --
	(246.42,143.95) --
	(269.72,143.72) --
	(293.02,143.69) --
	(316.32,143.91) --
	(339.62,143.49) --
	(362.92,143.77) --
	(386.22,143.98) --
	(409.52,143.10) --
	(432.82,143.18) --
	(456.12,143.18) --
	(479.42,142.70);
\definecolor{drawColor}{RGB}{15,174,182}
\definecolor{fillColor}{RGB}{15,174,182}

\path[draw=drawColor,line width= 0.4pt,line join=round,line cap=round,fill=fillColor] ( 60.01,166.60) circle (  1.67);

\path[draw=drawColor,line width= 0.4pt,line join=round,line cap=round,fill=fillColor] ( 83.31,147.55) circle (  1.67);

\path[draw=drawColor,line width= 0.4pt,line join=round,line cap=round,fill=fillColor] (106.61,114.39) circle (  1.67);

\path[draw=drawColor,line width= 0.4pt,line join=round,line cap=round,fill=fillColor] (129.92,121.16) circle (  1.67);

\path[draw=drawColor,line width= 0.4pt,line join=round,line cap=round,fill=fillColor] (153.22,133.47) circle (  1.67);

\path[draw=drawColor,line width= 0.4pt,line join=round,line cap=round,fill=fillColor] (176.52,137.13) circle (  1.67);

\path[draw=drawColor,line width= 0.4pt,line join=round,line cap=round,fill=fillColor] (199.82,129.51) circle (  1.67);

\path[draw=drawColor,line width= 0.4pt,line join=round,line cap=round,fill=fillColor] (223.12,128.68) circle (  1.67);

\path[draw=drawColor,line width= 0.4pt,line join=round,line cap=round,fill=fillColor] (246.42,128.43) circle (  1.67);

\path[draw=drawColor,line width= 0.4pt,line join=round,line cap=round,fill=fillColor] (269.72,137.59) circle (  1.67);

\path[draw=drawColor,line width= 0.4pt,line join=round,line cap=round,fill=fillColor] (293.02,132.59) circle (  1.67);

\path[draw=drawColor,line width= 0.4pt,line join=round,line cap=round,fill=fillColor] (316.32,140.42) circle (  1.67);

\path[draw=drawColor,line width= 0.4pt,line join=round,line cap=round,fill=fillColor] (339.62,144.80) circle (  1.67);

\path[draw=drawColor,line width= 0.4pt,line join=round,line cap=round,fill=fillColor] (362.92,146.50) circle (  1.67);

\path[draw=drawColor,line width= 0.4pt,line join=round,line cap=round,fill=fillColor] (386.22,138.99) circle (  1.67);

\path[draw=drawColor,line width= 0.4pt,line join=round,line cap=round,fill=fillColor] (409.52,139.57) circle (  1.67);

\path[draw=drawColor,line width= 0.4pt,line join=round,line cap=round,fill=fillColor] (432.82,148.33) circle (  1.67);

\path[draw=drawColor,line width= 0.4pt,line join=round,line cap=round,fill=fillColor] (456.12,156.57) circle (  1.67);

\path[draw=drawColor,line width= 0.4pt,line join=round,line cap=round,fill=fillColor] (479.42,161.72) circle (  1.67);
\definecolor{drawColor}{RGB}{17,51,58}
\definecolor{fillColor}{RGB}{17,51,58}

\path[draw=drawColor,line width= 0.4pt,line join=round,line cap=round,fill=fillColor] ( 60.01,152.22) circle (  1.67);

\path[draw=drawColor,line width= 0.4pt,line join=round,line cap=round,fill=fillColor] ( 83.31,148.18) circle (  1.67);

\path[draw=drawColor,line width= 0.4pt,line join=round,line cap=round,fill=fillColor] (106.61,144.94) circle (  1.67);

\path[draw=drawColor,line width= 0.4pt,line join=round,line cap=round,fill=fillColor] (129.92,143.78) circle (  1.67);

\path[draw=drawColor,line width= 0.4pt,line join=round,line cap=round,fill=fillColor] (153.22,143.63) circle (  1.67);

\path[draw=drawColor,line width= 0.4pt,line join=round,line cap=round,fill=fillColor] (176.52,143.61) circle (  1.67);

\path[draw=drawColor,line width= 0.4pt,line join=round,line cap=round,fill=fillColor] (199.82,143.99) circle (  1.67);

\path[draw=drawColor,line width= 0.4pt,line join=round,line cap=round,fill=fillColor] (223.12,144.35) circle (  1.67);

\path[draw=drawColor,line width= 0.4pt,line join=round,line cap=round,fill=fillColor] (246.42,143.95) circle (  1.67);

\path[draw=drawColor,line width= 0.4pt,line join=round,line cap=round,fill=fillColor] (269.72,143.72) circle (  1.67);

\path[draw=drawColor,line width= 0.4pt,line join=round,line cap=round,fill=fillColor] (293.02,143.69) circle (  1.67);

\path[draw=drawColor,line width= 0.4pt,line join=round,line cap=round,fill=fillColor] (316.32,143.91) circle (  1.67);

\path[draw=drawColor,line width= 0.4pt,line join=round,line cap=round,fill=fillColor] (339.62,143.49) circle (  1.67);

\path[draw=drawColor,line width= 0.4pt,line join=round,line cap=round,fill=fillColor] (362.92,143.77) circle (  1.67);

\path[draw=drawColor,line width= 0.4pt,line join=round,line cap=round,fill=fillColor] (386.22,143.98) circle (  1.67);

\path[draw=drawColor,line width= 0.4pt,line join=round,line cap=round,fill=fillColor] (409.52,143.10) circle (  1.67);

\path[draw=drawColor,line width= 0.4pt,line join=round,line cap=round,fill=fillColor] (432.82,143.18) circle (  1.67);

\path[draw=drawColor,line width= 0.4pt,line join=round,line cap=round,fill=fillColor] (456.12,143.18) circle (  1.67);

\path[draw=drawColor,line width= 0.4pt,line join=round,line cap=round,fill=fillColor] (479.42,142.70) circle (  1.67);
\definecolor{fillColor}{RGB}{15,174,182}

\path[fill=fillColor,fill opacity=0.20] ( 60.01,254.39) --
	( 83.31,202.13) --
	(106.61,152.37) --
	(129.92,154.42) --
	(153.22,161.34) --
	(176.52,161.45) --
	(199.82,154.00) --
	(223.12,152.05) --
	(246.42,150.97) --
	(269.72,159.75) --
	(293.02,155.53) --
	(316.32,162.99) --
	(339.62,167.13) --
	(362.92,168.79) --
	(386.22,162.35) --
	(409.52,165.39) --
	(432.82,177.24) --
	(456.12,188.62) --
	(479.42,202.24) --
	(479.42,121.21) --
	(456.12,124.51) --
	(432.82,119.41) --
	(409.52,113.75) --
	(386.22,115.63) --
	(362.92,124.21) --
	(339.62,122.47) --
	(316.32,117.86) --
	(293.02,109.64) --
	(269.72,115.43) --
	(246.42,105.90) --
	(223.12,105.31) --
	(199.82,105.03) --
	(176.52,112.82) --
	(153.22,105.60) --
	(129.92, 87.90) --
	(106.61, 76.40) --
	( 83.31, 92.96) --
	( 60.01, 78.81) --
	cycle;
\definecolor{drawColor}{RGB}{15,174,182}

\path[draw=drawColor,line width= 0.1pt,line join=round] ( 60.01,254.39) --
	( 83.31,202.13) --
	(106.61,152.37) --
	(129.92,154.42) --
	(153.22,161.34) --
	(176.52,161.45) --
	(199.82,154.00) --
	(223.12,152.05) --
	(246.42,150.97) --
	(269.72,159.75) --
	(293.02,155.53) --
	(316.32,162.99) --
	(339.62,167.13) --
	(362.92,168.79) --
	(386.22,162.35) --
	(409.52,165.39) --
	(432.82,177.24) --
	(456.12,188.62) --
	(479.42,202.24);

\path[draw=drawColor,line width= 0.1pt,line join=round] (479.42,121.21) --
	(456.12,124.51) --
	(432.82,119.41) --
	(409.52,113.75) --
	(386.22,115.63) --
	(362.92,124.21) --
	(339.62,122.47) --
	(316.32,117.86) --
	(293.02,109.64) --
	(269.72,115.43) --
	(246.42,105.90) --
	(223.12,105.31) --
	(199.82,105.03) --
	(176.52,112.82) --
	(153.22,105.60) --
	(129.92, 87.90) --
	(106.61, 76.40) --
	( 83.31, 92.96) --
	( 60.01, 78.81);
\definecolor{fillColor}{RGB}{17,51,58}

\path[fill=fillColor,fill opacity=0.20] ( 60.01,158.58) --
	( 83.31,151.97) --
	(106.61,147.70) --
	(129.92,146.00) --
	(153.22,145.65) --
	(176.52,145.40) --
	(199.82,145.75) --
	(223.12,146.08) --
	(246.42,145.64) --
	(269.72,145.36) --
	(293.02,145.30) --
	(316.32,145.54) --
	(339.62,145.13) --
	(362.92,145.36) --
	(386.22,145.61) --
	(409.52,144.81) --
	(432.82,144.93) --
	(456.12,145.21) --
	(479.42,145.29) --
	(479.42,140.11) --
	(456.12,141.15) --
	(432.82,141.43) --
	(409.52,141.38) --
	(386.22,142.34) --
	(362.92,142.17) --
	(339.62,141.84) --
	(316.32,142.28) --
	(293.02,142.09) --
	(269.72,142.08) --
	(246.42,142.25) --
	(223.12,142.63) --
	(199.82,142.23) --
	(176.52,141.81) --
	(153.22,141.62) --
	(129.92,141.56) --
	(106.61,142.18) --
	( 83.31,144.38) --
	( 60.01,145.86) --
	cycle;
\definecolor{drawColor}{RGB}{17,51,58}

\path[draw=drawColor,line width= 0.1pt,line join=round] ( 60.01,158.58) --
	( 83.31,151.97) --
	(106.61,147.70) --
	(129.92,146.00) --
	(153.22,145.65) --
	(176.52,145.40) --
	(199.82,145.75) --
	(223.12,146.08) --
	(246.42,145.64) --
	(269.72,145.36) --
	(293.02,145.30) --
	(316.32,145.54) --
	(339.62,145.13) --
	(362.92,145.36) --
	(386.22,145.61) --
	(409.52,144.81) --
	(432.82,144.93) --
	(456.12,145.21) --
	(479.42,145.29);

\path[draw=drawColor,line width= 0.1pt,line join=round] (479.42,140.11) --
	(456.12,141.15) --
	(432.82,141.43) --
	(409.52,141.38) --
	(386.22,142.34) --
	(362.92,142.17) --
	(339.62,141.84) --
	(316.32,142.28) --
	(293.02,142.09) --
	(269.72,142.08) --
	(246.42,142.25) --
	(223.12,142.63) --
	(199.82,142.23) --
	(176.52,141.81) --
	(153.22,141.62) --
	(129.92,141.56) --
	(106.61,142.18) --
	( 83.31,144.38) --
	( 60.01,145.86);
\end{scope}
\begin{scope}
\path[clip] (  0.00,  0.00) rectangle (505.89,289.08);
\definecolor{drawColor}{gray}{0.30}

\node[text=drawColor,anchor=base east,inner sep=0pt, outer sep=0pt, scale=  0.88] at ( 34.09, 76.35) {-200};

\node[text=drawColor,anchor=base east,inner sep=0pt, outer sep=0pt, scale=  0.88] at ( 34.09,139.91) {0};

\node[text=drawColor,anchor=base east,inner sep=0pt, outer sep=0pt, scale=  0.88] at ( 34.09,203.48) {200};
\end{scope}
\begin{scope}
\path[clip] (  0.00,  0.00) rectangle (505.89,289.08);
\definecolor{drawColor}{gray}{0.30}

\node[text=drawColor,anchor=base,inner sep=0pt, outer sep=0pt, scale=  0.88] at (153.22, 56.13) {0.25};

\node[text=drawColor,anchor=base,inner sep=0pt, outer sep=0pt, scale=  0.88] at (269.72, 56.13) {0.50};

\node[text=drawColor,anchor=base,inner sep=0pt, outer sep=0pt, scale=  0.88] at (386.22, 56.13) {0.75};
\end{scope}
\begin{scope}
\path[clip] (  0.00,  0.00) rectangle (505.89,289.08);
\definecolor{drawColor}{RGB}{0,0,0}

\node[text=drawColor,anchor=base,inner sep=0pt, outer sep=0pt, scale=  1.10] at (269.72, 44.09) {Quantiles};
\end{scope}
\begin{scope}
\path[clip] (  0.00,  0.00) rectangle (505.89,289.08);
\definecolor{drawColor}{RGB}{0,0,0}

\node[text=drawColor,rotate= 90.00,anchor=base,inner sep=0pt, outer sep=0pt, scale=  1.10] at ( 13.08,166.78) {Point Estimates};
\end{scope}
\begin{scope}
\path[clip] (  0.00,  0.00) rectangle (505.89,289.08);
\definecolor{drawColor}{RGB}{0,0,0}

\node[text=drawColor,anchor=base west,inner sep=0pt, outer sep=0pt, scale=  1.10] at (212.98, 14.44) {Model};
\end{scope}
\begin{scope}
\path[clip] (  0.00,  0.00) rectangle (505.89,289.08);
\definecolor{drawColor}{RGB}{15,174,182}

\path[draw=drawColor,line width= 0.6pt,line join=round] (249.86, 18.23) -- (261.43, 18.23);
\end{scope}
\begin{scope}
\path[clip] (  0.00,  0.00) rectangle (505.89,289.08);
\definecolor{drawColor}{RGB}{15,174,182}
\definecolor{fillColor}{RGB}{15,174,182}

\path[draw=drawColor,line width= 0.4pt,line join=round,line cap=round,fill=fillColor] (255.65, 18.23) circle (  1.67);
\end{scope}
\begin{scope}
\path[clip] (  0.00,  0.00) rectangle (505.89,289.08);
\definecolor{drawColor}{RGB}{15,174,182}
\definecolor{fillColor}{RGB}{15,174,182}

\path[draw=drawColor,line width= 0.1pt,fill=fillColor,fill opacity=0.20] (248.56, 11.14) rectangle (262.73, 25.31);
\end{scope}
\begin{scope}
\path[clip] (  0.00,  0.00) rectangle (505.89,289.08);
\definecolor{drawColor}{RGB}{17,51,58}

\path[draw=drawColor,line width= 0.6pt,line join=round] (293.16, 18.23) -- (304.72, 18.23);
\end{scope}
\begin{scope}
\path[clip] (  0.00,  0.00) rectangle (505.89,289.08);
\definecolor{drawColor}{RGB}{17,51,58}
\definecolor{fillColor}{RGB}{17,51,58}

\path[draw=drawColor,line width= 0.4pt,line join=round,line cap=round,fill=fillColor] (298.94, 18.23) circle (  1.67);
\end{scope}
\begin{scope}
\path[clip] (  0.00,  0.00) rectangle (505.89,289.08);
\definecolor{drawColor}{RGB}{17,51,58}
\definecolor{fillColor}{RGB}{17,51,58}

\path[draw=drawColor,line width= 0.1pt,fill=fillColor,fill opacity=0.20] (291.86, 11.14) rectangle (306.02, 25.31);
\end{scope}
\begin{scope}
\path[clip] (  0.00,  0.00) rectangle (505.89,289.08);
\definecolor{drawColor}{RGB}{0,0,0}

\node[text=drawColor,anchor=base west,inner sep=0pt, outer sep=0pt, scale=  0.88] at (268.37, 15.20) {CLP};
\end{scope}
\begin{scope}
\path[clip] (  0.00,  0.00) rectangle (505.89,289.08);
\definecolor{drawColor}{RGB}{0,0,0}

\node[text=drawColor,anchor=base west,inner sep=0pt, outer sep=0pt, scale=  0.88] at (311.67, 15.20) {MD};
\end{scope}
\begin{scope}
\path[clip] (  0.00,  0.00) rectangle (505.89,289.08);
\definecolor{drawColor}{RGB}{0,0,0}

\node[text=drawColor,anchor=base west,inner sep=0pt, outer sep=0pt, scale=  1.32] at ( 39.04,274.49) { };
\end{scope}
\end{tikzpicture}